%% file: manuscript.tex
\documentclass[12pt]{article}
\usepackage[letterpaper,margin=1in]{geometry}
\usepackage{setspace}
\usepackage{amssymb}
\usepackage{bbm}
\usepackage{booktabs}
\usepackage{caption}
\usepackage{graphicx}
\usepackage{amsmath}
\usepackage{float}
\usepackage{geometry}   %page block size
\usepackage{fancyhdr}   %frontpage
\usepackage{listings} 
\usepackage{xcolor}

\definecolor{revisionblue}{RGB}{0,72,160}
\newcommand{\rev}[1]{{\color{black}#1}}
\newenvironment{revision}{\begingroup\color{black}}{\endgroup}

\usepackage{amsthm}
\usepackage{mathrsfs}
\usepackage{algpseudocode}
\usepackage{algorithm}
\makeatletter
\renewcommand{\theALG@line}{\thealgorithm.\arabic{ALG@line}}
\providecommand{\theHALG@line}{}
\renewcommand{\theHALG@line}{\thealgorithm.\arabic{ALG@line}}
\makeatother
\usepackage[colorlinks,
linkcolor=blue,
anchorcolor=blue,
citecolor=blue]{hyperref}
\usepackage{comment}
\usepackage[round]{natbib}
\usepackage{enumitem}
\usepackage{bbm}
\usepackage{cleveref}
\usepackage{subcaption}

\newtheorem{theorem}{Theorem}
\newtheorem{lemma}{Lemma}
\newtheorem{definition}{Definition}
\newtheorem{proposition}{Proposition}

\newtheorem{assumption}{Assumption}
\newtheorem{remark}{Remark}

\crefname{assumption}{assumption}{assumptions}
\Crefname{assumption}{Assumption}{Assumptions}

\usepackage{titlesec}

\titlespacing*{\section}{0pt}{1.0ex plus 0.2ex minus 0.1ex}{0.5ex}
\titlespacing*{\subsection}{0pt}{0.8ex plus 0.2ex minus 0.1ex}{0.4ex}
\titlespacing*{\subsubsection}{0pt}{0.6ex plus 0.1ex minus 0.1ex}{0.3ex}

\algrenewcommand\algorithmicrequire{\textbf{Input:}}
\algrenewcommand\algorithmicensure{\textbf{Output:}}

\def\calC{{\mathcal C}}
\def\calD{{\mathcal D}}
\def\calE{{\mathcal E}}

\def\calM{{\mathcal M}}
\def\calN{{\mathcal N}}

\def\calP{{\mathcal P}}
\def\calQ{{\mathcal Q}}

\def\calS{{\mathcal S}}
\def\calT{{\mathcal T}}

\def\calW{{\mathcal W}}

\def\bcalA{{\boldsymbol{\mathcal A}}}

\def\bcalC{{\boldsymbol{\mathcal C}}}

\def\bcalM{{\boldsymbol{\mathcal M}}}

\def\bcalS{{\boldsymbol{\mathcal S}}}

\def\bcalX{{\boldsymbol{\mathcal X}}}

\def\EE{{\mathbb E}}

\def\OO{{\mathbb O}}
\def\PP{{\mathbb P}}

\def\RR{{\mathbb R}}

\def\c{{\boldsymbol c}}
\def\d{{\boldsymbol d}}
\def\e{{\boldsymbol e}}

\def\h{{\boldsymbol h}}

\def\u{{\boldsymbol u}}

\def\x{{\boldsymbol x}}
\def\y{{\boldsymbol y}}

\def\A{\mathbf{A}}
\def\B{\mathbf{B}}
\def\C{\mathbf{C}}
\def\D{\mathbf{D}}
\def\E{\mathbf{E}}

\def\H{\mathbf{H}}
\def\I{\mathbf{I}}

\def\M{\mathbf{M}}

\def\O{\mathbf{O}}
\def\P{\mathbf{P}}

\def\R{\mathbf{R}}
\def\S{\mathbf{S}}

\def\U{\mathbf{U}}
\def\V{\mathbf{V}}
\def\W{\mathbf{W}}
\def\X{\mathbf{X}}
\def\Y{\mathbf{Y}}
\def\Z{\mathbf{Z}}

\def\bY{{\boldsymbol Y}}
\def\bP{{\boldsymbol P}}

\def\bSigma{{\boldsymbol \Sigma}}
\def\bTheta{{\boldsymbol \Theta}}
\def\bDelta{{\boldsymbol \Delta}}

\def\bPi{{\boldsymbol \Pi}}

\def\btheta{{\boldsymbol \theta}}
\def\bome{\boldsymbol{\omega}}

\def\bpi{{\boldsymbol \pi}}

\def\rank{\textsf{rank}}

\def\tr{\textsf{Tr}}

\def\vec{\textsf{vec}}
\def\unvec{\textsf{mat}}

\def\hat{\widehat}
\def\tilde{\widetilde}

\newcommand\fro[1]{\left\| #1 \right\|_{\rm{F}}}
\newcommand\bfro[1]{\big\| #1 \big\|_{\rm{F}}}
\newcommand\op[1]{\left\| #1 \right\|}
\newcommand\bop[1]{\big\| #1 \big\|}

\newcommand{\inp}[2]{\left\langle #1,#2\right\rangle}

\newcommand\brac[1]{\left(#1\right)}

\newcommand\ebrac[1]{\left\{#1\right\}}
\newcommand\sqbrac[1]{\left[#1\right]}
\newcommand\ab[1]{\left|#1\right|}

\def\spacingset#1{\renewcommand{\baselinestretch}%
{#1}\small\normalsize}

\begin{document}
\spacingset{1}
\title{Mixed Membership Model of Low-rank Matrices with Multimodal Extension}
\author{David Snider$^1$, Zhongyuan Lyu$^2$, Jian Kang$^3$, Yuqi Gu$^1$}
\date{$^1$Department of Statistics, Columbia University\\
$^2$The University of Sydney Business School\\
$^3$Department of Biostatistics, University of Michigan
}
\maketitle
\begin{abstract}
\begin{revision}
Matrix-valued observations arise in multiplex networks, neuroimaging, and other domains where population-level patterns are often low-rank and subjects may express several latent patterns simultaneously. Existing tensor PCA methods provide continuous subject scores but their loading matrices can be difficult to interpret as population prototypes, while low-rank clustering yields interpretable prototypes with hard labels. We introduce a low-rank mixed membership model for matrix-valued data in which the expected value of each subject’s matrix is a convex combination of latent low-rank basis matrices. The model yields both interpretable population-level extreme profiles and continuous subject-level memberships. Our multimodal extension shares memberships across modalities with modality-specific basis matrices and can restore identifiability when one modality is insufficient. We establish identifiability under a pure-subject condition, propose a constrained least-squares estimator and scalable algorithm with spectral initialization and low-rank refinement, and derive nonasymptotic error bounds. The estimator achieves a minimax-optimal reconstruction rate up to a logarithmic factor, with separate basis and membership convergence rates under a geometric condition. Simulations corroborate the theoretical rates and show strong performance. In an analysis of Human Connectome Project functional connectivity data, the proposed method identifies interpretable brain connectivity profiles whose estimated memberships are strongly associated with cognitive phenotypes.
\end{revision}

\end{abstract}


\spacingset{1.7}

\section{Introduction}

\begin{revision}
Modern research in the biological and social sciences increasingly generates datasets of matrix-valued observations. Examples include functional and structural brain connectivity matrices \citep{Zhang2019TensorNetwork, Liu2023joint}, commodity trade networks between countries \citep{Lyu2023LatentSpaceHigherOrder, Cai2023GeneralizedLowRankSparseTensor}, genetic and bacterial covariance matrices \citep{Stanley2016StrataMLSBM, Larremore2013, Jing2021CommunityDetectionMixtureMultilayer}, and multi-layer social networks \citep{Dong2012ClusteringMultiLayer}. In these applications, the matrix structure itself contains scientifically meaningful information that is often lost by vectorization. Accordingly, the works cited above have developed statistical methods that directly model matrix- and network-valued data while exploiting their inherent structural characteristics, such as spatial dependence, network organization, and low-dimensional representations. In particular, such methods have yielded substantial improvements in interpretability, statistical efficiency, and scalability, especially for \emph{neuroimaging studies}~\citep{Kang2018SOIR, Li2019SAVC, WuGuoKang24, Li2026BSNMani, Kang2026NeuroOpportunities}.

One common objective in the analysis of matrix-valued observations is to recover matrix-valued structures that characterize prototypical patterns in the population, and which can be interpreted scientifically and related to auxiliary covariates. A motivating application is the \emph{Human Connectome Project} (HCP)~\citep{VanEssen13}, which provides functional brain connectivity matrices collected from hundreds of healthy participants performing multiple cognitive tasks, together with rich behavioral and cognitive phenotypes. In analyzing this dataset, a central scientific question is whether the subjects can be characterized by a small number of latent connectivity networks that are both biologically interpretable and predictive of cognitive function. 

Two statistical phenomena are central to the pursuit of this objective. First, latent matrix-valued patterns are often \emph{approximately low-rank}. For multiplex network data, prototypical patterns often have the structure of a stochastic blockmodel or latent space model \citep{Holland1983, Hoff2002LatentSpace, Jing2021CommunityDetectionMixtureMultilayer, Arroyo2021COSIE}; for neuroimaging data, shared connectivity behavior among neurons belonging to the same functional brain systems can similarly lead to structured low-dimensional patterns \citep{Power11,Li2026BSNMani}. Second, individual observations often \emph{vary continuously between latent extremes}. Political beliefs, genetic ancestry, disease presentation, and cognitive phenotypes are rarely well represented by mutually exclusive classes, and they more naturally lie on a continuum between extreme latent patterns \citep{CHG24, Novembre2008Genes}. A useful method for matrix-valued data should therefore exploit low-rank structure while retaining a continuous and interpretable subject embedding.

Existing methods tend to provide only one side of this combination. Tensor PCA and related semi-symmetric tensor factorization methods produce continuous subject scores that are significantly related to phenotypic traits \citep{Zhang2019TensorNetwork, Liu2023joint, Weylandt2025SS-TPCA}. However, their factor loading matrices are reconstruction directions: later components may describe residual variation rather than coherent population prototypes, making direct scientific interpretation difficult. Low-rank clustering methods take the opposite perspective. The Low-rank Lloyd algorithm of \cite{LX25} estimates interpretable low-rank cluster centers and has sharp clustering theory, but it assigns each observation to one class. This hard-label structure is restrictive when subjects express latent patterns in varying proportions. Figure~\ref{fig:model-illustration} illustrates this tradeoff.
\end{revision}

\begin{figure}[htbp]
    \centering
    \begin{subfigure}[t]{0.29\linewidth}
        \centering
        \includegraphics[width=0.9\linewidth]{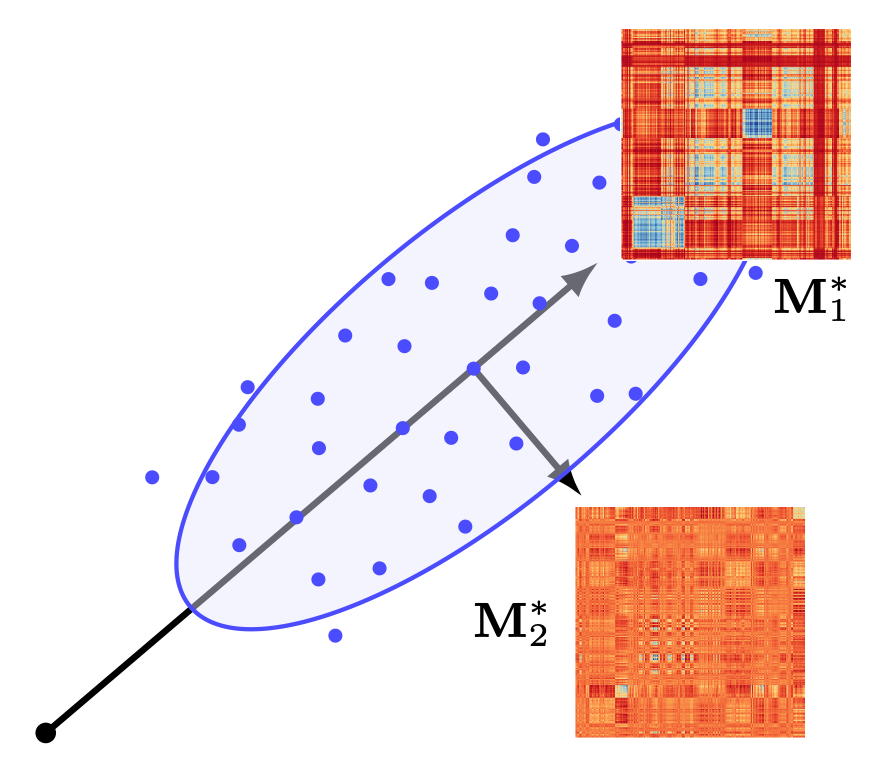}
        \subcaption{Tensor PCA}
    \end{subfigure}
    \qquad
    \begin{subfigure}[t]{0.29\linewidth}
        \centering
        \includegraphics[width=0.9\linewidth]{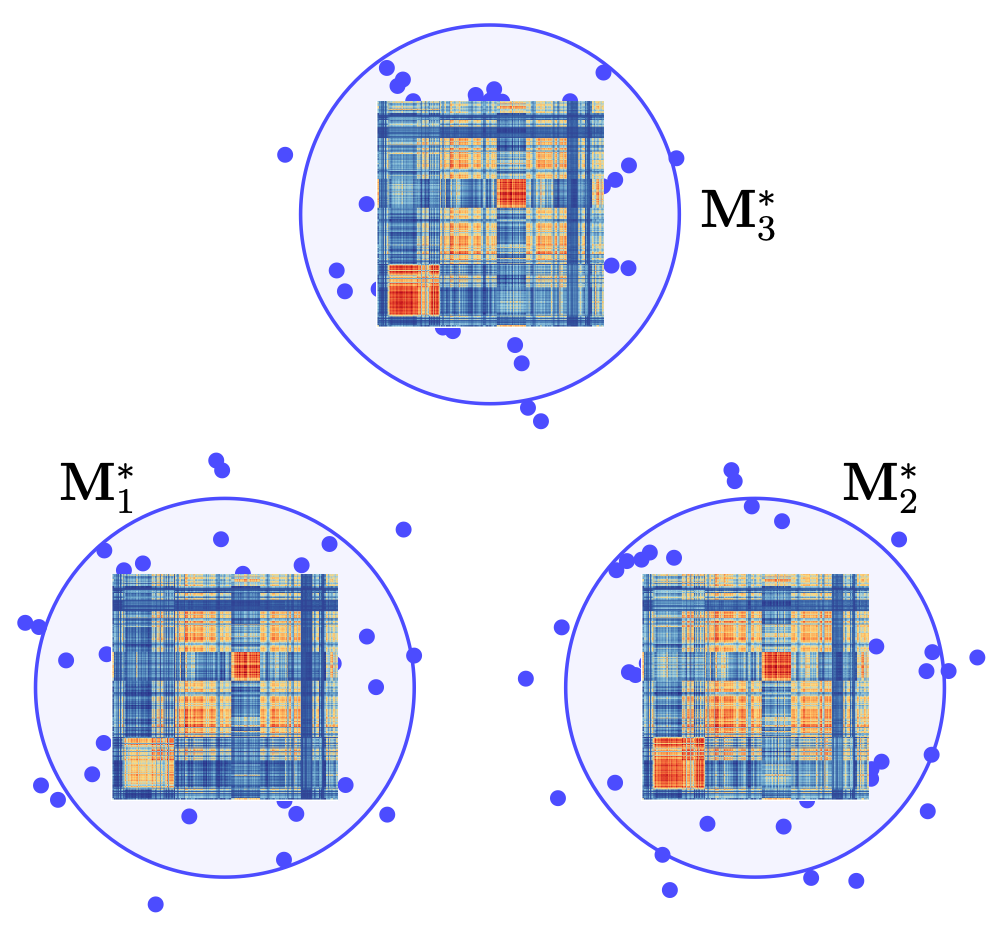}
        \subcaption{Low-rank clustering}
    \end{subfigure}
    \qquad
    \begin{subfigure}[t]{0.29\linewidth}
        \centering
        \includegraphics[width=\linewidth]{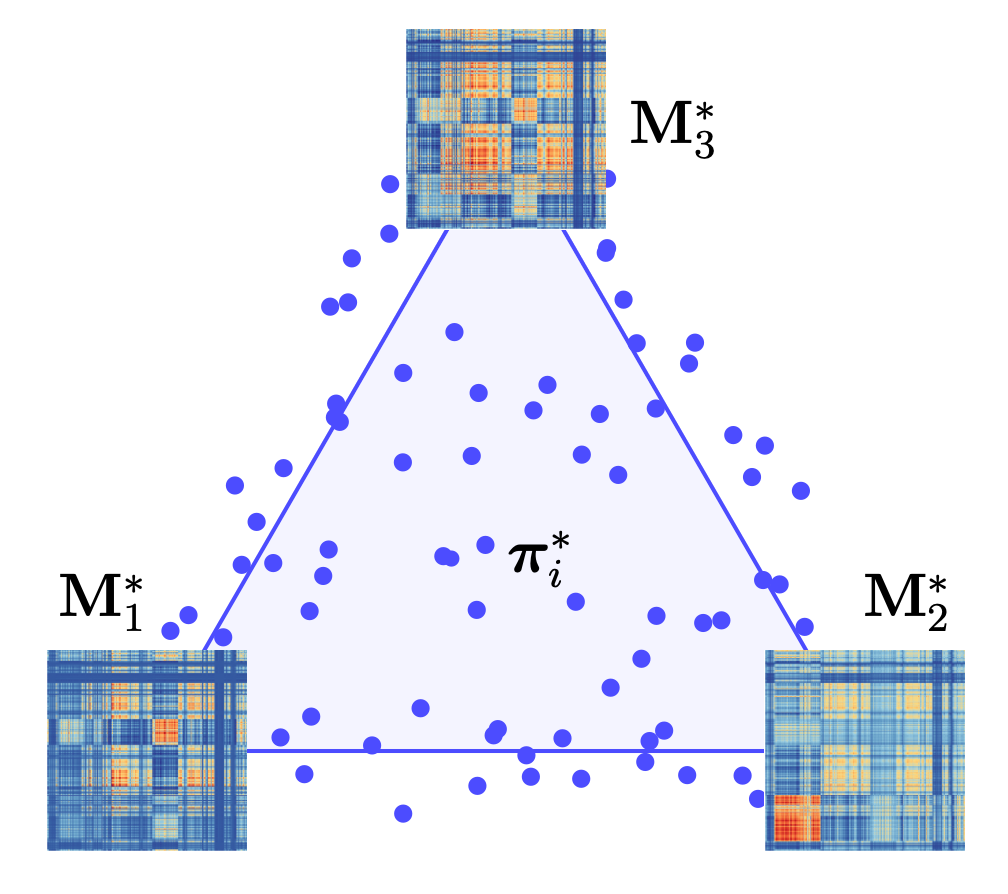}
        \subcaption{Proposed LrMMM}
    \end{subfigure}

    \caption{\rev{Comparison of tensor PCA models \citep{Liu2023joint, Weylandt2025SS-TPCA}, low-rank clustering \citep{LX25}, and the proposed LrMMM. Tensor PCA gives continuous scores with reconstruction-oriented loadings; low-rank clustering gives interpretable prototypes with hard labels; LrMMM gives interpretable low-rank extreme profiles with continuous simplex memberships.}}
    \label{fig:model-illustration}
\end{figure}

We address this gap by proposing a low-rank mixed membership model for matrix-valued observations. Mixed membership models represent each observation as a convex combination of latent extreme profiles, a perspective that underlies topic models such as latent Dirichlet allocation \citep{Blei2003latent} and grade-of-membership models in survey analysis \citep{Erosheva02}. We extend this perspective to matrix-valued data by modeling each subject's mean matrix as a convex combination of $K$ latent low-rank basis matrices. The basis matrices play the role of interpretable population-level extremes, while the convex weights provide a continuous embedding of each subject in the probability simplex. Thus, unlike tensor PCA, the latent matrices are intended to characterize population profiles directly. Unlike clustering, subjects are allowed to express several profiles simultaneously.

We also introduce the extension of the proposed model for multimodal studies. In many scientific settings, multiple sources of information describing each observation are available; multimodal methods have successfully exploited such additional data to improve signal and generate new insights in biology and social sciences \citep{Lock2013JIVE, Feng2018angle, Liu2023joint}. We focus on the setting in which each observation contributes several matrices, such as functional connectivity networks collected under different tasks or multiple network modalities measured on the same population. Our multimodal extension shares one membership vector across modalities while allowing each latent profile to have modality-specific low-rank basis matrices. This structure pools information across data sources without forcing the same matrix pattern to appear in every modality.

This paper makes four primary contributions. First, we introduce the mixed membership framework for low-rank matrix-valued data, bridging the gap between tensor factorization methods that provide continuous subject representations and clustering methods that estimate interpretable population prototypes. Second, we extend the framework to multimodal matrix-valued data by sharing subject memberships across modalities while allowing modality-specific basis matrices, and establish identifiability under mild conditions. Third, we develop an efficient estimation procedure with provable statistical guarantees, including minimax-optimal reconstruction error (up to a logarithmic factor) and convergence rates for both latent basis matrices and subject memberships. Finally, we demonstrate through simulations and Human Connectome Project data that the proposed approach yields interpretable latent connectivity patterns and biologically meaningful continuous subject embeddings.

The rest of the paper is organized as follows. Section~\ref{sec:Model} introduces LrMMM and its multimodal extension, and establishes model identifiability. Section~\ref{sec:Estimation} presents the constrained estimator and practical algorithms. Section~\ref{sec:StatisticalGuarantees} gives the theoretical guarantees. Section~\ref{sec:Simulations} reports simulation studies. Section~\ref{sec:DataAnalysis} analyzes Human Connectome Project data. Proofs and additional numerical details are provided in the Supplementary Material.

\section{Proposed Model and Its Identifiability}\label{sec:Model}
\subsection{Notation}

For any matrix $\A \in \RR^{n\times p}$,  we use $\mathbf{A}_{i,:}$ or $\mathbf{A}(i,:)$ to denote its $i$-th row vector, and we use $\mathbf{A}_{:,j}$ or $\mathbf{A}(:,j)$ to denote its $j$-th column vector. Let $\text{vec}(\A)\in\RR^{np}$ be the vector satisfying $\text{vec}(\A)((j-1)n + i) = \A_{ij}$, and let $\text{mat}_{n,p}$ be the function satisfying $\A = \text{mat}_{n,p}(\text{vec}(\A))$. Let $\sigma_k(\mathbf{A})$ denote the $k$-th largest singular value of $\mathbf{A}$ for $k = 1,...,\min\{n,p\}$. Let $\|\cdot\|$ denote the spectral norm (operator norm) for matrices and $\ell_2$ norm for vectors, and let $\|\cdot\|_{2,\infty}, \|\cdot\|_{\mathrm{F}}$ denote the two-to-infinity norm and Frobenius norm for matrices, respectively. Denote by $\OO_{d,r}$ the set of all $d \times r$ matrices $\U$ such that $\U^\top \U = \I_r$, where $\I_r$ is the $r \times r$ identity matrix. For matrices $\A \in \RR^{n_1 \times p_1}, \B \in \RR^{n_2 \times p_2}$, their Kronecker product is defined as the matrix $\A \otimes \B \in \RR^{n_1n_2 \times p_1p_2}$ satisfying $(\A\otimes\B)((i_1-1)n_2 + i_2, (j_1-1)p_2 + j_2) = \A(i_1, j_1)\B(i_2, j_2)$ for $i_k \in [n_k]$, $j_k \in [p_k]$, $k=1,2$. 

An order-p tensor is a $p$-dimensional array. For tensor $\bcalA \in \RR^{d_1 \times ... \times d_p}$, define $\text{vec}(\bcalA)\in\RR^{d_1d_2...d_p}$ such that $\text{vec}(\bcalA)(i_1 + (i_2 - 1)d_1 + (i_3 - 1)d_1d_2 + ... + (i_p-1)d_1d_2...d_{p-1}) = \bcalA(i_1, i_2, ..., i_p)$, and define $\bcalA(:,...,:, i_p)$ to be the order $(p-1)$ tensor satisfying $(\bcalA(:,...,:, i_p))(i_1, ..., i_{p-1}) = \bcalA(i_1, ..., i_p)$. Define the mode-$1$ unfolding of $\bcalA$, denoted $M_1(\bcalA)$, as the matrix in $\RR^{d_1 \times  d_2....d_p}$ whose $(i_1)$-th row is $\text{vec}(\bcalA(i_1,:,...,:))$; we define the mode-$k$ unfolding analogously. The mode-1 marginal multiplication between $\bcalA$ and a matrix $\mathbf{U}^\top \in \RR^{r_1 \times d_1}$ results in a tensor in $\RR^{r_1 \times d_2 \times ... \times d_p}$, whose elements are
$(\A \times_1 \U^\top)(j_1,i_2,..., i_p)
:= \sum_{i_1=1}^{d_1} \A(i_1,...,i_3)\U(i_1,j_1), \forall j_1 \in [r_1],\, i_2 \in [d_2],..., i_p \in [d_p]$. Mode $k$ multiplication is defined analogously.

We use $n\vee p:= \max(n,p)$ and $n \wedge p := \min(n,p)$. For matrices $\A, \B$, we say $\A \preceq \B$ (resp. $\A \succeq \B$) if and only if all singular values of $\B - \A$ (resp. $\A - \B$) are nonnegative. For two sequences $\{a_N\}, \{b_N\}$, we use $a_N \lesssim b_N$ (resp. $a_N \gtrsim b_N$) if and only if there exists some constant $C>0$ independent of $N$ such that $a_N \le C b_N$ (resp. $b_N \le C a_N$), and use $a_N \asymp b_N$ if and only if $a_N \lesssim b_N$ and $a_N \gtrsim b_N$ hold simultaneously. We use $a_N \ll b_N$ if and only if $a_N / b_N \to 0$. Define $\Delta_K:=\{\mathbf{x} \in \RR^K : \sum_i \mathbf{x}_i = 1, \text{ } \mathbf{x}_i \geq 0 \text{ } \forall i \}$, and define $S_K$ as the set of permutations on $K$ elements.

\subsection{Low-rank Mixed Membership Model}\label{sec:Model-Intro}

Consider a dataset with $n$ subjects, and suppose that each subject's data comes in the form of a real-valued $d_1 \times d_2$ matrix. Denote the matrix corresponding to the $i$-th subject as $\X_i$. At the level of the population, we assume there are $K$ latent \textit{extreme profiles}, each representing a prototypical matrix-valued pattern. In particular, the pattern corresponding to extreme profile $k$ is given by a deterministic, unknown \textit{basis matrix} $\M^*_k \in \RR^{d_1 \times d_2}$, which we assume to be a low-rank matrix with rank $r_k$ (where $r_k\ll d_1,d_2$). The $k$-th basis matrix represents the expectation of a matrix corresponding to a subject belonging solely to the $k$-th extreme profile. That is, if subject $i$ belongs solely to extreme profile $k \in [K]$, then $\X_i$ satisfies $\EE[\X_i] = \M^*_k$.
At the level of the individual, each subject $i$ is associated with a deterministic, unknown vector of mixed membership scores $\bpi^*_i = (\bpi^*_{i1}, ..., \bpi^*_{iK})$, with $\bpi^*_{ik} \geq 0$ and $\sum_{k=1}^K \bpi^*_{ik}=1$. The membership score $\bpi_{ik}^*$ denotes the extent to which subject $i$ partially belongs to extreme latent profile $k$. In particular, the expectation of the $i$-th subject's matrix is a convex combination of the basis matrices with weights determined by the membership scores. Altogether, the Low-rank Mixed Membership Model (LrMMM) assumes that: \begin{align}
    \X_i = \sum_{k=1}^K \bpi^*_{ik} \M_k^* + \tilde{\E}_i, \quad \text{ for } i=1,...,n,
    \label{model}
\end{align}
where $(\tilde{\E}_i)_{i=1}^n$ is a set of $n$ independent, mean-zero noise matrices of size $d_1 \times d_2$, each of which satisfies the sub-Gaussian error condition specified in Assumption \ref{ass:noise}. The goal of an analysis using the Low-rank Mixed Membership Model is to estimate the membership scores and basis matrices and thereby uncover the prototypical patterns in the population, along with the extent to which each individual displays those patterns. 

\begin{assumption}[Noise Distribution]
    \label{ass:noise}
    The entries of $\tilde{\E}_i$ are independent sub-Gaussian random variables. In particular, denoting the sub-Gaussian norm as $\|X\|_{\psi_2} = \inf \{ c>0: \EE [e^{X^2/c^2}] \leq 2 \}$, we have $\|\tilde{\E}_i(j_1, j_2)\|_{\psi_2} \leq \sigma$ for all $i, j_1, j_2$.
\end{assumption}

The principal difference between our model and the models used by other linear dimensionality reduction methods for matrix-valued observations lies in the structure of the embedding obtained by subject-level coefficients, given in our case by $(\bpi^*_i)_{i=1}^n$. The low-rank clustering method of \cite{LX25} models each subject's coefficients $\bpi^*_i$ to be a standard basis vector ascribing membership entirely to one class, whereas Tensor PCA methods model each subject's coefficients to be a real-valued vector \citep{Tucker1966some}. Figure \ref{fig:model-illustration} provides an illustration of these modeling differences, with subjects' data represented by blue dots.

In the context of the mixed membership modeling literature, the LrMMM is related to the mixed membership model of \cite{ChenGu24} for vector-valued observations. In their model, subject i's data $\y_i$ is a real-valued $p$-dimensional vector modeled as a convex combination of $K$ prototypical vectors: $\y_i = \sum_{k=1}^K \bpi^*_{ik} \btheta^*_k + \boldsymbol{\epsilon}_i$. Their computationally efficient spectral estimation method relies on the following structure inherent to their model. If $\Y$ is the $n\times p$ matrix whose $i$-th row is $\y_i$, $\bPi^*$ is the $n \times K$ matrix whose $i$-th row is $\bpi^*_i$, and $\bTheta^*$ is the $p \times K$ matrix whose $k$-th column is $\btheta^*_k$, then their model satisfies \begin{align}
    \EE[\Y] = \bPi^* \bTheta^{*\top} = \U^* \bSigma^* \V^{*\top},
    \label{spectral-unfolded}
\end{align}
where $\U^*$, $\bSigma^*$, and $\V^*$ are matrices from the compact-form SVD of $\EE[\Y]$. Their method applies rank-$K$ truncated SVD of $\Y$ to estimate $\U^*$, $\bSigma^*$, and $\V^*$, and applies additional post-processing and regression steps to estimate the membership vectors and extreme profile-level parameters. Their method's good performance relies on the ability of truncated SVD of $\Y$ to accurately estimate the compact form SVD of $\EE[\Y]$ provided the signal-to-noise ratio is sufficiently strong \citep{chen2021SMDS}. If we let $\y_i = \text{vec}(\X_i)$, $\btheta^*_k = \text{vec}(\M^*_k)$ and $p=d_1d_2$, we observe that our model can also be written in the form of \eqref{spectral-unfolded}, with the exception that our model enforces additional low-rankness at the level of the extreme profile-level parameters by requiring that $\text{rank}(\text{mat}_{d_1, d_2}(\btheta^*_k)) = r_k$. Accordingly, we will use the compressed notation $\bPi^*$, $\bTheta^*$ and $\Y$ to refer to our model's parameters when doing so is notationally convenient.

Our model may also be compared to the recently proposed Tensor Mixed-Membership Blockmodel of \cite{AZ25}. When applied to a dataset of matrix-valued observations, their model can be regarded as a special case of our model that additionally assumes the basis matrices have mixed-membership stochastic blockmodel structure of the form $\M^*_k = \bPi^*_1 \S_k^* (\bPi^*_2)^\top$ for $k=1,...,K$ \citep{Airoldi2008mixed}. Our model is more flexible and general than their method, as illustrated in Section \ref{sec:DataAnalysis}. Furthermore, our model is also readily extended for multimodal studies, which is discussed next.

\subsection{Multimodal Extension}

Suppose for each subject $i$, we have access to $M$ matrices, each representing data from a different ``modality," where the $m$-th matrix $\X_{i,m}$ is a $d_{1,m} \times d_{2,m}$ matrix, that is, the matrices may differ in size across modalities. The Multimodal Low-rank Mixed Membership Model (m-LrMMM) is written as
\begin{equation}
    \X_{i,m} = \sum_{k=1}^K \bpi^*_{ik} \M_{k,m}^* + \tilde{\E}_{i,m} \text{ for } i=1,...,n, \text{ } m=1,...,M,
    \label{mm-model}
\end{equation}
where $\bpi_i^*$ satisfy the same conditions as in (\ref{model}), $\M_{k,m}^*$ is a $d_{1,m} \times d_{2,m}$ matrix of rank $r_{k,m}$, and $\tilde{\E}_{i,m}$ is a $d_{1,m} \times d_{2,m}$ mean-zero noise matrix.
Similar to model (\ref{model}), model (\ref{mm-model}) retains only $K$ extreme latent profiles, and each subject $i$ retains one $K$-dimensional vector of membership scores $\bpi_i^*$. However, each extreme latent profile $k$ is allowed to have $M$ different basis matrices $(\M_{k,m}^*)_{m=1}^M$, each representing the prototypical pattern of the extreme latent profile within that modality. Note that setting $M=1$ recovers model (\ref{model}).

\subsection{Identifiability}

To ensure valid and reproducible estimation, it is essential to establish model identifiability. Conventionally, a model is said to be identifiable if there is a one-to-one mapping from the parameter space to the set of likelihood functions. However, factor models often require different notions of identifiability due to their inherent rotational ambiguity \citep{BaiLi2012}. In the case of LrMMM, there is inherent permutation ambiguity. Suppose the true parameters of the model are $(\bpi^*_i)_{i=1}^n$, $(\M^*_k)_{k=1}^K$. Let $\sigma$ be a permutation on the set of $K$ elements, and define parameter sets $(\bpi'_i)_{i=1}^n$ and $(\M'_k)_{k=1}^K$ such that $\bpi'_{ik} = \bpi^*_{i\sigma(k)}$ for all $i,k$ and $\M'_k = \M^*_{\sigma(k)}$ for all $k$. Then notice that $\EE[\X_i] = \sum_{k=1}^K \bpi^*_{ik} \M^*_k = \sum_{k=1}^K \bpi'_{ik} \M'_k$ for all $i$. That is, if the labels of the parameters reconstructing the denoised dataset are permuted, they still reconstruct the same denoised dataset. Following \cite{ChenGu24}, we adopt the following notion of identifiability, which accounts for such permutation ambiguity. 
\begin{definition}\label{def:exp-ident}
    A LrMMM model given by parameter sets $(\bpi^*_i)_{i=1}^n$ and $(\M^*_k)_{k=1}^K$ is said to be identifiable if for any other parameter sets $(\bpi'_i)_{i=1}^n$ and $(\M'_k)_{k=1}^K$ satisfying $\sum_{k=1}^K \bpi^*_{ik} \M^*_k = \sum_{k=1}^K \bpi'_{ik} \M'_k$ for all $i$, there exists a permutation $\sigma \in S_K$ such that $\bpi^*_{ik} = \bpi'_{i\sigma(k)}$ and $\M^*_k = \M'_{\sigma(k)}$.
\end{definition}

\noindent An LrMMM model satisfying \Cref{def:exp-ident} has the property that any other parameter set that also reconstructs the denoised dataset must have identical parameters after permutation of the labels. We now introduce an assumption under which we state our identifiability result. We say that subject $i$ is a ``pure subject" for extreme latent profile $k$ if $\bpi^*_i = \e_k$, the $k$-th standard basis vector.

\begin{assumption}[Pure Subjects]\label{ass:pure-subjects}
    $(\bpi^*_i)_{i=1}^n$ satisfies the property that every extreme latent profile has at least one pure subject. 
\end{assumption} 
\noindent The pure subject condition in Assumption~\ref{ass:pure-subjects} is analogous to the anchor word assumption in topic modeling \citep{KeWang24} and the pure node assumption in mixed membership stochastic block model for networks \cite{JinKL24, MaoSC2021}. Under Assumption \ref{ass:pure-subjects}, the following identifiability result holds as a result of Theorem 2 in \cite{ChenGu24}.
\begin{proposition}
\label{prop:exp-id}
    Suppose Assumption \ref{ass:pure-subjects} holds and one of the following two conditions hold: (a) $\text{rank}(\bTheta^*) = K$, or (b) both $\text{rank}(\bTheta^*) = K-1$ and no column of $\bTheta^*$ is an affine combination of other columns of $\bTheta^*$. Then model (\ref{model}) is identifiable. 
\end{proposition}

Now, we provide a first glance at the relative strength of this multimodal model through an identifiability result. Firstly, we introduce some notation. Let $\bTheta_m^*$ be the $d_{1,m} d_{2,m} \times K$ matrix whose $k$-th column is $\text{vec}(\M_{k,m}^*)$. Let $p_m = d_{1,m}d_{2,m}$ and $p_0 = \sum_{m=1}^M p_m$. Let $\bTheta_0^* = (\bTheta_1^{*\top}, ..., \bTheta_M^{*\top})^\top$, that is, $\bTheta_0^*$ is the long $p_0 \times K$ matrix given by vertical concatenation of the matrices $(\bTheta_m^*)_{m=1}^M$. The full definition of identifiability of the multimodal model, which is the natural generalization of Definition \ref{def:exp-ident}, is presented in Section~S.3.1 of the Supplement.

\begin{proposition}
\label{prop:mm-exp-id}
    Consider model (\ref{mm-model}) with parameters $(\bPi^*, \bTheta_0^*)$. Suppose $\bPi^*$ satisfies Assumption \ref{ass:pure-subjects} and one of the following two conditions hold: (a) $\text{rank}(\bTheta_0^*) = K$, or (b) both $\text{rank}(\bTheta_0^*) = K-1$ and no column of $\bTheta_0^*$ is an affine combination of other columns of $\bTheta^*$. Then the model is identifiable. 
\end{proposition}

Proposition \ref{prop:mm-exp-id} shows that the condition under which model (\ref{mm-model}) is identifiable is a relaxation of the condition under which model (\ref{model}) is identifiable according to Proposition \ref{prop:exp-id}. For example, suppose that matrix $\bTheta_1^*$, which contains the vectorized basis matrices from the first modality, is not full rank. Then the unimodal model that uses only that modality may not be identifiable, since Proposition \ref{prop:exp-id} no longer holds. But after incorporating data from multiple modalities, as long as the concatenation $\bTheta_0^*$ is full rank, the model is identifiable. Hence, signals that are not identifiable within one modality may become identifiable after incorporating data across multiple modalities. 

\section{Estimation}\label{sec:Estimation}

\subsection{Two-step Estimator}
Our estimator takes inspiration from the K-means optimization problem, and in particular, its low-rank adaptation in \citep{LX25}. In their work, each subject has a class membership vector $s_i \in [K]$ denoting membership entirely to one of $K$ classes. Their algorithm solves the problem $\min \sum_{i=1}^n \sum_{k=1}^K \mathbbm{1}(s_i = k) || \X_i - \M_k||_F^2$ subject to the constraint $\text{rank}(\M_k) = r_k\; (\forall k)$. We solve an analogous problem that relaxes discrete membership to continuous membership and enforces the pure subjects constraint. In particular, we first introduce an estimator that relies on oracle knowledge of the indices of the pure subjects, and then we demonstrate that when we lack such oracle knowledge, we can recover the same estimator with high probability when using an estimate of the pure subject indices.

Suppose we have oracle access to the pure subject indices $S^* = (S^*_1, ..., S^*_K) \subset [n]$. Consider the optimization problem:
\begin{align}\label{eq:estimator-oracle}
\begin{split}
    \min_{(\bpi_i), (\M_k)} \text{ }  \sum_{i=1}^n \|\X_i - \sum_{k=1}^K \bpi_{ik} \M_k \|_F^2 \quad \text{s.t. } \quad & \text{(a) } \text{rank}(\M_k) = r_k\;(\forall k), \\ 
    & \text{(b) } \bpi_i \in \Delta_K\;(\forall i), \quad \text{(c) } \bpi_{S^*_k} = \e_k\;(\forall k).
\end{split}
\end{align}
\noindent
Constraint (a) enforces the low-rankness constraint, constraint (b) enforces the simplex constraint on the membership vectors, and constraint (c) enforces the pure subjects constraint by constraining the subject known to be a pure subject for extreme profile $k$ to belong entirely to extreme profile $k$.

When we lack oracle knowledge of the indices of the pure subjects, they can be estimated exactly with high probability using the method of \cite{CHG24} provided certain conditions hold. In short, we obtain estimate $\hat{S} = (\hat{S}_1, ..., \hat{S}_K) \subset [n]$ of the pure subject indices by performing rank-K truncated SVD of $\Y$ to obtain left-singular subspace $\hat{\U}$, and then obtain estimate $\hat{S}$ given by the indices selected by the vertex-hunting algorithm called Successive Projection Algorithm (SPA) applied on the rows of $\hat{\U}$. To provide some intuition on their estimation procedure, we first share a Lemma from \cite{CHG24} which relates the terms of \eqref{spectral-unfolded}:
\begin{lemma}
    \label{lem:convex-hull}
    Let Assumption \ref{ass:pure-subjects} hold, let $\text{rank}(\bTheta^*)=K$, and let $S^* = (S^*_1, ..., S^*_K)$ denote the indices of the pure subjects, with $S^*_k$ being a pure subject from extreme latent profile $k$. Then 
    \begin{align}\label{eq:subspace-simplex}
        \U^* = \bPi^* \U^*_{S^*, :}
    \end{align}
\end{lemma}
Lemma \ref{lem:convex-hull} implies that the rows of left-singular subspace $\U^*$ of $\EE[\Y]$ lie in a simplex whose vertices are given by the indices of the pure subjects. Thus, if we had access to $\U^*$, we could apply a vertex-hunting algorithm on the rows of $\U^*$ to exactly identify the indices of the pure subjects. In reality, we lack access to $\U^*$, but we can estimate it well using the left singular subspace $\hat{\U}$ obtained by rank-$K$ truncated SVD of $\Y$ provided some signal-to-noise conditions hold. The vertex-hunting algorithm SPA is then applied on $\hat{\U}$. In fact, a consequence of the $l_{2,\infty}$ analysis in \cite{CHG24} of the error of $\hat{\U}$ is a uniform consistency result on the estimation of all membership vectors. We use this result to guarantee exact recovery of the pure subject indices with high probability provided the pure subjects and non-pure subjects are sufficiently separated in membership space. Hence, the following estimator, which uses no oracle knowledge, is equal to \eqref{eq:estimator-oracle} with high probability:
\begin{equation}\label{eq:estimator-twostep}
\begin{split}
    \min_{(\bpi_i), (\M_k)} \text{ }  \sum_{i=1}^n \|\X_i - \sum_{k=1}^K \bpi_{ik} \M_k \|_F^2 \quad \text{s.t. } \quad & \text{(a) } \text{rank}(\M_k) = r_k\;(\forall k), \\ 
    & \text{(b) } \bpi_i \in \Delta_K\;(\forall i), \quad \text{(c) } \bpi_{\hat{S}_k} = \e_k\;(\forall k).
\end{split}
\end{equation}
Estimator \eqref{eq:estimator-twostep} differs from estimator \eqref{eq:estimator-oracle} only in constraint (c), which in this case uses the estimated pure subject indices $\hat{S}$ instead of the true pure subject indices $S^*$. Since the ``two-step estimator" in \eqref{eq:estimator-twostep} requires no oracle knowledge, we propose to use this estimator instead of \eqref{eq:estimator-oracle}, noting that it performs just as well as \eqref{eq:estimator-oracle} with high probability.

\subsection{Estimation Algorithm}

Just as in the case of K-means clustering and its low-rank counterpart \citep{Lloyd1982, LX25}, solving the optimization problem given by \eqref{eq:estimator-twostep} directly is intractable. Like the original K-means clustering problem, there is nonconvexity in the objective function, and like the low-rank K-means problem, there is additional nonconvexity in the basis parameters. In line with these past works, we propose an alternating algorithm for use in practice. The empirical consistency results presented in Section \ref{sec:Simulations} justify our proposed algorithm. The full algorithm, including both initialization and iterative updates, is presented in Algorithm \ref{alg:Algorithm}. 

At a high level, our algorithm uses a spectral initialization of the membership vectors $(\hat{\bpi_i}^{(0)})_{i=1}^n$. Our algorithm then alternates between obtaining approximate updates to the basis parameters (``Problem 1") and membership parameters (``Problem 2"): 
\begin{enumerate}
    \item
    $\min_{(\M_k)_{k=1}^K} 
    \sum_{i=1}^n \left\| \X_i - \sum_{k=1}^K \hat{\bpi}_{ik}^{(t-1)} \M_k \right\|_F^2 $
    subject to $\operatorname{rank}(\M_k)=r_k$ for all $k$.

    \item
    $\min_{(\bpi_i)_{i=1}^n}
 \sum_{i=1}^n \left\| \X_i - \sum_{k=1}^K \bpi_{ik}\hat{\M}_k^{(t)} \right\|_F^2
    $
    subject to $\bpi_i \in \Delta_K$ for all $i$, and
    there exist $(S_1,\dots,S_K)$ such that $\bpi_{S_k,:}=\e_k$ for all $k$.
\end{enumerate}

Our initialization starts with an estimate $\hat{\U}$ of the left singular subspace of $\EE[\Y]$, and proceeds according to the same method used in \citep{CHG24} to estimate the membership parameters, which is detailed in Algorithm \ref{alg:GoM}. To estimate $\U^*$, instead of using truncated SVD of $\Y$, we apply the tensor decomposition method Higher-Order Orthogonal Iteration (HOOI) on an aggregated data tensor $\bcalX \in \RR^{d_1 \times d_2 \times n}$ which satisfies by $\bcalX(:,:,i) = \X_i$. This choice exploits the tensor structure of our data to estimate $\U^*$ more efficiently, since it has been proven that HOOI estimates mode-wise subspaces more accurately than ``Higher-Order SVD," which is SVD of the mode-wise unfolding ($\Y$ in our case), 
when the dataset has the structure of a low-rank tensor \citep{zhang2018tensor}. The choice to use HOOI was also made by \cite{AZ25} in their study of a related mixed-membership model with tensor structure. Full details underlying this choice are described in Section~S.3.2 of the Supplement. Finally, once we have singular subspace estimate $\hat{\U}$, we use the method of \cite{CHG24}. In particular, we apply SPA to identify the pure subject indices $\hat{S}$ and, taking the lead from \eqref{eq:subspace-simplex}, estimate $\bPi^*$ using $\tilde{\bPi} = \hat{\U} \hat{\U}^{-1}_{\hat{S},:}$. This estimate $\tilde{\bPi}$ is post-processed by projecting all rows onto $\Delta_K$ in the final step to obtain $(\hat{\bpi}^{(0)}_i)_{i=1}^n$. Full details are provided in Algorithm \ref{alg:GoM}.

\begin{algorithm}
\caption{(GoM) Estimation of Membership Matrix of Grade-of-Membership Model via Successive Projection Algorithm}
\label{alg:GoM}
\begin{algorithmic}[1] % The number tells where the line numbering should start
    \Require Estimated singular subspace $\hat{\U}$
    \Ensure Estimated membership matrix $\hat{\bPi}$
    \State $\Z = \hat{\U}$
    \For{$k=1$ to $K$}
        \State $\hat{S}_k = \text{argmax}(\|\Z_{i,:}\|_2^2: i \in [n])$
        \State $\u = \Z_{\hat{S}_k,:}/\|\Z_{\hat{S}_k,:}\|_2$
        \State $\Z = \Z(\I_K - \u\u^\top)$
    \EndFor
    \State $\tilde{\bPi} = \hat{\U} (\hat{\U}_{\hat{S},:})^{-1}$
    \State $\hat{\bPi}_{i,:} = \text{argmin}_{\pi \in \Delta_K} \| \pi - \tilde{\bPi}_{i,:}\|_2^2$ and $\hat{\bPi} = (\hat{\bPi}_{i,:})_{i=1}^n$
    \State \Return $\hat{\bPi}$
\end{algorithmic}
\end{algorithm}

As for the basis matrix updates, we observe that Problem 1 
is non-convex. Following \citep{LX25}, we approximate this update by first minimizing $\sum_{i=1}^n || \X_i - \sum_{k=1}^K \hat{\bpi}^{(t-1)} \M_k ||_F^2 $ without enforcing the low-rankness constraint, which amounts to Ordinary Least Squares, and then apply truncated rank-$r_k$ SVD on each solution $\M_k$ to enforce the low-rankness constraint. The form of the OLS problem after unfolding is $\text{argmin}_{\bTheta} \| \Y - \bPi^{(t-1)} \bTheta^\top||_F^2$.

\begin{algorithm}
\caption{LrMMM}
\label{alg:Algorithm}
\begin{algorithmic}[1] % The number tells where the line numbering should start
    \Require Dataset $(\X_i)_{i=1}^n \subset \RR^{d_1 \times d_2}$, number of latent structures $K$, ranks of latent structures $(r_1,...,r_K)$, rank of concatenated left subspaces $r_U$, rank of concatenated right subspaces $r_V$, number of iterations $T$
    \Ensure Estimated basis matrices $( \hat{\M}_k )_{k=1}^K$, estimated membership vectors $(\hat{\bpi}_i)_{i=1}^n$
    \State Form tensor $\bcalX\in \RR^{d_1 \times d_2 \times n}$ with $\bcalX(:,:,i) = \X_i$
    \State $(\hat{\bcalC}, \hat{\W}_1, \hat{\W}_2, \hat{\W}_3) = \text{HOOI}(\bcalX, \text{Tucker Rank} = (r_U, r_V, K))$ \label{line:HOOI}
    \State $\hat{\bPi}^{(0)} = \text{GoM}(\hat{\W}_3)$
    \For{$t = 1$ to $T$}
        \State $\tilde{\M}^{(t)}_k = \text{mat}_{d_1,d_2}(\tilde{\bTheta}^{(t)}_{:,k})$ where $\tilde{\bTheta}^{(t)} = M_3(\bcalX)^\top \hat{\bPi}^{(t-1)} (\hat{\bPi}^{(t-1)\top} \hat{\bPi}^{(t-1)})^{-1}$ \label{line:OLS}
        \State $\hat{\M}^{(t)}_k = \text{SVD}(\tilde{\M}^{(t)}_k, \text{rank} = r_k)$ \label{line:proximal}
        \State $\tilde{\bPi}_{i,:}^{(t)} = \text{argmin}_{\pi \in \Delta^{K-1}} \| \X_i - \sum_k \pi_k \hat{\M}^{(t)}_k\|_F^2$ and $\tilde{\bPi}^{(t)} = (\tilde{\bPi}_{i,:}^{(t)})_{i=1}^n$ \label{alg:proj-step}
        \State $\hat{\bPi}^{(t)} = \text{GoM}(\tilde{\bPi}^{(t)})$
    \EndFor
    \State $\hat{\M}_k = \hat{\M}^{(T)}_k \text{ } (\forall k)$, and  $\hat{\bpi}_i \leftarrow \hat{\bPi}^{(T)}_{i,:} \text{ }(\forall i)$
    \State \Return $(\hat{\M}_k)_{k=1}^K, (\hat{\bpi}_i)_{i=1}^n$
\end{algorithmic}
\end{algorithm}

As for the updates to the membership vectors, Problem 2
notably differs in structure from \eqref{eq:estimator-twostep} in the way that it enforces the pure subjects constraint. Instead of using the estimated indices of the pure subjects obtained in the initialization as done by constraint (c) of \eqref{eq:estimator-twostep}, the estimated pure subject indices are allowed to update in each iteration. We make this choice based on our observation that in simulations, iteratively updating the pure subject indices leads to improved accuracy. This change comes at the cost of nonconvexity of the membership update. Our solution is to first solve $\min_{(\bpi_i)_{i=1}^n} \sum_{i=1}^n \left\| \X_i - \sum_{k=1}^K \bpi_{ik}\hat{\M}_k^{(t)} \right\|_F^2$ subject to the simplex constraint $\bpi_i \in \Delta_K \text{ } (\forall i)$ but with no enforcement of the pure subjects constraint. This amounts to $n$ separate quadratic programming problems. We write the solution to this problem as $(\tilde{\bpi}^{(t)})_{i=1}^n = \tilde{\bPi}^{(t)}$, and we regard those $\tilde{\bpi}^{(t)}_i$ closest to the standard basis vectors as those which the update deems most likely to be pure subjects. Given the simplex structure of $\tilde{\bPi}^{(t)}$, it is reasonable to use SPA to identify the most extreme points. We observe in practice that the difference between applying SPA on the left singular subspace of $\tilde{\bPi}^{(t)}$ or directly on $\tilde{\bPi}^{(t)}$ is negligible. Hence, for simplicity, we apply the same GoM method applied previously on $\hat{\U}$ (Algorithm \ref{alg:GoM}), directly on $\tilde{\bPi}^{(t)}$ itself to re-estimate the pure subjects and obtain update $\hat{\bPi}^{(t)}$.

\subsection{Multimodal Algorithm}\label{sec:MultimodalAlgorithm}

\begin{algorithm}[t]
\caption{Multimodal-LrMMM}
\label{alg:Algorithm-MM}
\begin{spacing}{1}
% \footnotesize
\begin{algorithmic}[1] % The number tells where the line numbering should start
    \Require $((\X_{i,m})_{m=1}^M)_{i=1}^n$, $K$, $(r_{k,m})_{k,m}$, $(r_{U,m})_m$, $(r_{V,m})_m$, $T$, weights $\gamma = (\gamma_1, ..., \gamma_M)$
    \Ensure Estimated basis matrices $(( \hat{\M}_{k,m} )_{k=1}^K)_{m=1}^M$, estimated membership vectors $(\hat{\bpi}_i)_{i=1}^n$
    \For{$m=1$ to $M$}\label{alg-mm:first-step}
    \State Form tensor $\bcalX_m \in \RR^{d_{1,m} \times d_{2,m} \times n}$ with $\bcalX_m(:,:,i) = \X_{i,m}$ \label{alg-mm:form-tensor}
    \State $(\hat{\bcalC}_m, \hat{\W}_{1,m}, \hat{\W}_{2,m}, \hat{\W}_{3,m}) = \text{HOOI}(\bcalX_m, \text{Tucker Rank} = (r_{U,m}, r_{V,m}, K))$ \label{alg-mm:HOOI}
    \State $\tilde{\bPi}_m^{(0)} = \text{GoM}(\hat{\W}_{3,m})$ \label{alg-mm:GoM}
    \State Solve $\P_m = \text{argmin}_{\P} \|\tilde{\bPi}_1^{(0)} - \tilde{\bPi}_m^{(0)}\P\|_F^2$ and align $\tilde{\bPi}_m^{(0)} \leftarrow \tilde{\bPi}_m^{(0)} \P_m$ \label{alg-mm:align}
    \EndFor
    \State $\hat{\bPi}^{(0)} = \text{GoM}((1/\sum_m \gamma_m) \sum_{m=1}^M \gamma_m \tilde{\bPi}_m^{(0)})$ \label{alg-mm:renormalize-init}
    \For{$t = 1$ to $T$}
        \For{$m = 1$ to $M$}
            \State $\tilde{\M}^{(t)}_{k,m} = \text{mat}_{d_{1,m}, d_{2,m}} ((\tilde{\bTheta}^{(t)}_{m})_{:,k})$, where $\tilde{\bTheta}_m^{(t)} = M_3(\bcalX_m)^\top \hat{\bPi}^{(t-1)} (\hat{\bPi}^{(t-1)\top} \hat{\bPi}^{(t-1)})^{-1}$\label{alg-mm:theta-tilde}
            \State $\hat{\M}_{k,m}^{(t)} = \text{SVD}(\tilde{\M}^{(t)}_{k,m}, \text{rank} = r_{k,m})$ 
        \EndFor
        \State $\tilde{\bPi}_{i,:}^{(t)} = \text{argmin}_{\pi \in \Delta^{K-1}} \sum_{m=1}^M \gamma_m \| \X_{i,m} - \sum_k \pi_k \hat{\M}_{k,m}^{(t)} \|_F^2$ and $\tilde{\bPi}^{(t)} = (\tilde{\bPi}_{i,:}^{(t)})_{i=1}^n$ \label{alg-mm:proj-step}
        \State $\hat{\bPi}^{(t)} = \text{GoM}(\tilde{\bPi}^{(t)})$
    \EndFor
    \State $\hat{\M}_{k,m} = \hat{\M}^{(T)}_{k,m} \text{ } (\forall k, m)$, and  $\hat{\bpi}_i \leftarrow \hat{\bPi}^{(T)}_{i,:} \text{ }(\forall i)$
    \State \Return $((\hat{\M}_{k,m})_{k=1}^K)_{m=1}^M, (\hat{\bpi}_i)_{i=1}^n$
\end{algorithmic}
\end{spacing}
\end{algorithm}

We extend our algorithm to the multimodal case in Algorithm \ref{alg:Algorithm-MM}, minimizing the following objective function subject to our parameter constraints: \begin{align*}
    \sum_{m=1}^M \gamma_m \|\X_{i,m} - \sum_{k=1}^K \bpi_{ik} \M_{k,m} \|_F^2.
\end{align*} 
There are a few high-level differences between Algorithms \ref{alg:Algorithm} and \ref{alg:Algorithm-MM}. First, nonnegative weights $\gamma = (\gamma_1, ..., \gamma_M)$ are used to govern the influence of each modality on the objective function. As a default recommendation, we suggest letting $\gamma_m = 1/(d_{1,m}d_{2,m})$ so that the mean-squared error of each modality is weighted equally. Second, we obtain separate initializations for each modality and then combine them by applying Algorithm \ref{alg:GoM} to the weighted average after aligning labels. Third, when iterating, the basis estimation steps naturally decouple by modality, and the membership estimation is solved via a weighted objective incorporating all modalities. Full details are presented in Section~S.3.3 of the Supplement.

As a final remark, we note that the Algorithm \ref{alg:Algorithm-MM} can be adapted to combine datasets for which in some modalities, the subjects' data is vector-valued, and in other modalities, the subjects' data is matrix-valued. Details of this adaptation are available in Section~S.3.4 of the Supplement.

\section{Theoretical Guarantees}\label{sec:StatisticalGuarantees}

\subsection{Reconstruction Error Rate}

In this subsection, we characterize the error rate of the two-step estimator defined in \eqref{eq:estimator-twostep} with respect to the reconstruction of the expectation of the dataset. As a prerequisite, we prove we can accurately estimate the pure subjects using the method of \cite{CHG24}. Firstly, we state the appropriate signal-to-noise assumption and separation condition on the true membership scores. To do so, we first define the incoherence parameters $\mu_1:= (n/K)\|\U^* \|^2_{2,\infty}$, $\mu_2:= (p/K)\|\V^*\|^2_{2,\infty}$, where $\U^*$ and $\V^*$ are from \eqref{spectral-unfolded}. 

\begin{assumption}[Signal-to-Noise Condition]
    \label{ass:snr}
    Assume that $K \ll n\wedge p$ and \begin{gather}
        \frac{c_0n}{K}\I_K\preceq \bPi^{\star\top}\bPi^\star \preceq \frac{C_0n}{K}\I_K,  \label{membership-spread} \\
        c_0s_\star^2\I_K \preceq \bTheta^{\star\top}\bTheta^\star \preceq C_0s_\star^2\I_K, \label{basis-spread} \\
        \max \left\{ \sqrt{\frac{\mu_1 K}{n}}, \sqrt{\frac{\mu_2 K}{p}} \right\} \lesssim (\log(n \vee p))^{-11}, \label{snr-incoherence} \\
        \sigma \lesssim \frac{s_\star \sqrt{n/K}}{\sqrt{n} + \sqrt{p}}. \label{snr-sub-Gaussian}
    \end{gather}
\end{assumption}

Condition (\ref{membership-spread}) is a standard assumption ensuring that the membership scores are evenly distributed throughout the simplex. This condition is satisfied with high probability in the common setting where $K$ is fixed, and the rows of $\bPi^*$ are independently and identically distributed according to a Dirichlet distribution with parameter vector $\alpha$, where $\alpha$ does not change as $n \to \infty$. Condition (\ref{basis-spread}) ensures $\bTheta^*$ scales appropriately relative to the magnitude of the noise $\sigma$ and matrix size $p$, where the relation to $\sigma$ and $p$ comes from (\ref{snr-sub-Gaussian}). A standard choice for $s_\star$ is $\sqrt{p}$, however our condition also allows for flexibility in choice of $s_\star$. Condition (\ref{snr-incoherence}) assumes the incoherence degrees are sufficiently controlled, and (\ref{snr-sub-Gaussian}) guarantees a lower bound on the signal strength of the matrix product $\bPi^*\bTheta^{*\top}$. Assumption \ref{ass:snr} is related to Assumption 3 in \cite{CHG24}. In Section \ref{sec:Simulations}, we provide examples of reasonable generative models that satisfy these conditions.
\begin{assumption}[Separation Condition]
    \label{ass:separation}
    Suppose Assumption \ref{ass:pure-subjects} holds, and denote the set of all pure subjects as $S^*$. There exists a universal constant $C$ such that \begin{gather}
        \kappa^2(\bPi^*) \sigma_1(\bPi^*) \xi_1  \leq C \min_{j \not\in S^*, i \in S^*} \| \bpi_i^* - \bpi_j^* \|_2,\label{eq:sep-constraint}\\
        \kappa^2(\bPi^*) \sigma_1(\bPi^*) \xi_1 \ll 1/\sqrt{K},
    \end{gather}
    where $\xi_1$ is defined as \begin{align*}
        \xi_1 := \frac{\sigma \sqrt{n \log(n \vee p) }}{\sigma_K^*} \sqrt{\frac{\mu_1 K}{n}} + \frac{\kappa^* \sigma^2 p}{\sigma_K^{*2}} \sqrt{\frac{\mu_1 K}{n}} + \frac{\sigma^2 (\log(n \vee p))^{12} \sqrt{n + p}}{\sigma_K^{*2}} \sqrt{\frac{\mu_2 K}{p}}.
    \end{align*}
\end{assumption}
Assumption \ref{ass:separation} is a separation condition for pure subjects. It requires the pure subjects to be sufficiently separated from non-pure subjects in membership space. This condition is satisfied if the gap $\min_{j \not\in S^*, i \in S^*} \| \bpi_i^* - \bpi_j^* \|_2$ is a sufficiently large constant, or if the gap converges to 0 at a rate depending on the terms in the statement. \eqref{eq:sep-constraint} is an additional constraint on the scaling of the gap, which is used by \cite{CHG24} to obtain uniform $l_{2,\infty}$ bounds on the error of their estimate of the membership matrix. We make this assumption to utilize the same result and as a consequence, achieve accurate estimation of the pure subjects. Proposition \ref{prop:est-Shat} formalizes our result.

\begin{proposition}\label{prop:est-Shat}
    Suppose Assumptions \ref{ass:noise}, \ref{ass:pure-subjects}, \ref{ass:snr}, and \ref{ass:separation} hold. Denote by $\hat{S}$ the ordered indices of the pure subjects identified by the method in \cite{CHG24}. With probability at least $1-O((n \vee p)^{-10})$, there exists a permutation $\sigma \in S_K$ such that $\hat{S}_k = S^*_{\sigma(k)}$ for all $k$.
    \label{prop:gomForPS}
\end{proposition}
Proposition \ref{prop:est-Shat} shows that the estimated pure-subject set $\widehat S$ equals $S^*$, up to permutation, with high probability. Consequently, the two-step estimator \eqref{eq:estimator-twostep} is equal to the oracle estimator \eqref{eq:estimator-oracle} with probability tending to one. The remaining theoretical analysis can therefore be reduced to the oracle problem.

\begin{theorem}\label{thm:main}
    Suppose Assumptions \ref{ass:noise}, \ref{ass:pure-subjects}, \ref{ass:snr}, and \ref{ass:separation} are satisfied. Then for the two-step estimator \eqref{eq:estimator-twostep} given by parameters $((\hat{\bpi_i})_{i=1}^n, (\hat{\M}_k)_{k=1}^K)$, the following holds with probability at least $1 - O((n \vee p)^{-10})$: \begin{align}
    \sum_{i=1}^n \| \EE[\X_i] - \sum_{k=1}^K \hat{\bpi}_{ik} \hat{\M}_k \|_F^2 \leq C^2\sigma^2\brac{nK+\sum_{k=1}^K r_k(d_1+d_2)}\log \brac{n\vee s_\star}.
        \label{product-rate}
	\end{align}
\end{theorem}
Theorem \ref{thm:main} provides a nonasymptotic bound for the reconstruction error of the expected matrices of all subjects. Up to a logarithmic factor, the rate is governed by the effective number of free parameters, $nK + \sum_{k=1}^K r_k(d_1+d_2)$.
The key point is that the contribution of each basis matrix depends on $d_1+d_2$, rather than on the ambient dimension \(d_1d_2\), reflecting the reduction in complexity induced by the low-rank structure. 

The following result shows that rate (\ref{product-rate}) is minimax optimal to a logarithmic factor. First, note that the term on the left-hand side of (\ref{product-rate}) can be written as $||\W^* - \hat{\W}||_F^2$, where $\W^* = \bPi^* \bTheta^{*\top}$ and $\hat{\W} = \hat\bPi \hat\bTheta^{\top}$. Consider the following parameter space:
\begin{align*}
	\calP_s:=\ebrac{\W=\bPi\bTheta^\top :\bPi\in\Delta_K^n,~\rank\big(\unvec_{d_1,d_2}(\bTheta_{\cdot,k})\big)\le r_k,~\forall k\in[K],~\sigma_{K}(\bTheta)\ge s},
\end{align*}
where $\Delta_K^n := \{ \A \in [0,1]^{n \times K}: \A_{i,:} \in \Delta_K (\forall i) \text{ and } \exists S \subset [n] \text{ with } |S|=K \text{ and } \A_{S,:} = \I_K \}$. Theorem \ref{thm:product-lb} shows that the upper bound obtained in Theorem \ref{thm:main} for the two-step estimator \eqref{eq:estimator-twostep} is minimax optimal up to the logarithmic factor:
\begin{theorem}\label{thm:product-lb}

Suppose $n\ge K$. Then there exists some universal constants $c,C>0$ such that 
	\begin{equation*}
		\inf_{\hat\W \in \calP_s}\sup_{\W^* \in \calP_s}\EE\fro{\hat\W-\W^\star}^2 \ge c\sigma^2 \left( nK+\sum_{k=1}^Kr_k(d_1+d_2) \right).
	\end{equation*}
	provided that $s\ge C\sigma K^{3/2}$.
\end{theorem}

\subsection{Error Rates for Membership and Basis Estimation}

While Theorem \ref{thm:main} is sufficient for signal reconstruction, scientific interpretation of the latent structure requires separate control of the estimation errors of $(\hat{\bpi}_i)_{i=1}^n$ and $(\hat{\M}_k)_{k=1}^K$. Such bounds are crucial for guaranteeing the estimator accurately identifies the prototypical patterns and membership scores of subjects. This turns out to be a more delicate problem, because a small product error $\|\hat{\bPi}\hat{\bTheta}^{\top} - \bPi^*\bTheta^{*\top} \|_F^2$ does not automatically rule out compensating perturbations in the two factors $\hat{\bPi}$ and $\hat{\bTheta}$.  To separate the two errors, we impose a local geometric condition on perturbations of $\bTheta$.  
% These results hold under an additional assumption on the shape of the parameter space, stated below. 

To that end, we define $ \calS_\theta^\star := \mathrm{span}\{\btheta_1^\star,\ldots,\btheta_K^\star\}\subset\RR^p$. For any matrix $\A=(a_1,\ldots,a_K)\in\RR^{p\times K}$, let $\calP_{\calS_\theta^\star}(\A)
:=\big(\calP_{\calS_\theta^\star}(a_1),\ldots,\calP_{\calS_\theta^\star}(a_K)\big)$
be the columnwise orthogonal projection of $\A$ onto $\calS_\theta^\star$. Fix $\rho>0$ and define the local parameter space
\begin{equation*}
    \calM(\rho):=\Big\{\bTheta\in\RR^{p\times K}:\ \max_{k\in[K]}\|\bTheta_{\cdot,k}-\bTheta_{\cdot,k}^\star\|_2\le \rho,\ \rank\big(\unvec_{d_1,d_2}(\bTheta_{\cdot,k})\big)\le r_k\ \forall k\in[K]\Big\}.
\end{equation*}
For any $\bTheta\in\calM(\rho)$, decompose each columnwise perturbation as
\begin{equation*}
    \btheta_k-\btheta_k^\star = d_k\btheta_k^\star + \h_k,
    \qquad
    d_k:=\frac{\langle\btheta_k-\btheta_k^\star,\btheta_k^\star\rangle}{\|\btheta_k^\star\|_2^2},
    \qquad
    \langle \h_k,\btheta_k^\star\rangle=0,
\end{equation*}
for $k\in[K]$. Let $\D(\bTheta):=\mathrm{diag}(d_1,\ldots,d_K)$ and $
\H(\bTheta):=(\h_1,\ldots,\h_K)$, then we can write
\begin{align*}
\bTheta = \bTheta^\star\big(\I_K+\D(\bTheta)\big)+\H(\bTheta),
\end{align*}
where the $k$th column of $\H(\bTheta)$ is orthogonal to $\btheta_k^\star$.

\begin{assumption}
    \label{ass:align} 
    There exists a  constant $\tau\in[0,1)$ such that 
    	\begin{align*}
    \tau_\rho:=\sup_{\bTheta\in \calM(\rho)}\frac{\fro{\calP_{\mathcal{S}_\theta^\star}( \H(\bTheta) )}}{\fro{\H(\bTheta)}}\le \tau.
    	\end{align*}
\end{assumption}

Assumption \ref{ass:align} is a local non-alignment condition. It requires that, within a neighborhood of radius $\rho$ around $\bTheta^\star$, the component of the basis-matrix error that is orthogonal to each signal column is not itself mostly contained in the global signal span $\calS_\theta^\star$. When $\tau_\rho$ is small, perturbations of $\bTheta$ cannot masquerade as perturbations of $\bPi$ so that errors in $\bTheta$ and $\bPi$ move the product $\bPi\bTheta^\top$ in genuinely different directions, which is exactly what is needed to decouple the two estimation errors.

The next lemma shows that Assumption \ref{ass:align} is mild when the latent basis matrices have sufficiently separated singular subspaces. Let
\begin{align*}
        s_{\min}:=\min_{k\in[K]}\sigma_{r_k}(\M_k^\star),
    \qquad
    \varepsilon_U:=\max_{j\neq k}\|\U_j^{\star\top}\U_k^\star\|,
    \qquad
    \varepsilon_V:=\max_{j\neq k}\|\V_j^{\star\top}\V_k^\star\|,
\end{align*}

\begin{lemma}\label{lem:tau-bound}
Assume that $\rho/s_{\min}<1/2$, then we have
\begin{align*}
	 \tau_\rho\le \sqrt{\frac{C_0K}{c_0}}\Big(\varepsilon_U\varepsilon_V+\varepsilon_U+\varepsilon_V+\frac{\rho}{s_{\min}}\Big).
\end{align*}
Moreover, when $\rho/ s_{\min}<\frac{c_0}{8C_0K^{1/2}} $ and $\max\ebrac{\varepsilon_U,\varepsilon_V}\le  \frac{c_0}{8C_0K^{1/2}}$, we have $\tau_\rho<\sqrt{c_0/(4C_0)}$.
\end{lemma}

Lemma \ref{lem:tau-bound} gives a concrete sufficient condition for Assumption \ref{ass:align}. The quantities $\varepsilon_U$ and $\varepsilon_V$ measure how much the singular subspaces of different latent basis matrices overlap. If these overlaps are small, then the latent profiles are well separated: changing one basis matrix in a direction orthogonal to itself does not closely resemble changing another basis matrix. The ratio $\rho/s_{\min}$ can be regarded as the SNR condition, ensuring that we stay in a neighborhood where the low-rank geometry is stable.

Finally, we can now state the separate consistency results for each of $(\hat{\bpi}_i)_{i=1}^n$ and $(\hat{\M}_k)_{k=1}^K$ obtained from the two-step estimator.

\begin{theorem}\label{thm:decouple-rate}
	Suppose Assumptions \ref{ass:noise}, \ref{ass:pure-subjects}, \ref{ass:snr}, \ref{ass:separation}, and \ref{ass:align} hold with $$\rho:=C\sigma\sqrt{\brac{nK+\sum_{k=1}^K r_k(d_1+d_2)}\log(n\vee s_\star)},\qquad  \tau<\sqrt{c_0/(4C_0)}.$$ In addition, assume that  
	\begin{equation}
    \label{large-network}
    \begin{gathered}
    \frac{d_1 d_2}{n}\ge C K^2 \log (n \vee s_\star), \qquad
    d_1 \wedge d_2 \ge C K \Bigl( \sum_{k=1}^K r_k \Bigr)\log (n \vee s_\star), \\
    s_\star^2 \ge C \sigma^2 \max(K^3, nK) \left( \frac{d_1+d_2}{n}\sum_{k=1}^K r_k + K \right)\log(n \vee s_\star).
    \end{gathered}
    \end{equation}
	for some constant $C>0$ depending only on $c_0$ and $C_0$. For the two-step estimator \eqref{eq:estimator-twostep} given by parameters $((\hat{\bpi_i})_{i=1}^n, (\hat{\M}_k)_{k=1}^K)$, we have with probability at least $1-(d_1d_2)^{-10}-e^{-n\vee (d_1+d_2)}$, there exists a permutation $\sigma \in S_K$ such that
	\begin{align*}
	\frac{1}{nK} \sum_{i=1}^n \sum_{k=1}^K (\hat{\bpi}_{ik} - \bpi^*_{i\sigma(k)})^2 &\le C_1^2\sigma^2\brac{\frac{K}{s_\star^2}+\frac{(d_1+d_2)\sum_{k=1}^K r_k}{n s_\star^2}}\log (n\vee s_\star),\\
	\frac{1}{d_1d_2K}\sum_{k=1}^K \bfro{\hat{\M}_k - \M^*_{\sigma(k)}}^2 &\le C_2^2 \sigma^2\brac{\frac{K^2}{d_1d_2}+\frac{(d_1+d_2)K\sum_{k=1}^K r_k}{n d_1d_2}}\log (n\vee s_\star).
\end{align*}
\end{theorem}

Theorem \ref{thm:decouple-rate} states that our proposed estimator of $(\hat{\bpi}_i)_{i=1}^n$ and $(\hat{\M}_k)_{k=1}^K$ achieves a desirable rate of consistency in terms of convergence of the mean squared error to 0. Condition (\ref{large-network}) requires the size of each subject's matrix to be sufficiently large relative to the number of observed networks and the number and size of the prototypical patterns. 

The rates in Theorem \ref{thm:decouple-rate} demonstrate the utility of the low-rankness assumption in facilitating efficient parameter estimation. This is made clear by a comparison of our rates against those obtained by \cite{CHG24} for their mixed membership model of vector-valued observations, written in Corollary 1 of their work. They present their rates in terms of $l_{2,\infty}$ norm for $\bPi^*$ and $l_\infty$ norm for $\bTheta^*$. Their rates are stronger than but also immediately imply Frobenius norm rates, which we use to compare with our Frobenius norm rate. Holding $K$, $\sigma$ constant, letting $s_\star^2 = d_1 d_2$, and provided $d_1 d_2 \gtrsim n$, the MSE of basis parameters and membership parameters for both our models and theirs scales like $O(f(n) + \frac{1}{d_1d_2})$ up to a logarithmic factor, where $f(n)$ depends on error type (basis/membership) and method. Comparison of the rates is made explicit in Table \ref{tab:comparison}. Based on these results, we can say that our exploitation of low-rank structure improves basis estimation by multiplying the $1/n$ term in the rate by factor $\left(\frac{(d_1+d_2)(\sum_k r_k)}{d_1d_2}\right)$, and it improves membership error as well as long as $d_1 \wedge d_2 \gg n/(\sum_k r_k)$, i.e. the number of rows and columns of the matrix scale fast relative to the number of observations and ranks of the basis matrices.

\begin{table}[h!]
\centering
\begin{tabular}{|c|c|c|}
\hline
$f(n)$ & Basis & Membership \\ \hline
\cite{CHG24} & $1/n$ & $1/n^2$ \\ \hline
LrMMM & $(1/n)\left(\frac{(d_1+d_2)(\sum_k r_k)}{d_1d_2}\right)$ & $(1/n)\left(\frac{(d_1+d_2)(\sum_k r_k)}{d_1d_2}\right)$ \\ \hline
\end{tabular}
\caption{Comparison of $f(n)$-components of MSE. Our model exploits low-rank structure to achieve faster estimation rates.}
\label{tab:comparison}
\end{table}

\begin{remark}
    We remark that even if \Cref{ass:align} does not hold, consistent estimation of membership and basis estimation parameters still holds as a direct corollary of \Cref{thm:main} provided a stronger condition on the size of the matrices than (\ref{large-network}) holds. However, the guaranteed rate of convergence in this case is slower. Details are provided in Section~S.3.5 of the Supplement.
\end{remark}

\section{Simulation Studies}\label{sec:Simulations}

In this section, we present a simulation demonstrating consistency, rate verification, and the benefit of multimodality under a reasonable double-asymptotic setting. Section~S.4 of the Supplement contains extensive additional simulations, including a simulated neuroscience application, further rate verification, and evidence our data-generating mechanism satisfies our assumptions. We generate data from two generative models: one using Bernoulli error and the other using Normal error. At a high level, we simulate membership vectors from a thresholded Dirichlet distribution that separates $K$ pure subjects from the rest of the subjects by a constant. For the Bernoulli data, we generate basis matrices as stochastic blockmodels with randomly generated node membership and edge probabilities, while for the Normal data, we generate basis matrices using randomly obtained singular subspaces scaled to satisfy our signal-to-noise conditions. We simulate multimodal datasets with $M=5$ modalities, and apply Algorithm \ref{alg:Algorithm} to the first modality and Algorithm \ref{alg:Algorithm-MM} to the entire multimodal dataset, and we simulate under a double asymptotic regime where $d:=d_1=d_2$, and $n,d\to \infty$ with $d=n/5$. The full details of this simulation  are presented in Section~S.4 of the Supplement.

\begin{figure}[h!]
    \centering
    \captionsetup[subfigure]{skip=0pt}
    \begin{subfigure}[t]{1.0\linewidth}
        \centering
        \includegraphics[width=\linewidth]{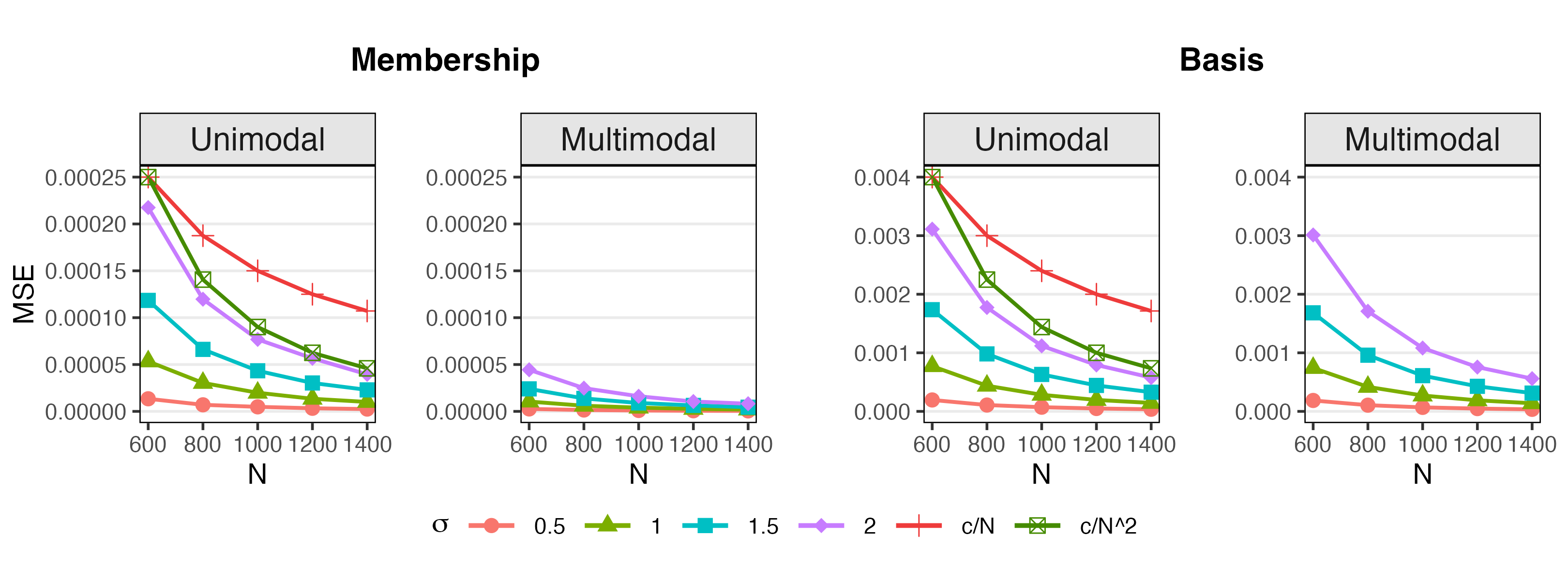}
        \subcaption{Normal}
    \end{subfigure}
    \begin{subfigure}[t]{1.0\linewidth}
        \centering
        \includegraphics[width=\linewidth]{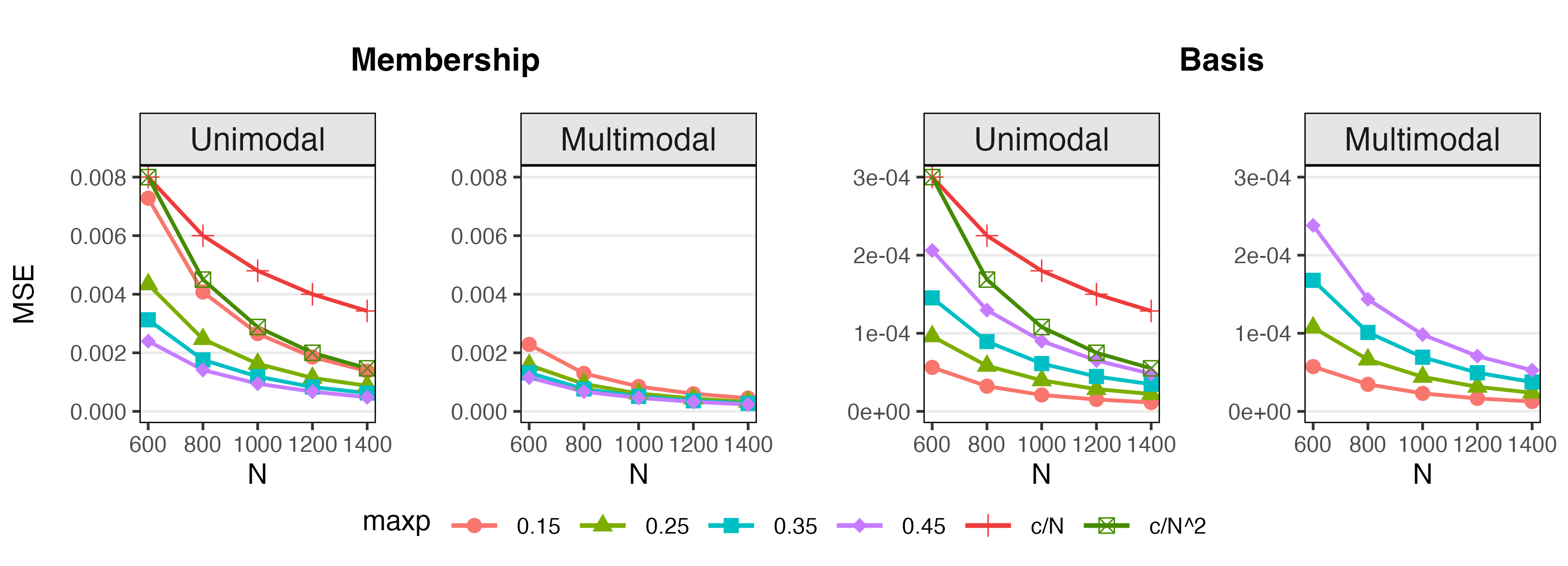}
        \subcaption{Bernoulli}
    \end{subfigure}

    \caption{Consistency and relative performance of Algorithm \ref{alg:Algorithm} (Unimodal) and Algorithm \ref{alg:Algorithm-MM} (Multimodal) in the scaling regime $d=n/5$. Unimodal plots include reference lines $c/N$ and $c/N^2$ for comparison.}
    \label{fig:Consistency}
\end{figure}

Figure \ref{fig:Consistency} summarizes the main findings. First, we observe consistency of both algorithms. Second, the multimodal Algorithm \ref{alg:Algorithm-MM} improves membership estimation relative to the unimodal Algorithm \ref{alg:Algorithm} while achieving comparable accuracy in estimating latent basis matrices, demonstrating the benefit of borrowing information across modalities. Third, estimation is generally more accurate for lower $\sigma$ (noise) and  lower sparsity (i.e., higher $\text{maxp}$), except for basis estimation in the Bernoulli case, where this pattern is reversed. Finally, Algorithm \ref{alg:Algorithm} appears to scale at a rate of $O(n^{-2})$, which aligns with the rates obtained for the two-step estimator in Theorem \ref{thm:decouple-rate} in our double-asymptotic setting, modulo a log factor. Additional simulations in the Supplement further verify these convergence rates under varying matrix dimensions and sample sizes.

\section{Application to Human Connectome Project Data}
\label{sec:DataAnalysis}

In this section, we apply our method to the HCP dataset. We first demonstrate that our method is unique among benchmark methods in achieving both biologically interpretable brain connectivity networks and validated subject-level scores simultaneously. We then illustrate how these two results together facilitate discovery of the relationships between the connectivity networks and cognition. After following the data preprocessing steps described in Section~S.5 of the Supplement, our dataset consists of binary functional connectivity matrices of $n=913$ subjects on $d=264$ ``regions of interest" (ROIs) in the brain, with each subject providing $M=8$ connectivity matrices, and each matrix corresponding to a unique task performed by the subject during data collection; we also have $W=24$ cognitive covariates for each subject. Since each of the $8$ cognitive tasks engages partially overlapping but distinct functional networks, treating each task as a separate modality allows the proposed multimodal model to distinguish stable subject-specific connectivity behavior from task-specific modulation while borrowing strength across tasks.

\subsection{Model Comparison}\label{sec:ModelComparison}

We compare Algorithms \ref{alg:Algorithm} (unimodal LrMMM or ``uLrMMM") and \ref{alg:Algorithm-MM} (multimodal mLrMMM or ``mLrMMM") with the Low-rank Lloyd (``LrLloyd") clustering algorithm of \cite{LX25}, the Tensor Mixed Membership Blockmodel (``TMMBM") of \cite{AZ25}, the Semi-symmetric Tensor PCA method (``SS-TPCA") of \cite{Weylandt2025SS-TPCA}, and the multimodal Generalized Joint Semi-symmetric Tensor PCA (``gJisstPCA") algorithm of \cite{Liu2023joint}. We fit each model with $K=3$ and $r_{k,m} = r = 3$ based on a scree plot model selection procedure; as an exception we use $K=2$ for the Tensor PCA methods SS-TPCA and gJisstPCA since a) this was the number of components chosen in the analysis of the HCP dataset in \citep{Liu2023joint}, and b) so as to compare embeddings that use the same number of parameters. Full details of our model selection procedure and model comparison procedure are in Section~S.5 of the Supplement.

\begin{figure}[htbp]
    \centering
    \makebox[\textwidth][c]{%
    \hspace*{-0.1\textwidth}%
    \begin{subfigure}[b]{0.60\textwidth}
      \centering
      \includegraphics[height=6.0cm]{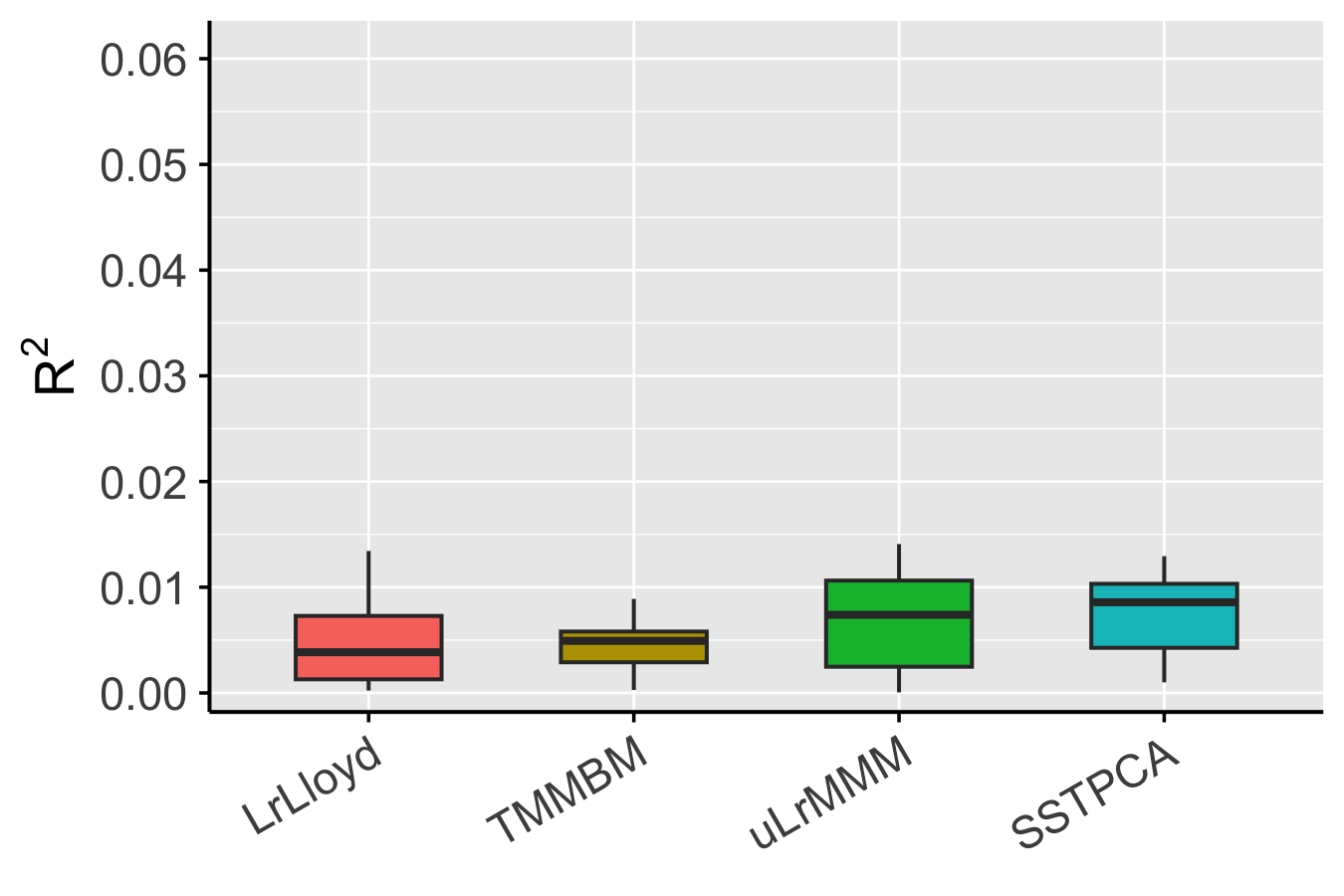}
      \subcaption{Unimodal Methods}
    \end{subfigure}%
    \hspace{0.03\textwidth}
    \begin{subfigure}[b]{0.32\textwidth}
      \centering
      \includegraphics[height=6.0cm]{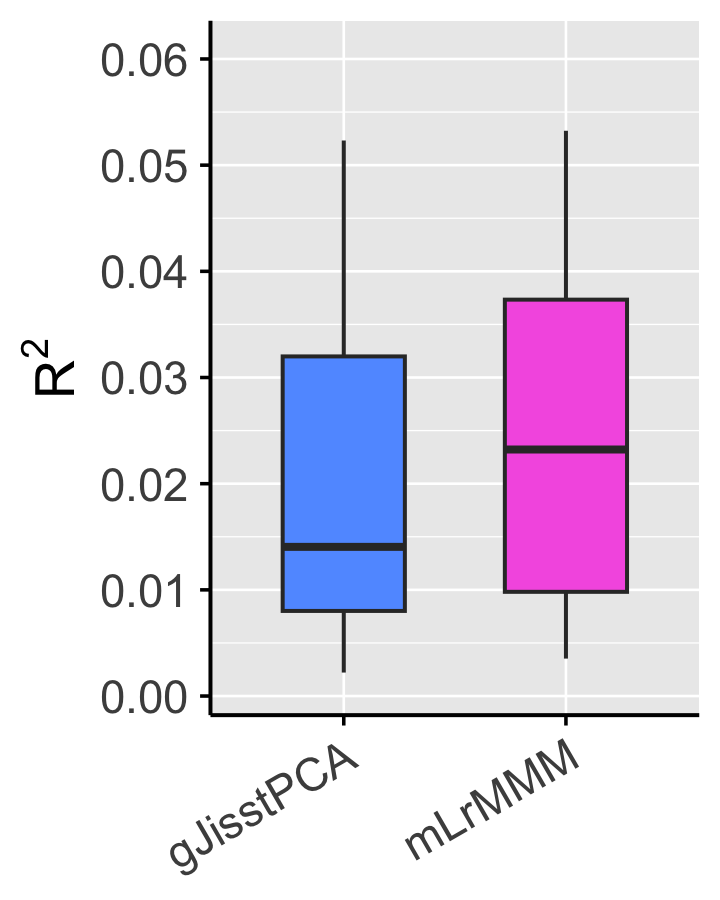}
      \subcaption{Multimodal Methods}
    \end{subfigure}%
    }

    \caption{$R^2$ values of $W=24$ phenotypic covariates regressed on embedding}
    \label{fig:RSquaredSeparate}
\end{figure}

We first assess the validation of a method's fit in terms of the relationship of its estimated subject-level latent embeddings with the auxiliary cognitive covariates. We measure the strength of this relationship in two ways: $R^2$ plots from OLS regressing covariate on embedding as pictured in Figure \ref{fig:RSquaredSeparate}, and prediction results on unseen data obtained by using these unsupervised methods as part of a supervised learning pipeline, with these results presented in Figure \ref{fig:PredictionError}. Considering the $R^2$ plots first, we observe that the multimodal methods have substantially higher $R^2$ values, illustrating the benefit of multimodality in estimating validated latent structures. We also observe that, among multimodal methods, mLrMMM achieves the highest $R^2$ values, and among unimodal methods, uLrMMM's $R^2$ values are among the highest. These results suggest our method is competitive at estimating validated latent structures. 

To obtain prediction results, we start by randomly dividing the dataset into 10 folds. For each fold, we regard that fold as the ``test set" and the remaining folds as the ``training set." We apply the method to the training set, fit an OLS model regressing covariate on embedding using the training set, project the observations in the test set onto the membership/score space of the training set, and calculate the squared error of the OLS model for each observation in the test set (details in Section~S.5 of the Supplement). When analyzing results, we account for the difference in scale between covariates by comparing the methods in terms of \textit{relative mean squared error} for each covariate, by which we refer to the mean squared error of one method divided by the mean squared error of a baseline method, which we choose to be TMMBM, along with bootstrap confidence intervals. Based on Figure \ref{fig:PredictionError}, the multimodal methods, and in particular mLrMMM, generally achieve the lowest relative MSE, while performance of unimodal methods is generally comparable. These results further indicate the competitiveness of our methods in estimating validated latent structures. However, our method's relevance for scientific discovery lies not solely in its attainment of validated embeddings, but rather in its simultaneous attainment of validated embeddings alongside interpretable basis matrices representing population-level patterns. Our method clearly outperforms other methods in this regard. 

\begin{figure}
    \centering
    \includegraphics[width=1.0\linewidth]{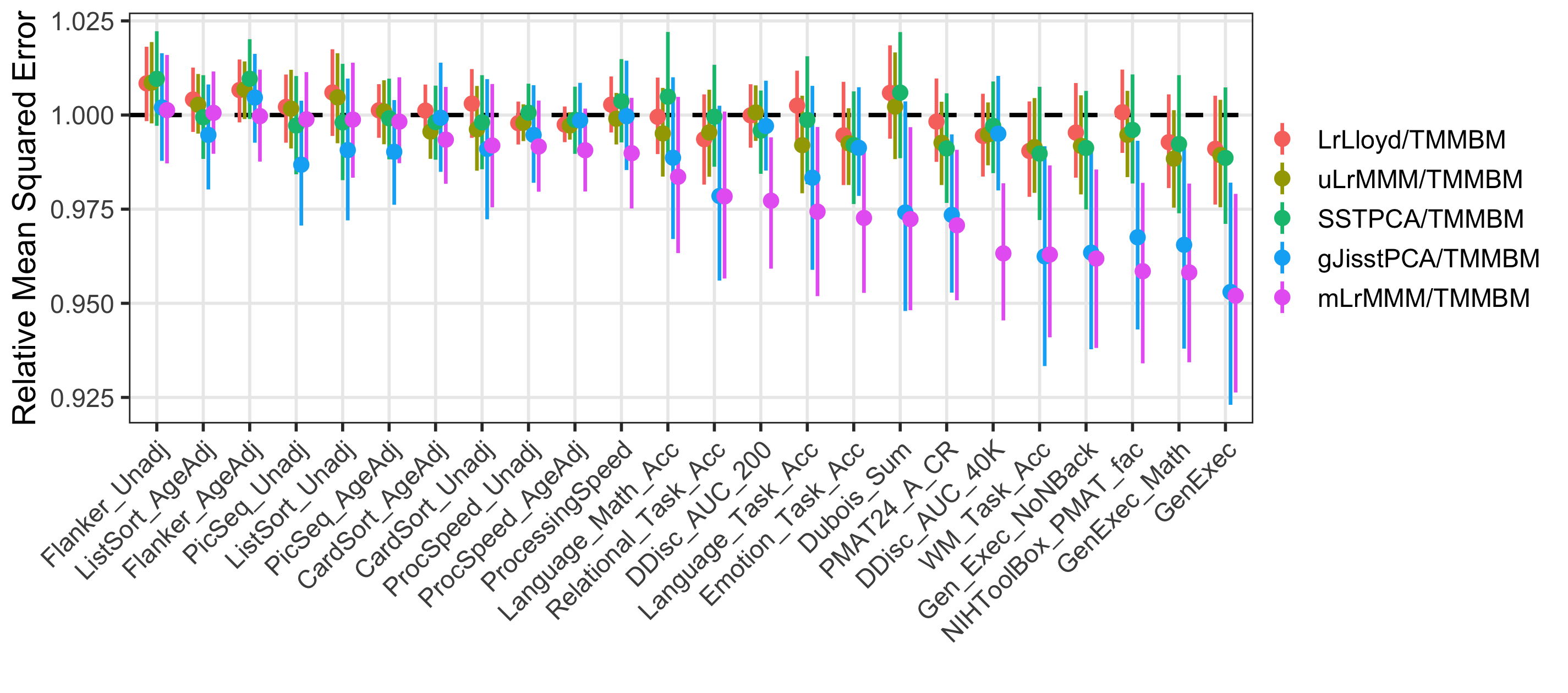}
    \caption{Relative prediction error of OLS model on unseen data}
    \label{fig:PredictionError}
\end{figure}

In Figure \ref{fig:BasisComparison}, we compare the methods in terms of the interpretability of the estimated matrix-valued patterns. Alongside the rows and columns of each basis matrix, we attach color bars that represent different functional units of the brain as identified by domain literature \citep{Power11}. As a first remark, we note that the structures identified by all methods appear to align with the functional units from domain literature. However, the second and third factor loading matrices identified by gJisstPCA lack structure and appear ``ghostly." This is an artifact of the method's goal of additive reconstruction of the dataset. The first loading matrix identifies meaningful population-level structure, which aligns with the structures identified by the other methods, while the second and third loading matrices fit to residual structure which lacks interpretability. In contrast, the emphasis of mLrMMM, LrLloyd, and TMMBM on identifying population-level patterns facilitates meaningful comparisons between each method's estimated basis matrices. For instance, in the mLrMMM results, we observe that there is higher functional connectivity in the somatomotor, cingular-opercular, and auditory regions in the second basis matrix than in the first and third. If we find that high membership in the second extreme profile is significantly associated with high levels of some cognitive covariate, then we can say that high functional connectivity in these regions is associated with high levels of that cognitive feature. 

The LrLloyd and TMMBM algorithms also appear to identify interpretable population-level structure. Moreover, the basis matrices estimated by TMMBM are highly structured. This is a useful consequence of TMMBM's assumption of shared node-wise membership structure across basis matrices, as noted in Section \ref{sec:Model-Intro}. This assumption is reasonable for this dataset since ROIs belonging to the same functional unit of the brain tend to share behavior. Nevertheless, mLrMMM is more useful than these other methods for scientific discovery since, as discussed previously, mLrMMM is much more strongly validated than either of these methods in terms of its embedding's relationship to auxiliary phenotypic covariates, as shown by Figures \ref{fig:RSquaredSeparate} and \ref{fig:PredictionError}. This difference is likely due to the multimodality of the method, and it plays a critical role in interpretation of the estimated latent structures, since it can mean the difference between identifying or not identifying a significant relationship between a population-level pattern and a trend in the auxiliary covariates. Other advantages of our model include its flexibility and generality: relative to TMMBM, our method also accommodates basis matrices that lack shared node-wise membership structure, and in cases where such shared structure is present, it facilitates observation of differential behavior within each community of nodes. And relative to LrLloyd, mLrMMM estimates a continuous embedding, thus storing more information per parameter for downstream regression; furthermore, it estimates population extremes rather than cluster means, isolating patterns in their clearest form.
Overall, the advantage of mLrMMM is as follows: of all methods considered, only mLrMMM obtains both highly validated continuous embeddings and interpretable population-level patterns in their clearest form. The attainment of both features is the key ingredient for scientific discovery, as demonstrated in the next section.

\begin{figure}[htbp]
    \centering
    \begin{subfigure}[c]{0.72\textwidth}
      \centering
      \includegraphics[width=\linewidth]{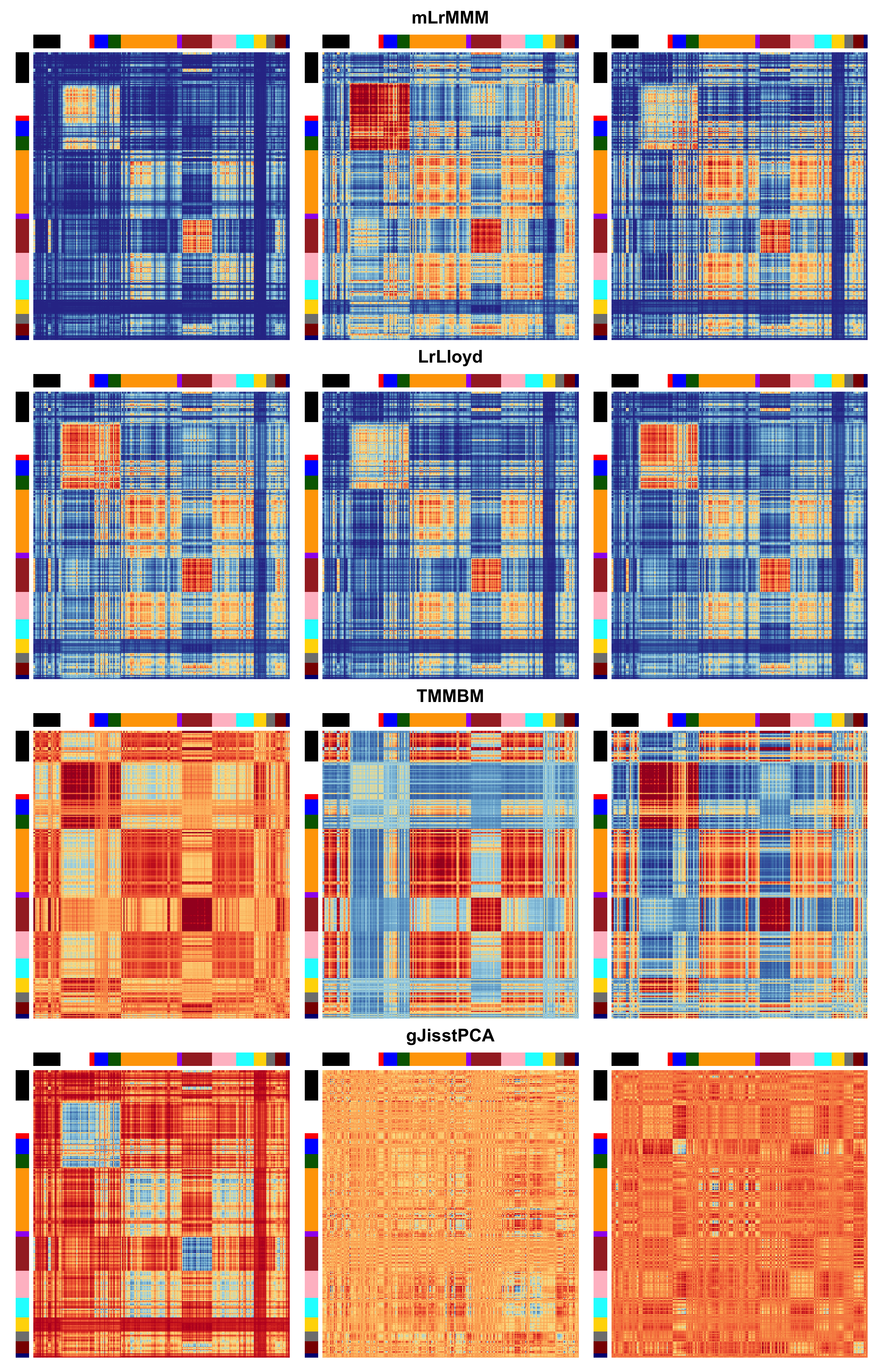}
    \end{subfigure}%
    \begin{subfigure}[c]{0.14\textwidth}
      \centering
      \includegraphics[width=\linewidth]{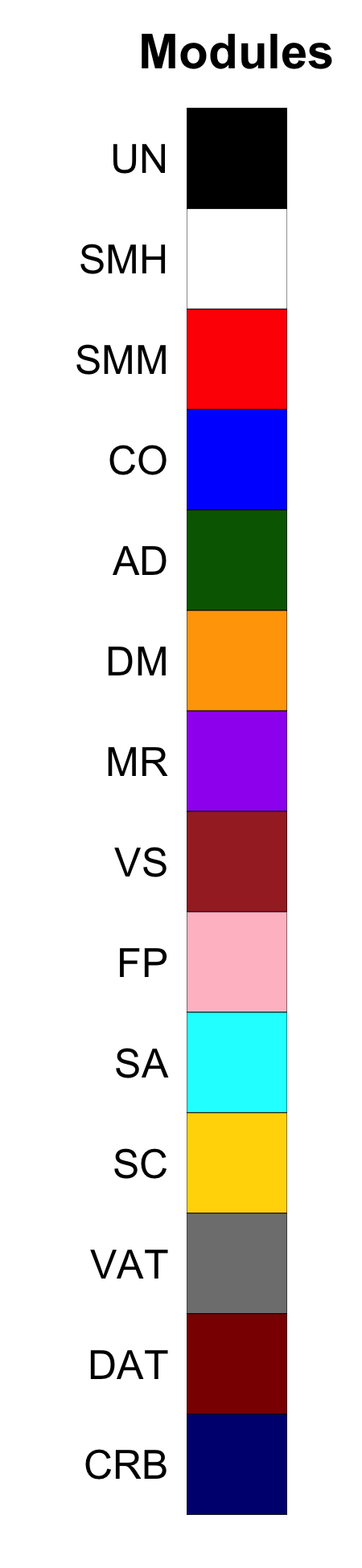}
    \end{subfigure}
    
    \caption{HCP: Prototypical matrix-valued patterns in first modality estimated by mLrMMM (ours), LrLloyd, and gJisstPCA. Here, three basis matrices are pictured from gJisstPCA since they are conveniently obtained; fitting gJisstPCA with $K=2$ yields the first two matrices in that row. Warmer color indicates higher estimated value. Results are plotted alongside functional units from domain literature, classified as follows: UN=Uncertain, SMH=Sensory/Somatomotor Hand, SMM=Sensor/Somatomotor Mouth, CO=Cingulo-opercular Task Control, AD=Auditory, DM=Default Mode, MR=Memory Retrieval, VS=Visual, FP=Fronto-parietal Task control, SA=Salience, SC=Subcortical, VAT=Ventral Attention, DAT=Dorsal Attention, CRB=Cerebellar.}
    \label{fig:BasisComparison}
\end{figure}

\begin{figure}[htbp]
    \centering
    \includegraphics[width=0.9\linewidth]{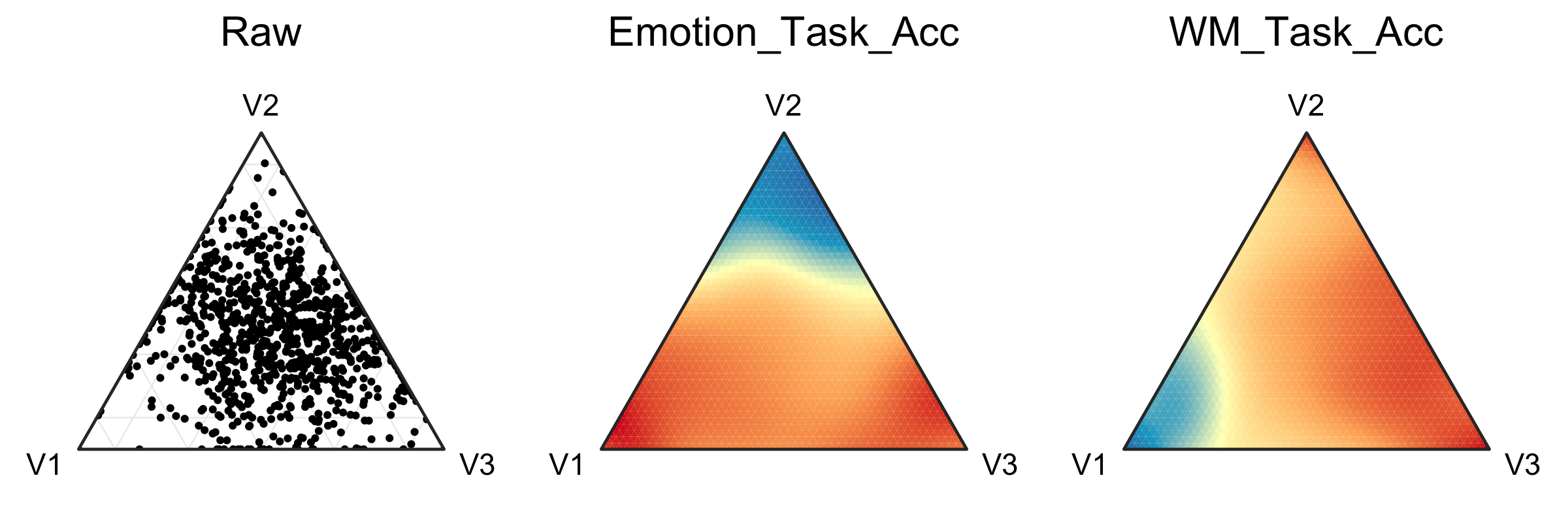}
    \caption{Membership scores of HCP participants with kernel regression results. Warmer color indicates higher value.}
    \label{fig:ternary-three}
\end{figure}

\subsection{Findings of Scientific Interest}\label{sec:findings}

Among the three prototypical connectivity profiles estimated by LrMMM (pictured in the first row of Figure \ref{fig:BasisComparison}), profile 1 represents a baseline connectivity organization with relatively weaker between-network connectivity. Profile 3 exhibits stronger interactions among the default mode, frontoparietal, attention, and cerebellar networks, and profile 2 is characterized by increased connectivity within these networks and additionally within the sensorimotor, cingulo-opercular, and auditory networks. Because each participant is represented by a convex combination of these profiles, the estimated memberships quantify the extent to which an individual’s brain organization resembles each connectivity archetype. 

We now uncover the relationship between these brain architectures and cognition by regressing each available auxiliary covariate on the estimated membership scores. We first visualize these regression results using Nadaraya-Watson kernel regression with Gaussian kernel, and also perform ordinary least squares with FDR-corrected p-values to assess the significance of these relationships. Figure \ref{fig:ternary-three} depicts the raw membership scores alongside kernel regression results for two covariates representing performance on emotional processing and working memory tests. We observe that higher membership in the second extreme profile is associated with worse performance on the emotional processing task, and OLS confirms the significance of this relationship (FDR $p<0.001$). Hence, we conclude that higher functional connectivity between the sensorimotor, cingulo-opercular, and auditory networks is significantly associated with worse emotional processing skills.

Next, we observe that higher membership in the second and third extreme profile is associated with improved performance on the working memory task, and OLS confirms the significance of these relationships (FDR $p=0.001$, FDR $p<0.001$ respectively). Hence, we associate higher functional connectivity between the default, fronto-parietal, salience, ventral attention, dorsal attention, and cerebellar modes with higher working memory skills. The same pattern identified for working memory is also present at significant levels for measurements of mathematical skills, language skills, and general executive functioning. Hence, we associate higher levels of functional connectivity between these regions with higher cognitive skills in a general sense. Complete regression results and additional details of our analysis are in Section~S.5 of the Supplement.

\section{Discussion}

In this work, we have proposed a methodology that facilitates discovery of meaningful latent structures in matrix-valued data. The advantages of our method lie in its ability to simultaneously extract population-level latent patterns alongside a highly interpretable continuous embedding, a recipe for productive scientific interpretation of estimated latent structure. We have also rigorously studied our methodology in terms of identifiability and estimation rates, quantifying the improvement in estimation obtained from our model's low-rank structure relevant to the closest alternative mixed membership model.

Our current theory is centered on the two-step estimator, which already captures the main statistical difficulty of the problem. Extending the analysis to the alternating refinement in Algorithm~\ref{alg:Algorithm} appears to be substantially more delicate, as it would require a one-step contraction bound for the iterative map. This is harder in the mixed-membership setting than in discrete clustering in \cite{LX25} since the membership update is continuous rather than a hard assignment, and small perturbations in the current basis estimates can still change the updated memberships and hence propagate across iterations. Another related piece is the algorithm in \cite{AZ25}, which is more amenable to a direct perturbation analysis. However, its guarantee is based on HOOI followed by simplex recovery rather than on controlling the error in an iterative refinement with simplex constraint. A full analysis of our current Algorithm~\ref{alg:Algorithm} is therefore left for future work.

\spacingset{1}
\paragraph{Supplementary Material.}
The Supplementary Material contains all proofs of the theoretical results and additional numerical results.

\vspace{-3mm}

\paragraph{Data Availability Statement.} The Human Connectome Project data analyzed in this study are available through the HCP Young Adult ConnectomeDB, subject to registration and acceptance of the HCP Data Use Terms. Simulation data are generated from the models described in the article and Supplementary Material.

\vspace{-3mm}

\paragraph{Acknowledgements.}
Data were provided by the Human Connectome Project, WU-Minn Consortium (Principal Investigators: David Van Essen and Kamil Ugurbil; 1U54MH091657) funded by the 16 NIH Institutes and Centers that support the NIH Blueprint for Neuroscience Research; and by the McDonnell Center for Systems Neuroscience at Washington University. The authors acknowledge Yuhan Geng for support in preprocessing the data.

\setlength{\bibsep}{9pt}
\bibliographystyle{apalike}
\bibliography{ref.bib}

\clearpage
\input{supplement}

\end{document}

%% file: supplement.tex
\renewcommand{\thesection}{S.\arabic{section}}
\renewcommand{\thesubsection}{S.\arabic{section}.\arabic{subsection}}
\renewcommand{\thesubsubsection}{S.\arabic{section}.\arabic{subsection}.\arabic{subsubsection}}
\renewcommand{\theequation}{S.\arabic{equation}}
\renewcommand{\thefigure}{S.\arabic{figure}}
\renewcommand{\thetable}{S.\arabic{table}}
\renewcommand{\thealgorithm}{S.\arabic{algorithm}}
\renewcommand{\thetheorem}{S.\arabic{theorem}}
\renewcommand{\thelemma}{S.\arabic{lemma}}
\renewcommand{\theproposition}{S.\arabic{proposition}}
\renewcommand{\thecorollary}{S.\arabic{corollary}}
\renewcommand{\thedefinition}{S.\arabic{definition}}
\renewcommand{\theremark}{S.\arabic{remark}}
\renewcommand{\theassumption}{S.\arabic{assumption}}

\setcounter{section}{-1}
\setcounter{subsection}{0}
\setcounter{subsubsection}{0}
\setcounter{equation}{0}
\setcounter{figure}{0}
\setcounter{table}{0}
\setcounter{algorithm}{0}
\setcounter{theorem}{0}
\setcounter{lemma}{0}
\setcounter{proposition}{0}
\setcounter{corollary}{0}
\setcounter{definition}{0}
\setcounter{remark}{0}
\setcounter{assumption}{0}

%%%%%%%%%%%%%%%%%%%%%%%%%%%%%%%%

\clearpage
\spacingset{1}

\begin{center}
{\Large\bfseries
Supplement to ``Mixed Membership Model of Low-rank Matrices with Multimodal Extension''}
\end{center}

\vspace{1em}

\spacingset{1.7}
\section{Overview}

This Supplement contains proofs of theoretical results alongside supplementary numerical results and additional details for the content in ``Mixed Membership Model of Low-rank Matrices with Multimodal Extension". Section \ref{sec:ProofsMainResults} contains proofs of the main theorems, Section \ref{sec:ProofsLemmas} contains proofs of lemmas, Section \ref{sec:AdditionalTheory} contains additional theoretical and algorithmic details, Section \ref{sec:AdditionalSimulations} contains additional simulations and details regarding simulations from the main text, and Section \ref{sec:AdditionalData} contains additional details about our real data analysis. 

\section{Proofs of Main Results}\label{sec:ProofsMainResults}
\subsection[Proof of main reconstruction theorem]{Proof of \Cref{thm:main}}\label{pf-thm:main}
Note that $\PP(\hat{S}=S^\star)\ge 1-O((n\vee p)^{-10})$ holds by \Cref{prop:est-Shat}. It suffices for us to proceed with our analysis on the event $\{ \hat{S} = S^\star \}$.
By global optimality of $(\hat\bPi,\hat\bTheta)$ and feasibility of
$(\bPi^\star,\bTheta^\star)$, we have
\begin{align*}
	\bfro{\Y-\hat\W}^2 \le \bfro{\Y-\W^\star}^2,
\end{align*}
where $\hat\W := \hat\bPi \hat\bTheta$ and $\W^\star := \bPi^\star \bTheta^\star$.
Expanding with $\Y=\W^\star+\Z$ yields 
\begin{equation}\label{eq:basic-ineq-thm1}
\fro{\bDelta}^2 \le 2\langle \Z,\bDelta\rangle,
\end{equation}
where $\bDelta := \hat\W - \W^\star$.
Define the parameter space for $\W$
$$\mathcal{W}_S:=\{\bPi\bTheta^\top:\ \bPi\in\Delta_K^n,\ \bPi_{S,:}=\I_K,\ \mathrm{rank}(\mathrm{mat}_{d_1,d_2}(\bTheta_{\cdot,k}))\le r_k,\ \forall k\in[K]\},$$
and  simply write $\calW\equiv \calW_{S^\star}$. By definition, we have  $\W^\star\in\mathcal{W}$ and $\hat\W\in\mathcal{W}$. For any fixed $\zeta>0$, define $R(\zeta):=\fro{\bTheta^\star}+\zeta$ and
\begin{align*}
	\calQ(\zeta)&:=\ebrac{\bTheta\in\RR^{ d_1d_2 \times K}:\ \fro{\bTheta}\le R(\zeta),\ \rank(\mathrm{mat}_{d_1,d_2}(\bTheta_{\cdot,k}))\le r_k,\ \forall k\in[K]}\\
	\calM_k(\zeta)&:=\ebrac{\M\in\RR^{d_1\times d_2}:\rank(\M)\le r_k,~ \fro{\M}\le R(\zeta)}.
\end{align*}
Define the local parameter space for  $\W$:
\begin{align*}
	\calW(\zeta) := \Big\{  \bPi\bTheta^\top:\ \bPi\in\Delta_K^n,\ \bPi_{S,:}=\I_K,\ \bTheta\in\calQ(\zeta) \Big\}.
\end{align*}
and the localized difference class
\begin{equation}\label{eq:Tzeta-def}
\mathcal{T}(\zeta) := \Big\{ \W-\W^\star:\ \W\in\calW(\zeta),\ \fro{\W-\W^\star}\le \zeta \Big\}.
\end{equation}
\begin{lemma}
	There exists a universal constant $C>0$ such that for any $\epsilon\in(0,1)$ and $\zeta>0$, 
	\begin{align*}
		\log N\brac{\epsilon, \calT(\zeta),\fro{\cdot}}&\le \sqbrac{nK+\sum_{k=1}^Kr_k(d_1+d_2)}\log\frac{CR(\zeta)\sqrt{nK}}{\epsilon}.
	\end{align*}
\end{lemma}
\begin{proof}

Note that for any $\W=\bPi\bTheta^\top\in\calW$ with $\fro{\W-\W^\star}\le \zeta$,
\begin{align}
        \bfro{\bTheta - \bTheta^\star } \leq  \bfro{\bPi\bTheta^\top - \bPi^\star\bTheta^{\star\top} }\le \zeta, 
\end{align}
due to ${\bPi}_{S^\star,:} = \bPi^\star_{S^\star,:} = \I_K$. We thus arrive at $\fro{\bTheta}\le R(\zeta)$.  
By standard covering number bound for low-rank matrices (e.g.,\cite{zhang2018tensor}), 
\begin{align*}
	\log N\brac{\epsilon, \calQ(\zeta),\fro{\cdot}}\le \sum_{k=1}^K\log N\brac{\frac{\epsilon}{\sqrt{K}}, \calM_k(\zeta),\fro{\cdot}}\le \sum_{k=1}^Kr_k(d_1+d_2)\log\frac{CR(\zeta)\sqrt{K}}{\epsilon}.
\end{align*}
Note that for any $\W=\bPi\bTheta^\top\in\calW $, $\W^\prime=\bPi^\prime\bTheta^{\prime\top}\in\calW$,
\begin{align*}
	\bfro{\bPi\bTheta^\top-\bPi^\prime\bTheta^{\prime\top}}&\le 	\bfro{\bPi-\bPi^\prime}\fro{\bTheta}+\bfro{\bTheta-\bTheta^\prime}\fro{\bPi^\prime}\\
	&\le R(\zeta)\bfro{\bPi-\bPi^\prime}+\sqrt{n}\bfro{\bTheta-\bTheta^\prime}.
\end{align*}
Let $\calN_\Pi$ be an $\epsilon_\Pi$-net for $(\Delta_K^n, \fro{\cdot})$ and $\calN_\Theta$ be an $\epsilon_\Theta$-net for $(\calQ(\zeta), \fro{\cdot})$. Define 
\begin{align*}
	\calN_W:=\ebrac{\tilde\bPi\tilde\bTheta^\top:\tilde\bPi\in\calN_\Pi,\ \tilde\bTheta\in\calN_\Theta}.
\end{align*}
We claim that $\calN_W$ is a $\epsilon_W$-net for $\calW(\zeta)$ with $\epsilon_W: = R(\zeta)\epsilon_\Pi  + \sqrt{n}\epsilon_\Theta$. To see this, for any $\W=\bPi\bTheta^\top\in\calW(\zeta)$, there exist $\tilde\bPi\in\calN_\Pi$ and $\tilde\bTheta\in\calN_\Theta$ such that $\bfro{\bPi-\tilde\bPi}\le \epsilon_\Pi$ and $\bfro{\bTheta-\tilde\bTheta}\le \epsilon_\Theta$. Then, $\tilde\W:=\tilde\bPi\tilde\bTheta^\top\in\calN_W$ satisfies
\begin{align*}
	\bfro{\W-\tilde\W}\le \bfro{\bPi-\tilde\bPi}\fro{\bTheta} + \bfro{\bTheta-\tilde\bTheta}\bfro{\tilde\bPi}\le  R(\zeta) \epsilon_\Pi+ \sqrt{n}\epsilon_\Theta=\epsilon_W.
\end{align*}
Moreover, we have $\ab{\calN_W}\le \ab{\calN_\Pi}\ab{\calN_\Theta}$ by definition. Thus, we can choose 
\begin{align*}
	\epsilon_{\Pi}=\frac{\epsilon}{2R(\zeta)},\qquad \epsilon_{\Theta}=\frac{\epsilon}{2\sqrt{n}},
\end{align*}
which yields the following bound:
\begin{align*}
	\log N\brac{\epsilon, \calW(\zeta),\fro{\cdot}}&\le \log N\brac{\frac{\epsilon}{2R(\zeta)}, \Delta_K^n,\fro{\cdot}}+\log N\brac{\frac{\epsilon}{2\sqrt{n}}, \calQ(\zeta),\fro{\cdot}}\\
	&\le nK\log\frac{CR(\zeta)}{\epsilon} +\sum_{k=1}^Kr_k(d_1+d_2)\log\frac{CR(\zeta)\sqrt{nK}}{\epsilon}\\
	&\le \sqbrac{nK+\sum_{k=1}^Kr_k(d_1+d_2)}\log\frac{CR(\zeta)\sqrt{nK}}{\epsilon}.
\end{align*}
The proof is completed by noticing that
\begin{align*}
	N\brac{\epsilon, \calT(\zeta),\fro{\cdot}}&=N\brac{\epsilon, \ebrac{ \W:\ \W\in\calW(\zeta),\ \fro{\W-\W^\star}\le \zeta },\fro{\cdot}}\\
	&\le N\brac{\epsilon, \calW(\zeta),\fro{\cdot}}.
\end{align*}
\end{proof}
We will use the following localized empirical-process bound. 
\begin{lemma}[Localized empirical process bound]\label{lem:local-ep}
Let $Z$ have i.i.d.  mean-zero sub-Gaussian entries with
$\|Z_{ij}\|_{\psi_2}\le\sigma$ . Define $\mathcal{T}(\zeta)$ as in~\eqref{eq:Tzeta-def} and set $D:=nK+\sum_{k=1}^Kr_k(d_1+d_2)$.
Then there exist universal constants $c,C>0$ such that for any $t\ge 0$, with probability at least
$1-e^{-t}$, for every $\zeta>0$,
\begin{equation}\label{eq:local-ep}
\sup_{\H\in \mathcal{T}(\zeta)} \langle \Z,\H\rangle \le
C\sigma\zeta\brac{\sqrt{D\log\frac{nR(\zeta)}{\zeta }}+ \sqrt{t}}.
\end{equation}
\end{lemma}
\begin{proof}
	Define the process $X_\H:=\langle \Z,\H\rangle$ for $\H\in\calT(\zeta)$. For any $\H_1,\H_2\in\calT(\zeta)$, we have
\begin{align*}
	\|X_{\H_1}-X_{\H_2}\|_{\psi_2}&=\|\langle \Z,\H_1-\H_2\rangle\|_{\psi_2}\le C\sigma\|\H_1-\H_2\|_{\rm{F}}.
\end{align*}
Thus, $X_\H$ is a sub-Gaussian process with respect to the metric $d(\H_1,\H_2):=\|\H_1-\H_2\|_{\rm{F}}$. By Theorem 3.2 in \cite{dirksen2015tail}, we get for any $t\ge 0$, with probability at least $1-e^{-t}$,
\begin{align*}
\sup_{\H\in \mathcal{T}(\zeta)}\ab{X_\H} &\le C\sigma\brac{\int_0^{\infty}\sqrt{\log N\brac{\epsilon, \mathcal{T}(\zeta),\fro{\cdot}}}d\epsilon + \mathrm{diam}(\mathcal{T}(\zeta))\sqrt{t}},
\end{align*} 
where we have used the fact that $\mathbf{0}\in T(\zeta)$ and a bound for the $\gamma_2$-functional in terms of the covering number (e.g., \cite{talagrand2005generic}). Moreover, by definintion we have $\mathrm{diam}(\mathcal{T}(\zeta))\le 2\zeta$. Thus, the tail bound in~\eqref{eq:local-ep} follows from the following bound for the entropy integral:
\begin{align*}
	\int_0^{\infty}\sqrt{\log N\brac{\epsilon, \mathcal{T}(\zeta),\fro{\cdot}}}d\epsilon&\le C\int_0^{2\zeta}\sqrt{D\log\frac{CR(\zeta)\sqrt{nK}}{\epsilon}}d\epsilon\\
&\lesssim\sigma\zeta\sqrt{D\log\frac{nR(\zeta)}{\zeta }}.
\end{align*}
\end{proof}
We now conclude the proof of Theorem~\ref{thm:main}. Fix $t_0:=n\vee(d_1+d_2)$ and define
$\zeta_\ell := 2^\ell \zeta_0$ for $\ell\ge 0$, where $\zeta_0:=8C\sigma(\sqrt{\bar{D}}+\sqrt{t_0})\ge 1$ with  $\bar{D}:=D\brac{\log n+\log\brac{1+\frac{\fro{\bTheta^\star}}{\sigma \sqrt{D}}}}$, and $C$ is the universal constant in Lemma~\ref{lem:local-ep}.
For each $\ell\ge 0$, apply Lemma~\ref{lem:local-ep} with $\zeta=\zeta_\ell$ and tail parameter $t_0+\ell$ and  denote 
$$\calE_\ell:=\ebrac{\sup_{\H\in \mathcal{T}(\zeta_\ell)} \langle \Z,\H\rangle \le
C\sigma\zeta_\ell\brac{\sqrt{D\log\frac{nR(\zeta_\ell)}{\zeta_\ell}} + \sqrt{t_0+\ell }}}.$$ We thus have $\PP\brac{\calE_\ell^c}\le e^{-(t_0+\ell)}$, and  by a union bound over $\ell=0,1,2,\ldots$, we get
\begin{align*}
	\sum_{\ell=0}^\infty\PP\brac{\calE_\ell^c}\le \sum_{\ell=0}^\infty e^{-(t_0+\ell)}\le \frac{e^{-t_0}}{1-e^{-1}}\le 2e^{-t_0}.
\end{align*}
Therefore, with probability at least $1-2e^{-t_0}$ we have 
for all $\ell\ge 0$,
\begin{equation}\label{eq:union-shell}
\sup_{\H\in \mathcal{T}(\zeta_\ell)} \langle \Z,\H\rangle \le
C\sigma\, \zeta_\ell\brac{\sqrt{D\log\frac{nR(\zeta_\ell)}{\zeta_\ell}}+\sqrt{t_0+\ell}}.
\end{equation}
On the event~\eqref{eq:union-shell}, or equivalently, $\calE:=\bigcap_{\ell=0}^\infty\calE_\ell$, suppose for contradiction that $\fro{\bDelta}>\zeta_0$.
Choose $\ell\ge 1$ such that $\zeta_{\ell-1}<\fro{\bDelta}\le \zeta_\ell$. Then $\bDelta\in \mathcal{T}(\zeta_\ell)$
by definition of $\mathcal{T}(\zeta)$, and combining~\eqref{eq:basic-ineq-thm1} with~\eqref{eq:union-shell} gives
\begin{align*}
	\fro{\bDelta}^2
&\le 2\langle \Z,\bDelta\rangle
\le 2\sup_{\H\in\mathcal{T}(\zeta_\ell)}\langle \Z,\H\rangle
\le 2C\sigma\zeta_\ell\brac{\sqrt{D\log\frac{nR(\zeta_\ell)}{\zeta_\ell}}+\sqrt{t_0+\ell}}\\
&\le 4C\sigma\zeta_{\ell-1}\brac{\sqrt{D\log\frac{nR(\zeta_{\ell})}{\zeta_{\ell}}}+\sqrt{t_0+\ell}}\le 4C\sigma\fro{\bDelta}\brac{\sqrt{D\log\frac{nR(\zeta_{\ell})}{\zeta_{\ell}}}+\sqrt{t_0+\ell}},
\end{align*}
where we have used the definition of $\zeta_{\ell}$ and that $\fro{\bDelta}> \zeta_{\ell-1}$. Since $R(\zeta_{\ell})=\fro{\bTheta^\star}+\zeta_\ell$,  we thus arrive at
\begin{align*}
	\zeta_{\ell-1}<\fro{\bDelta}\le 4C\sigma\brac{\sqrt{D\log\frac{nR(\zeta_0)}{\zeta_0}}+\sqrt{t_0+\ell}},
\end{align*}
By the definition of $\zeta_{\ell}$, we get
\begin{equation}\label{eq:contrad-0}
	2^{\ell-1}\zeta_0<4C\sigma\brac{\sqrt{D\brac{\log n+\log\brac{1+\frac{\fro{\bTheta^\star}}{\zeta_0}}}}+\sqrt{t_0+\ell}}.
\end{equation}
In addition, 
\begin{align*}
	\log\brac{1+\frac{\fro{\bTheta^\star}}{\zeta_0}}\le \log\brac{1+\frac{\fro{\bTheta^\star}}{\sigma\sqrt{D}}},
\end{align*}
by the definition of $\zeta_0$, and thus combining with the definition of $\bar D$ and \eqref{eq:contrad-0} we attain
\begin{equation}\label{eq:contrad}
	2^{\ell}(\sqrt{\bar{D}}+\sqrt{t_0})<\sqrt{\bar{D}} +\sqrt{t_0+\ell}.
\end{equation}
But \eqref{eq:contrad} is impossible for any $\ell\ge 1$ as long as $D\ge 1$ and $t\ge 1$. This contradiction shows $\fro{\bDelta}\le \zeta_0$ on $\calE$, which yields the claimed bound after absorbing constants.

\subsection[Proof of decoupled error-rate theorem]{Proof of \Cref{thm:decouple-rate}}
Without loss of generality, we assume $\sigma=1$ and $\bP=\I_K$. We start with the following lemma.
\begin{lemma}\label{lem:decouple-rate}
	Suppose  \Cref{ass:snr} and \ref{ass:align}  hold, $\tau_\rho<\sqrt{c_0/(4C_0)}$,
    \begin{gather*}
    	\frac{1}{K^3 \log(n\vee s_\star)} \min \left( \frac{s_\star^2}{K}, \frac{ns_\star^2}{(d_1+d_2)(\sum_{k=1}^{K}r_k)} \right) \ge CC^2_0c_0^{-3}\brac{1-\tau_\rho\sqrt{\frac{C_0}{c_0}}}^{-1}
    \end{gather*}
    and
	\begin{align}\label{eq:theta-consist}
		\PP\ebrac{\frac{1}{\sqrt{s_\star^2K}}\bfro{\hat\bTheta-\bTheta^\star} \le  \frac{c^{3/2}_0}{4C_0K}\brac{1-\tau_\rho\sqrt{\frac{C_0}{c_0}}}^{1/2}}\ge 1-\zeta,
	\end{align}
    
	then we have with probability at least $1-\zeta-e^{-n\vee (d_1+d_2)}$,
	\begin{align*}
	\frac{1}{\sqrt{nK}}\bfro{\hat\bPi-\bPi^\star} &\le C_1\sqrt{ \left( \frac{K}{s_\star^2}+\frac{(d_1+d_2)\sum_{k=1}^K r_k}{n s_\star^2} \right) \log( n \vee s_\star )},\\
	\frac{1}{\sqrt{d_1d_2K}}\bfro{\hat\bTheta-\bTheta^\star} &\le C_2 \sqrt{ \left( \frac{K^2}{d_1d_2}+\frac{(d_1+d_2)K\sum_{k=1}^K r_k}{n d_1d_2} \right) \log( n \vee s_\star) },
\end{align*}
for some constants $C_1,C_2>0$ depending on $c_0$ and $C_0$.
\end{lemma}
It remains to verify \eqref{eq:theta-consist} under the conditions in \Cref{thm:decouple-rate} holds with $\zeta=(d_1d_2)^{-10}+e^{-n\vee (d_1+d_2)}$, which is implied by \Cref{thm:main} such that with probability at least $1-O((n\vee p))^{-10}$,
\begin{align}
        \bfro{\hat\bTheta - \bTheta^\star } \leq  \bfro{\hat\bPi\hat\bTheta^\top - \bPi^\star\bTheta^{\star\top} }\le C\sigma\sqrt{\brac{nK+\sum_{k=1}^K r_k(d_1+d_2)}\log \brac{n\vee s_\star}}, 
\end{align}
where the first inequality holds due to ${\hat\bPi}_{S^\star,:} = \bPi^\star_{S^\star,:} = \I_K$. We use the remaining conditions from (\ref{large-network}) to obtain the result.

\subsection[Proof of product lower-bound theorem]{Proof of \Cref{thm:product-lb}}
\paragraph{Lower bound for $\bTheta$}
Without loss of generality we assume $r_k=r$ for all $k\in[K]$, and assume $n$ is a multiple of $K$, so $n_k = m = n/K$. Notice that it suffices to show that 
\begin{align*}
   \inf_{\hat\W}\sup_{\W^* \in \calP_s}\EE\fro{\hat\W-\W^\star} \ge \sigma\sqrt{cKrd_1},
\end{align*}
then obtain the analogous bound with $d_1$ replaced by $d_2$ by transposition.
Let $\bPi_0 \in \{0,1\}^{n \times K}$ be the block-diagonal indicator matrix:$$\bPi_0 = \begin{bmatrix}
\mathbf{1}_{m} & \mathbf{0} & \dots & \mathbf{0} \\
\mathbf{0} & \mathbf{1}_{m} & \dots & \mathbf{0} \\
\vdots & \vdots & \ddots & \vdots \\
\mathbf{0} & \mathbf{0} & \dots & \mathbf{1}_{m}
\end{bmatrix}.$$
This matrix satisfies the simplex constraint  and its columns are orthogonal $$\bPi_0^\top \bPi_0 = m \I_K = \frac{n}{K} \I_K.$$

Next, we fix some $\U_0\in \mathbb{O}_{d_1,r}$ and consider the ball centered at $\U_0$ with radius of $\epsilon\in (0,\sqrt{2r}]$ under the chordal Frobenius-norm metric $\textsf{dist}(\U_1,\U_2):=\min_{\O\in\mathbb{O}_r}\fro{\U_1- \U_2\O}$:
$$B_\epsilon(\U_0):=\{\U:\textsf{dist}(\U,\U_0)\le \epsilon\}$$
By Lemma 1 in \cite{cai2013sparse} and the equivalence between $\textsf{dist}(\cdot,\cdot)$ and $\|\sin\Theta(\cdot,\cdot)\|$, we have for any $\alpha\in(0,1)$, there exists $\{\U_i^\prime\}_{i=1}^m$, a packing of $B_\epsilon(\U_0)$ such that for some absolute constant $c_0>0$:
$$m\ge \left(\frac{c_0}{\alpha}\right)^{r(d_1-r)},\quad \min_{i< j}\textsf{dist}(\U_i^\prime,\U_j^\prime)\ge \alpha \epsilon$$
Denote $\O_i=
\arg\min_{\O\in\mathbb{O}_r}\fro{\U_i^\prime- \U_0\O}$. Fix $\bSigma= \text{diag}(\lambda_1,\cdots,\lambda_r)$ with $\lambda_1= \cdots = \lambda_r=\lambda$ and $\V\in\mathbb{O}_{d_2,r}$, we can construct ${\tilde\M}_{(i)}= \U_i^\prime \O_i^\top\bSigma \V^\top $ for $i=1,\cdots, m$. Notice that
\begin{align*}
\fro{\tilde\M_{(i)}-\tilde\M_{(j)}}&=\fro{\U_i^\prime \O_i^\top\bSigma \V^\top  -\U_j^\prime \O_j^\top\bSigma \V^\top }= \lambda \fro{\U_i^\prime \O_i^\top -\U_j^\prime \O_j^\top}\\
&\ge \lambda \cdot\textsf{dist}(\U_i^\prime,\U_j^\prime)\ge \lambda \alpha\epsilon.	
\end{align*}
On the other hand,
\begin{align*}
	\fro{\tilde\M_{(i)}-\tilde\M_{(j)}}&= \lambda \fro{\U_i^\prime \O_i^\top -\U_j^\prime \O_j^\top}\le\lambda \brac{\textsf{dist}(\U_i^\prime,\U_0)+\textsf{dist}(\U_j^\prime,\U_0)}\le 2\lambda \alpha\epsilon.	
\end{align*}

Let $\calM=\ebrac{\tilde\M_{(i)},\cdots,\tilde\M_{(j)}}$. Consider $\calN_{\Theta}:=\ebrac{\bTheta:\unvec_{d_1,d_2}\brac{\bTheta_{\cdot,k}}\in\calM,~\forall k\in[K]}$.
Recall $\bTheta = [\vec(\M_1), \dots, \vec(\M_K)]$.  Notice that  for any $\bTheta,\bTheta'\in\calN_{\Theta}$,
\begin{align*}
	\fro{\bTheta- \bTheta'}^2 = \sum_{k=1}^K \fro{\M_k - \M_k'}^2\asymp\lambda^2\alpha^2\epsilon^2K
\end{align*}
Let $\calN:=\ebrac{\bPi_0\bTheta^\top:\unvec_{d_1,d_2}\brac{\bTheta_{\cdot,k}}\in\calM,~\forall k\in[K]}$, we have $\log\ab{\calN}\asymp Krd_1$. Then for any $\W,\W'\in\calN$
\begin{align*}
\fro{\W - \W'}^2 = \frac{n}{K} \tr\left( (\bTheta - \bTheta') (\bTheta - \bTheta')^\top \right) = n\lambda^2\alpha^2\epsilon^2.
\end{align*}
By Fano's inequality, the minimax risk is lower bounded by $c Krd_1$ by choosing $\epsilon$ such that $	n\lambda^2\alpha^2\epsilon^2/\sigma^2= c Krd_1$ for some universal small constant $c>0$. 

\paragraph{Lower bound for $\bPi$}
Let $\bTheta_0 \in \RR^{p \times K}$ be fixed such that $\bTheta_0^\top \bTheta_0 = s^2 \I_K$. Let $\c = (1/K, \dots, 1/K)^\top$ be the center of $\Delta_K$. We define a perturbation magnitude $\alpha > 0$ (to be determined later). 
Let $\U \in \RR^{K \times (K-1)}$ be a matrix where the columns $\u_1, \dots, \u_{K-1}$ form an orthonormal basis for the subspace $\{\x : \mathbf{1}^\top \x = 0\}$. Note that $\mathbf{1}^\top \U = \mathbf{0}$. For any binary vector $\bome \in \{0,1\}^{K-1}$, we define
\begin{align*}
	\bpi(\bome) = \c + \alpha\U\brac{\bome - 0.5\mathbf{1}_{K-1}},
\end{align*}
where $\alpha$ is to be determined later.  To ensure $\bpi\in \Delta_K$, we require $\bpi\ge 0$ (sum-to-one is satisfied by $\U^\top \mathbf{1}=0$). Since $c_k = 1/K$, it suffices that $\|\alpha \U (\bome - 0.5\mathbf{1}_{K-1})\|_\infty \le 1/K$. This condition is satisfied if $\alpha \lesssim K^{-3/2}$.
\begin{comment} Then we have
\begin{align*}
	\op{\bpi(\bome) - \bpi(\bome')}^2 = \op{ (\c + \A\bome) - (\c + \A\bome') }^2=\alpha^2 \op{ \bome - \bome' }^2\asymp \alpha^2K.
\end{align*}
Here we have used that $\op{ \bome - \bome' }^2=d_{H}(\bome,\bome^\prime)\ge (K-1)/4$ and $\op{ \bome - \bome' }^2=d_{H}(\bome,\bome^\prime)\le K-1$. 
\end{comment}

By Varshamov-Gilbert Lemma, there exists a subset $\boldsymbol{\Omega} \subset \{0,1\}^{n(K-1)}$ such that (i) $|\boldsymbol{\Omega}| \ge 2^{n(K-1)/8}$, and (ii) for any distinct $\bar\bome, \bar\bome' \in \boldsymbol{\Omega}$, the Hamming distance $d_H(\bar\bome, \bar\bome')$ satisfies $d_H(\bar\bome, \bar\bome') := \sum_{j=1}^{n(K-1)} \mathbb{I}(\bar\omega_j \neq \bar\omega'_j) \ge n(K-1)/4$. Next, we define the mapping from a single long binary vector $\bar\bome \in \boldsymbol{\Omega}$ to the matrix $\mathbf{\Pi}$ by chopping $\bar\bome$ into $n$ chunks. Let $\bar\bome := (\bome_1, \bome_2, \dots, \bome_n)$ where each chunk $\bome_i \in \{0,1\}^{K-1}$. The $i$-th row of the matrix $\mathbf{\Pi}(\bar\bome)$ is generated by the $i$-th chunk as $\bpi_i:=\bpi(\bome_i)$.  Note that for distinct $\bar\bome, \bar\bome' \in \boldsymbol{\Omega}$:
 \begin{align*}
 	\fro{\W(\bar\bome) - \W(\bar\bome^\prime)}^2 &= \tr((\bPi(\bar\bome) - \bPi'(\bar\bome^\prime)) \bTheta_0^\top \bTheta_0 (\bPi(\bar\bome) - \bPi'(\bar\bome^\prime))^\top)\\
 	& = s^2 \fro{\bPi(\bar\bome) - \bPi'(\bar\bome^\prime)}^2\asymp s^2\alpha^2nK
 \end{align*}
 Let $\bar\calP:=\ebrac{\W(\bar\bome)=\bPi(\bar\bome)\bTheta_0^\top:\bar\bome\in\boldsymbol{\Omega}}$, then $\log |\bar\calP| \asymp nK$. It suffices to consider the parameter set $\bar\calP$ as $\bar\calP\subset\calP_s$, thus
 \begin{align*}
	\inf_{\hat\W}\sup_{\W^\star\in \calP_s}\EE\fro{\hat\W-\W^\star}^2\ge \inf_{\hat\W}\sup_{\W^\star\in \bar\calP}\EE\fro{\hat\W-\W^\star}^2.
\end{align*}
Under the Gaussian model $\bY \sim N(\vec(\W), \sigma^2 \I_{np})$, the KL divergence is given by 
$$D_{KL}(\mathbb{P}_{\bar\bome} \| \mathbb{P}_{\bar\bome'}) = \frac{1}{2\sigma^2} \fro{\W(\bar\bome) - \W(\bar\bome')}^2\asymp s^2\alpha^2nK.$$
By Fano's inequality, the minimax risk is lower bounded by $s^2\alpha^2nK$ up to constant provided that $	s^2 \alpha^2 nK/\sigma^2 \le c nK$ for some universal small constant $c>0$, which is implied by $\alpha\le \sqrt{c}\sigma/s$. By choosing $\alpha = \sqrt{c}\sigma / s$, we require $\alpha$ is small enough to stay in the simplex provided $s \gtrsim \sigma K^{3/2}$. 

\subsection[Proof of pure-subject recovery proposition]{Proof of \Cref{prop:est-Shat}}

Let $\hat{S} = (\hat{S}_1, ..., \hat{S}_K)$ be the estimated pure subject indices identified by SPA applied on the rows of left singular subspace $\hat{\U}$ obtained from rank-K truncated SVD of $\Y$. \citep{CHG24} estimates $\bPi^*$ as $\hat{\bPi} = \hat{\U} \hat{\U}_{\hat{S},:}^{-1}$. As a result of Theorem 2 of \citep{CHG24}, $\exists$ a constant $C$ such that with probability at least $1 - O((n \vee p)^{-10})$, \begin{align*}
    \| \hat{\bPi} - \bPi^* \P \|_{2,\infty} \leq C \kappa^2(\bPi^*) \sigma_1(\bPi^*) \xi_1.
\end{align*}
As long as $\text{min}_{i \in S^*, j \not\in S^*} ||\bpi^*_i - \bpi^*_j||_2 > C\kappa^2(\bPi^*)\sigma_1(\bPi^*)\xi_1$, then the only estimated membership vectors closer than $C\kappa^2(\bPi^*)\sigma_1(\bPi^*)\xi_1$ to a standard basis vector in $l_2$ norm are those belonging to pure subjects. Observe that $\hat{\bPi}_{\hat{S},:} = \I_K$. Hence, $\hat{S}$ recovers the pure subhect indices up to permutation.

\section{Proofs of Lemmas}\label{sec:ProofsLemmas}

\subsection[Proof of rank-size lemma]{Proof of \Cref{lem:rank-size}}

Let $\U_k^* \bSigma_k^* \V_k^{*\top}$ denote the singular value decomposition of $\M^*_k$, and let $\bcalM^*$ denote the tensor satisfying $\bcalM^*(:,:,k) = \M^*_k$. Then $\EE[\bcalX]$ can be represented as 
\begin{align}\label{tens-tucker}
    \EE[\bcalX] &= \bcalM^* \times_3 \bPi^* = \bcalC^* \times_1 \U^*_{\cup} \times_2 \V^*_{\cup} \times_3 \bPi^*,
\end{align}
where $\U^*_{\cup} = (\U^*_1 | ... | \U^*_K) \in \RR^{d_1 \times \mathring{r}}$ and $\V^*_{\cup} = (\V^*_1 | ... | \V^*_K) \in \RR^{d_2 \times \mathring{r}}$ are the column-wise concatenations of the $K$ distinct left and right singular subspaces respectively, $\mathring{r} = \sum_{k=1}^K r_k$, and $\bcalC^*$ is the $\mathring{r} \times \mathring{r} \times K$ tensor satisfying $\bcalC^*(:,:,k) = \text{diag}(0_{r_1}, ..., 0_{r_{k-1}}, \bSigma_k, 0_{r_{k+1}}, ..., 0_{r_K})$. Take the mode-1 and mode-2 unfoldings of (\ref{tens-tucker}) to observe that $r_U, r_V \leq \mathring{r}$.

\subsection[Proof of model-selection lemma]{Proof of \Cref{lem:model-selection}}

\begin{gather*}
    \EE[\underline{\bcalX}] = \underline{\bcalC}^* \times \underline{\U}_\cup^* \times_2 \underline{\V}_\cup^* \times_3 \bPi^*, \text{ where } \\
    \underline{\U}_\cup^* = (\U^{(1)*}_\cup | ... | \U^{(M)*}_\cup) \qquad \U^{(m)*}_\cup = (\U_1^{(m)*} | ... | \U_K^{(m)*}) \\ 
    \underline{\V}_\cup^* = (\V^{(1)*}_\cup | ... | \V^{(M)*}_\cup) \qquad \V^{(m)*}_\cup = (\V_1^{(m)*} | ... | \V_K^{(m)*}) \\
    \underline{\bcalC}^*(:,:,k,m) = \text{diag}(\mathbf{0}_{\mathring{r}^{(1)}},...,\mathbf{0}_{\mathring{r}^{(m-1)}}, \mathbf{0}_{r^{(m)}_1},...,\mathbf{0}_{r^{(m)}_{k-1}},  \bSigma_k^{(m)*}, \mathbf{0}_{r^{(m)}_{k+1}}, ..., \mathbf{0}_{r^{(m)}_K}, \mathbf{0}_{\mathring{r}^{(m+1)}}, ..., \mathbf{0}_{\mathring{r}^{(M)}})\\
    \M_k^{(m)*} = \U_k^{(m)*}\bSigma_k^{(m)*}\V_k^{(m)*\top},
\end{gather*}
and $\mathring{r}^{(m)} = \sum_{k=1}^K r_k^{(m)}$.

\subsection[Proof of decoupling-rate lemma]{Proof of \Cref{lem:decouple-rate}}
We start with the following the decomposition:
 \begin{align*}
	\bDelta_W = \hat\bPi\hat\bTheta^\top-\bPi^\star\bTheta^{\star\top}
= \bDelta_{\Pi}\bTheta^{\star\top}  +\bPi^\star\bDelta_{\Theta}^\top + \bDelta_{\Pi}\bDelta_{\Theta}^\top.
 \end{align*}
To lower bound $\fro{\bDelta_W}$, we need the following lemma.
  \begin{lemma}\label{lem:overall-sum-lb}
	Suppose \Cref{ass:snr} and \ref{ass:align}  hold and  $\tau_\rho<\sqrt{c_0/C_0}$, we then have
\begin{align*}
	\fro{\bDelta_\Pi\bTheta^{\star\top}+\bPi^\star\bDelta_{\Theta}^\top }^2\ge\frac{c^2_0}{4C^2_0K} \brac{1-\tau_\rho\sqrt{\frac{C_0}{c_0}}}\brac{\fro{\bDelta_\Pi\bTheta^{\star\top}}^2+\fro{\bPi^\star\bDelta_{\Theta}^\top }^2}.
\end{align*}
\end{lemma}
 Denote $\delta_W :=\fro{\bDelta_W} $, $\delta_\Pi :=\fro{\bDelta_\Pi} $, $\delta_\Theta :=\fro{\bDelta_\Theta} $,  $ \delta_D:= \fro{\bDelta_{\Pi}\bDelta_{\Theta}^\top}$. By \Cref{lem:overall-sum-lb} and the triangle inequality, we have
\begin{align*}
\sqrt{\eta}\fro{\bDelta_{\Pi}\bTheta^{\star\top}}\le \delta_W+\delta_\Pi\delta_\Theta,\qquad \sqrt{\eta}\fro{\bPi^\star\bDelta_{\Theta}^\top }\le \delta_W+\delta_\Pi\delta_\Theta,
\end{align*}
where $\eta:=\frac{c^2_0}{4C^2_0K} \brac{1-\tau_\rho\sqrt{\frac{C_0}{c_0}}}$.
Combined with \Cref{ass:snr}, we obtain
\begin{align}\label{eq:pi-theta-ineq}	
\delta_\Pi \le \frac{\delta_W+\delta_\Pi\delta_\Theta}{\mathfrak{C}_p },\qquad  \delta_\Theta\le \frac{\delta_W+\delta_\Pi\delta_\Theta}{\mathfrak{C}_n },
\end{align}
where $\mathfrak{C}_p := \sqrt{c_0\eta s_\star^2}$ and $\mathfrak{C}_n := \sqrt{c_0\eta n/K}$. Note that the event in \eqref{eq:theta-consist} is equivalent to $\delta_\Theta\le  \mathfrak{C}_p/2$. On the event that \eqref{eq:theta-consist} holds, we have 
\begin{align*}
	\fro{\bDelta_\Pi} \le \frac{2}{\mathfrak{C}_p}\fro{\bDelta_W}.
\end{align*}
Combined with \Cref{thm:main} and a union bound argument, we get with probability at least $1-\zeta-e^{-n\vee (d_1+d_2)}$,
\begin{align}\label{eq:pi-rate}
	\frac{1}{\sqrt{nK}}\fro{\hat\bPi-\bPi^\star} \le C_1\sqrt{ \left( \frac{K}{s_\star^2}+\frac{(d_1+d_2)\sum_{k=1}^K r_k}{n s_\star^2} \right) \log(n \vee s_\star)},
\end{align}
where the constant $C_1>0$ depends on $c_0$ and $C_0$.
On the other hand, \eqref{eq:pi-rate} also implies that $\delta_\Pi\le \mathfrak{C}_n/2$ provided that 
% \begin{align*}
% 	\frac{K}{d_1d_2}+\frac{(d_1+d_2)\sum_{k=1}^K r_k}{n d_1d_2}\le \frac{c_0^2}{4CK^3}\brac{1-\tau_\rho\sqrt{\frac{C_0}{c_0}}}.
% \end{align*}
\begin{gather*}
	\frac{1}{K^3 \log(n \vee s_\star)}\min \left\{ \frac{s_\star^2}{K}, \frac{ns_\star^2}{(d_1+d_2)(\sum_{k=1}^{K}r_k)} \right\} \ge CC^2_0c_0^{-3}\brac{1-\tau_\rho\sqrt{\frac{C_0}{c_0}}}^{-1}
\end{gather*}
for some universal constant $C>0$. Thus, we can similarly deduce that  with probability at least $1-\zeta-e^{-n\vee (d_1+d_2)}$,
\begin{align*}
	\frac{1}{\sqrt{d_1d_2K}}\fro{\hat\bTheta-\bTheta^\star} \le C_2 \sqrt{\left( \frac{K^2}{d_1d_2}+\frac{(d_1+d_2)K\sum_{k=1}^K r_k}{n d_1d_2} \right) \log(n \vee s_\star)},
\end{align*}
where the constant $C_2>0$ depends on $c_0$ and $C_0$.

\subsection[Proof of tau-bound lemma]{Proof of \Cref{lem:tau-bound}}
Fix $k\in[K]$. With slight abuse of notation we write $\H$ and $\D$ instead of $\H(\bTheta)$ and $\D(\bTheta)$. We start by noticing that 
\begin{align}\label{eq:proj-hk-bound}
	\|\calP_{\calS_\bTheta^\star}(\h_k)\|^2=(\bTheta^{\star\top}\h_k)^\top(\bTheta^{\star\top}\bTheta^\star)^{-1}(\bTheta^{\star\top}\h_k)
\le \frac{1}{c_0p}\|\bTheta^{\star\top}\h_k\|^2.
\end{align}
In addition, we have $(\bTheta^{\star\top}\h_k)_j=\langle \btheta_j^\star,\h_k\rangle=\langle\M_j^\star,\H_k\rangle$ for $j\in[K]$. By definition,  $(\bTheta^{\star  \top}\h_k)_k=0$. Let $M_k:=\unvec_{d_1,d_2}(\bTheta_{\cdot,k})$.
Denote 
\begin{align*}
	\H_{k,11}:=\U_k^\top \H_k \V_k,\quad
\H_{k,12}:=\U_k^\top \H_k \V_{k,\perp},\quad
\H_{k,21}:=\U_{k,\perp}^\top \H_k \V_k,\quad
\H_{k,22}:=\U_{k,\perp}^\top \H_k \V_{k,\perp}.
\end{align*}
Since $\bTheta\in\calM(\rho)$,  $\rho<s_{\min}/2$ and $\rank(\M_k)\le r_k$, \rev{by a standard Schur-complement argument for fixed-rank perturbations} we get 
\begin{align*}
	\H_{k,22}= \H_{k,21}\H_{k,11}^{-1}\H_{k,12}.
\end{align*}
Moreover, we have $\bop{\H_{k,11}^{-1}}\le 2/\sigma_{r_k}(\M_k^\star)\le 2/s_{\min}$, and hence
\begin{align*}
	\fro{\H_{k,22}}\le \frac{2}{s_{\min}}\op{\H_{k,21}}\op{\H_{k,12}}\le \frac{1}{s_{\min}}\fro{\H_{k}}^2\le \frac{\rho}{s_{\min}}\fro{\H_{k}}.
\end{align*}
Therefore, for $j\neq k$, we have
\begin{align*}
	|\langle \M_j^\star,\H_k\rangle| &\le |\
\langle \U_j\bSigma_j\V_j^\top,\U_k\H_{k,11}\V_k^\top\rangle| + |\langle \U_j\bSigma_j\V_j^\top,\U_k\H_{k,12}\V_{k,\perp}^\top\rangle| \\
&\quad + |\langle \U_j\bSigma_j\V_j^\top,\U_{k,\perp}\H_{k,21}\V_k^\top\rangle| + |\langle \U_j\bSigma_j\V_j^\top,\U_{k,\perp}\H_{k,22}\V_{k,\perp}^\top\rangle| \\
&\le \fro{\M_j^\star}\Big(\bop{\U_j^\top \U_k}\bop{\V_{j}^\top \V_k}\fro{\H_{k,11}} + \bop{\U_j^\top \U_{k}}\fro{\H_{k,12}}\\
&\hspace{2cm}+\bop{\V_{j}^\top \V_k}\fro{\H_{k,21}} +\fro{\H_{k,22}}\Big) \\
&\le  \fro{\M_j^\star}\brac{\varepsilon_U\varepsilon_V + \varepsilon_U + \varepsilon_V+\frac{\rho}{s_{\min}}}\fro{\H_k}.
\end{align*}
Squaring and summing over $j\in[K]$, we obtain
\begin{align*}
	\bop{\bTheta^{\star\top}\h_k}^2
=\sum_{j\neq k} |(\bTheta^{\star\top}\h_k)_j|^2
\le
\Big(\varepsilon_U\varepsilon_V+\varepsilon_U+\varepsilon_V+\frac{\rho}{s_{\min}}\Big)^2
\op{\h_k}^2\sum_{j\neq k}\op{\btheta_j^\star}^2.
\end{align*}
Combined with \Cref{ass:snr} and \eqref{eq:proj-hk-bound}, we get
\begin{align*}
	\|\calP_{\calS_\bTheta^\star}(\h_k)\|^2
\le \frac{C_0K}{c_0}\Big(\varepsilon_U\varepsilon_V+\varepsilon_U+\varepsilon_V+\frac{\rho}{s_{\min}}\Big)^2 \op{\h_k}^2.
\end{align*}
The proof is completed by noting that $\bfro{\calP_{\calS_\bTheta^\star}(\H)}^2=\sum_{k=1}^K\bop{\calP_{\calS_\bTheta^\star}(\h_k)}^2$.

\subsection[Proof of overall sum lower-bound lemma]{Proof of \Cref{lem:overall-sum-lb}}
For simplicity we drop the dependence  on $\bTheta$ for $\H(\cdot)$ and $\D(\cdot)$ and simply write as $\H$ and $\D$. By definition, we have  $\bDelta_{\Theta} =\bTheta^\star\D+\H$. Notice that  $\bPi^\star\bDelta_{\Theta}^\top= \bPi^\star\D\bTheta^{\star\top}+\bPi^\star\H^\top$, we can decompose 
\begin{align}\label{eq:sum-decomp}
	\bDelta_\Pi\bTheta^{\star\top}+\bPi^\star\bDelta_{\Theta}^\top =\bDelta_\Pi\bTheta^{\star\top}+\bPi^\star\D\bTheta^{\star\top}+\bPi^\star\H^\top.
\end{align}
We have the following two lemmas:
\begin{lemma}\label{lem:sum-lb}
	Suppose \Cref{ass:snr} holds. We have
	\begin{align*}
	\fro{\bDelta_\Pi\bTheta^{\star\top}+\bPi^\star\D\bTheta^{\star\top}}^2\ge \frac{c^2_0}{2C^2_0K}\brac{\fro{\bDelta_\Pi\bTheta^{\star\top}}^2+\fro{\bPi^\star\D\bTheta^{\star\top}}^2}.
\end{align*}
\end{lemma}
\begin{lemma}\label{lem:cross-ub}
	Suppose \Cref{ass:align} holds. For any $\X\in\RR^{n\times p}$ whose rows lie in $\calS_\bTheta^\star$, 
	\begin{align*}
		\ab{\inp{\X}{\bPi^\star\H^\top}}\le \tau \sqrt{\frac{C_0}{c_0}}\fro{\X}\fro{\bPi^\star\H^\top}.
	\end{align*}
\end{lemma}
By \Cref{lem:cross-ub}, we get
\begin{align*}
	\ab{\inp{\bDelta_\Pi\bTheta^{\star\top}+\bPi^\star\D\bTheta^{\star\top}}{\bPi^\star\H^\top}}\le \tau\sqrt{\frac{C_0}{c_0}}\fro{\bDelta_\Pi\bTheta^{\star\top}+\bPi^\star\D\bTheta^{\star\top}}\fro{\bPi^\star\H^\top}.
\end{align*}
We then proceed as
\begin{align*}
	\fro{\bDelta_\Pi\bTheta^{\star\top}+\bPi^\star\bDelta_{\Theta}^\top }^2\ge &\fro{\bDelta_\Pi\bTheta^{\star\top}+\bPi^\star\D\bTheta^{\star\top}}^2+\fro{\bPi^\star\H^\top}^2\\
	&-2 \tau \sqrt{\frac{C_0}{c_0}}\fro{\bDelta_\Pi\bTheta^{\star\top}+\bPi^\star\D\bTheta^{\star\top}}\fro{\bPi^\star\H^\top}\\
	\ge & \brac{1-\tau\sqrt{\frac{C_0}{c_0}}}\brac{\fro{\bDelta_\Pi\bTheta^{\star\top}+\bPi^\star\D\bTheta^{\star\top}}^2+\fro{\bPi^\star\H^\top}^2}.
\end{align*}
where in the last inequality we've used $\tau<\sqrt{c_0/C_0}$ and  $x^2+y^2-2t xy\ge (1-t)(x^2+y^2)$ for $t\in[0,1)$.  Using $\fro{\A}^2+\fro{\B}^2\ge \fro{\A+\B}^2/2$ and \Cref{lem:sum-lb}, we arrive at the desired result. 
\subsection[Proof of sum lower-bound lemma]{Proof of \Cref{lem:sum-lb}}

\begin{proof}
	Notice that 
	\begin{align*}
		\fro{\bDelta_\Pi\bTheta^{\star\top}+\bPi^\star\D\bTheta^{\star\top}}^2= \tr\brac{(\bDelta_\Pi+\bPi^\star\D)\bTheta^{\star\top}\bTheta^{\star}(\bDelta_\Pi+\bPi^\star\D)^\top }\ge c_0s_\star^2\fro{\bDelta_\Pi+\bPi^\star\D}^2.
	\end{align*}
	Define $\calD := \{\X\in\mathbb{R}^{n\times K}: \X\mathbf 1_K=0\ \text{rowwise}\} $ and $\calC_0 := \{\bPi^\star \D: \D\ \text{diagonal}\}$, and note that $\bDelta_{\Pi}\in\calD$ and $\bPi^\star\D\in\calC_0$. Firstly, we show that 
	\begin{align*}
		\fro{\C+\D}^2\ge \frac{c_0}{2C_0K}\brac{\fro{\C}^2+\fro{\D}^2},\qquad \forall \C\in\calC_0,\quad \D\in\calD.
	\end{align*}
	By definition, $\calP_{\calD^\perp}(\A)=\frac{1}{K} \A \mathbf{1}_K\mathbf{1}_K^\top$ for any $\A$. Fix any $\C=\bPi^\star\D \in\calC_0$ and let $\d:=(D_{11},\cdots,D_{KK})$, we get
	\begin{align*}
		\text{dist}(\C,\calD)^2&=\fro{\calP_{\calD^\perp}(\C)}^2=\frac{1}{K^2}\sum_{i=1}^n\op{\C_{i,\cdot}\mathbf{1}_K\mathbf{1}_K^\top}^2\\
		&= \frac{1}{K}\sum_{i=1}^n\brac{(\bPi^\star\d)_{i}}^2=\frac{1}{K}\op{\bPi^\star\d}^2\ge \frac{c_0 n}{K^2}\op{\d}^2.
	\end{align*} 
	On the other hand, we have
	\begin{align*}% \fro{\C}^2=\fro{\bPi^\star\D}^2=\sum_{k=1}^KD_{kk}^2\op{\bPi_{\cdot,k}}^2\le \frac{C_0 n}{K}\op{\d}^2.
        \fro{\C}^2=\fro{\bPi^\star\D}^2\le \frac{C_0 n}{K}\op{\d}^2.
	\end{align*}
	We thus conclude that 
	\begin{align*}
		\frac{\mathrm{dist}(\C,\calD)^2}{\fro{\C}^2}\ge\frac{c_0}{C_0 K}=:\alpha.
	\end{align*}
	Write $\C=\calP_{\calD}(\C)+\calP_{\calD^\perp_0}(\C)$, then we have
	\begin{align*}
		\fro{\calP_{\calD}(\C)}^2=\fro{\C}^2-\fro{\calP_{\calD^\perp}(\C)}^2\le (1-\alpha)\fro{\C}^2.
	\end{align*}
	Define $t:=\sqrt{1-\alpha}\in(0,1)$. Then for any $\C\in \calC_0$ and $\D\in\calD$, we have
	\begin{align*}
		\inp{\C}{\D}=\inp{\calP_{\calD}(\C)+\calP_{\calD^\perp_0}(\C)}{\D}\le  \op{\calP_{\calD}(\C)}_F\op{\D}_F\le t\op{\C}_F\op{\D}_F.
	\end{align*}
	We thus arrive at
	\begin{align*}
		\fro{\C+\D}^2&=\fro{\C}^2+\fro{\D}^2+2\inp{\C}{\D}\\
		&\ge \fro{\C}^2+\fro{\D}^2-2t\op{\C}_F\op{\D}_F\\
		&\ge (1-t)\brac{\fro{\C}^2+\fro{\D}^2}.
	\end{align*}
	Use $1-\sqrt{1-\alpha}\ge \alpha/2$ for small $\alpha$ to obtain the desired inequality.

    From here, we have $\fro{\bDelta_\Pi\bTheta^{\star\top}+\bPi^\star\D\bTheta^{\star\top}}^2 \geq c_0 s_\star^2 \frac{c_0}{2C_0K}(\fro{\bDelta_\Pi}^2 + \fro{\bPi^*\D}^2)$. To conclude, we use $\fro{\bDelta_\Pi\bTheta^{*\top}}^2 \leq C_0 s_\star^2 \fro{\bDelta_\Pi}^2$ and $\fro{\bPi^* \D}^2 \leq C_0 s_\star^2 \fro{\bPi^* \D \bTheta^{*\top}}^2$.
\end{proof}

\subsection[Proof of cross upper-bound lemma]{Proof of \Cref{lem:cross-ub}}
Notice that $\inp{\X}{\bPi^\star\H^\top}=\inp{\bPi^{\star\top}\X}{\H^\top}$. Every row of $\bPi^{\star\top}\X$ lies in $\calS_\bTheta^\star$, and hence orthogonal to $\calP_{\S_\bTheta^{\star\perp}}(\H)$. By \Cref{ass:align}, we get
\begin{align*}
	\ab{\inp{\bPi^{\star\top}\X}{\H^\top}}\le \fro{\bPi^{\star\top}\X}\fro{\calP_{\S_\bTheta^\star}(\H)}\le \tau\op{\bPi^{\star}}\fro{\H}\fro{\X}.
\end{align*}
On the other hand, $\fro{\bPi^\star\H^\top}\ge \sigma_{\min}(\bPi^\star)\fro{\H}$. We thus complete the proof by using \Cref{ass:align} again.

\subsection{Proof of Lemma \ref{lem:verify-gen-model}}
Note that $\sqrt{\mu_1 K/n} = \|\U^*\|_{2,\infty}$ and $\sqrt{\mu_2 K/p} = \|\V^*\|_{2,\infty}$. We observe that:
\begin{align*}
    \|\U^*\|_{2, \infty} &= \|\U^*_{S^*,:}\|_{2, \infty} \leq \|\U^*\| = 1/\sigma_K(\bPi^*) \lesssim 1/\sqrt{n} \\
    \|\V^*\|_{2,\infty} &= \|\bTheta^* \U_{S^*,:}^{-\top} \bSigma^{-1}\|_{2,\infty} \leq \|\bTheta^* \|_{2,\infty} \| \U_{S^*,:}^{-\top} \bSigma^{-1}\| \leq \|\bTheta^* \|_{2,\infty} \| \U_{S^*,:}^{-1} \| /\sigma_K(\bPi^* \bTheta^{*\top}) \\
    &= \|\bTheta^* \|_{2,\infty} \sigma_1(\bPi^*) /\sigma_K(\bPi^* \bTheta^{*\top}) \leq \|\bTheta^* \|_{2,\infty} \kappa(\bPi^*)/\sigma_K(\bTheta^*) \lesssim \|\bTheta^*\|_{2,\infty}/\sqrt{p}.
\end{align*}
Use $d=n/5$ to obtain the result.

\section{Additional Theory \& Algorithm Details}\label{sec:AdditionalTheory}

\subsection{Identifiability of Multimodal Model}\label{sec:ident-mm}

Here, we present the generalization of our definition of identifiability for multimodal datasets, which is used in Proposition \ref{prop:mm-exp-id}. A multimodal LrMMM model given by parameter sets $(\bpi^*_i)_{i=1}^n$ and $((\M^*_{k,m})_{k=1}^K)_{m=1}^M$ is said to be identifiable if for any other parameter sets $(\bpi'_i)_{i=1}^n$ and $((\M'_{k,m})_{k=1}^K)_{m=1}^M$ satisfying $\sum_{k=1}^K \bpi^*_{ik} \M^*_{k,m} = \sum_{k=1}^K \bpi'_{ik} \M'_{k,m}$ for all $i \in [n]$ and $m \in [M]$, there exists a permutation $\sigma \in S_K$ such that $\bpi^*_{ik} = \bpi'_{i\sigma(k)}$ for all $i \in [n]$ and $\M^*_{k, m} = \M'_{\sigma(k), m}$ for all $k \in [K], m \in [M]$.

\subsection{Choice of HOOI in Initialization}\label{sec:ChoiceOfHOOI}

To estimate $\U^*$ in \eqref{spectral-unfolded}, instead of using truncated SVD of $\Y$, we use a tensor decomposition method that exploits the tensor structure of our dataset. In particular, let $\bcalX$ denote the $d_1 \times d_2 \times n$ tensor satisfying $\bcalX(:,:,i) = \X_i$.
\begin{lemma}\label{lem:rank-size}
    Suppose \Cref{ass:pure-subjects} holds and $\text{rank}(\bTheta^*)=K$, and let $r_U := \text{rank}(M_1(\EE[\bcalX]))$ and $r_V := \text{rank}(M_2(\EE[\bcalX]))$. Then $r_U\vee r_V \leq \sum_k r_k$.
\end{lemma}
\noindent By Lemma \ref{lem:rank-size}, $\EE[\bcalX]$ can be written as $\EE[\bcalX] = \bcalC^* \times_1 \W^*_1 \times_2 \W^*_2 \times_3 \W^*_3$,
where $\bcalC \in \RR^{r_U \times r_V \times K}, \W^*_1 \in \OO^{d_1 \times r_U}$, $\W^*_2 \in \OO^{d_2 \times r_V}$, and $\W^*_3 \in \OO^{n \times K}$.
Accordingly, the mode-3 unfolding of this tensor satisfies \begin{equation}
    \label{eq:tensor-unfolded}
    \EE[\Y] = M_3(\E[\bcalX]) = \W^*_3 M_3(\bcalC^*) [\W^*_1 \otimes \W^*_2]^\top.
\end{equation}
Thus, the mode-3 subspace $\W^*_3$ from (\ref{eq:tensor-unfolded}) and the left singular subspace $\U^*$ from (\ref{spectral-unfolded}) represent the same singular subspace up to rotational ambiguity. Therefore, the tensor decomposition method Higher-Order Orthogonal Iteration (HOOI) with ranks $(r_U, r_V, K)$ may also be used to estimate $\U^*$. As a final remark, we note that in practice it is necessary to estimate $r_U, r_V$. This can be done using a scree plot method on the mode-wise unfoldings of $\bcalX$.

\subsection{Multimodal Algorithm Details}\label{sec:MultimodalAlgorithmDetails}

Here, we discuss in detail our multimodal algorithm: Algorithm \ref{alg:Algorithm-MM}. The algorithm is largely the same in principal as Algorithm \ref{alg:Algorithm}, with a few key differences. Firstly, we combine initializations done separately for each modality as follows. An HOOI-based initial estimate $\tilde{\bPi}_m^{(0)}$ of the membership matrix is done separately for each modality according to the same steps as in Algorithm \ref{alg:Algorithm}. Then, we align each modality's membership matrix estimate with the membership matrix estimate of modality 1 by right-multiplying by the permutation matrix that brings $\tilde{\bPi}_m^{(0)}$ closest to $\tilde{\bPi}_1^{(0)}$ in squared Frobenius norm. We then average these aligned matrices and apply Algorithm \ref{alg:GoM} to the average to obtain our final initialized membership matrix $\hat{\bPi}^{(0)}$.

In the iterative portion of the algorithm, estimation of basis matrices proceeds separately for each modality. That is, let $\bcalX_m$ be the $d_{1,m} \times d_{2,m} \times n$ tensor satisfying $\bcalX_m(:,:,i) = \X_{i,m}$, and let $\Y_m$ be its mode-3 unfolding. We regress $\Y_m$ on $\hat{\bPi}^{(t-1)}$ to obtain initial estimates of the basis matrices for modality $m$, and then refine those estimates via SVD as before. For estimation of the membership matrix, the initial estimate of membership scores is obtained for each subject by minimizing the weighted sum of $M$ separate objective functions using pre-specified nonnegative weights $\gamma=(\gamma_1,...,\gamma_M)$, that is, we solve $\text{argmin}_{\pi \in \Delta^{K-1}} \sum_{m=1}^M \gamma_m \| \X_{i,m} - \sum_k \pi_k \hat{\M}_{k,m}^{(t)}\|_F^2$. This incorporation of modality weights allows a practitioner to ensure the data from each modality is contributing appropriately to the estimation of the membership matrix according to its importance and dataset size. The re-estimation of pure subjects proceeds as in Algorithm \ref{alg:Algorithm}, since we are left with only one membership matrix representing all subjects. 

\subsection{Multimodal Datasets with Vector- and Matrix-valued Data}\label{sec:multimodal-vector}

Suppose that for each subject $i$, we have access to $M$ matrix-valued observations $\{ \X_{i,m} \}_{m=1}^M$, with $\X_{i,m} \in \RR^{d_{1,m} \times d_{2,m}}$, and $V$ vector-valued observations $\{ \x_{i,m} \}_{m=M+1}^{M+V}$, with $\x_{i,m} \in \RR^{d_m}$. We obtain an initial membership estimate $\tilde{\bPi}_m^{(0)}$ separately for each of the $M+V$ modalities in the following manner. For $m \leq M$, apply the initialization specified in steps \ref{alg-mm:form-tensor}, \ref{alg-mm:HOOI}, and \ref{alg-mm:GoM} of Algorithm \ref{alg:Algorithm-MM} to $\{ \X_{i,m} \}_{i=1}^n$, and for $m\geq M+1$, we apply the method in \cite{ChenGu24}, which consists of forming the $n \times d_m$ matrix $\R_m$ whose $i$-th row is $\x_{i,m}$, obtaining $\hat{\U}_m$ the left-singular subspace of $\R_m$ using rank-$K$ SVD, and applying Algorithm \ref{alg:GoM} to $\hat{\U}_m$. We align and aggregate our estimates according to steps \ref{alg-mm:align} and \ref{alg-mm:renormalize-init}, this time using all $M+V$ membership matrix estimates. In the iterative part of the algorithm, we follow the same steps for the basis matrix estimation for matrix-valued modalities, and for vector-valued modalities, we regress $\R_m$ on $\hat{\bPi}^{(t-1)}$ to obtain a $d_m \times K$ matrix whose $k$-th column is \textit{basis vector} estimate $\hat{\btheta}_{k,m}^{(t)}$. We forgo SVD of this vector, since we have no assumption that this vector has the structure of a low-rank matrix. The membership estimation step proceeds similarly, now minimizing the objective function \begin{align*}
    \sum_{m=1}^M \gamma_m \|\text{vec}(\X_{i,m}) - \sum_k \pi_k (\hat{\bTheta}_m^{(t)})_{:,k}\|_2^2 + \sum_{m=M+1}^{M+V} \gamma_m \| \x_{i,m} - \sum_k \pi_k \hat{\btheta}_{k,m}^{(t)}\|_2^2.
\end{align*}
Our output consists of basis matrices $(( \hat{\M}_{k,m} )_{k=1}^K)_{m=1}^M$, basis vectors $(( \hat{\btheta}_{k,m} )_{k=1}^K)_{m=M+1}^{M+V}$, and membership matrix $\hat{\bPi}$. The $k$-th basis vector for modality $m\geq M+1$ is regarded as the prototypical pattern of the $k$-th extreme latent profile within modality $m$.

\subsection[Consistency without the alignment condition]{Consistency when \Cref{ass:align} does not hold}\label{sec:still-consist}

\begin{proposition}
    Provided Assumptions \ref{ass:noise}, \ref{ass:pure-subjects}, \ref{ass:snr}, and \ref{ass:separation} hold, there exists a permutation matrix $\P$ such that the following bounds hold for the two step estimator \eqref{eq:estimator-twostep} with probability at least $1 - O((n \vee p)^{-10})$: 
    \begin{align}
    \begin{split}
        &\frac{1}{\sqrt{d_1d_2K}} ||\hat{\bTheta} - \bTheta^*\P ||_F \leq C \sigma \frac{\sqrt{nK + \sum_k r_k (d_1 + d_2)}}{\sqrt{Kd_1d_2}}\sqrt{\log(n \vee s_\star)} \\
        &\frac{1}{\sqrt{nK}}||\hat{\bPi} - \bPi^*\P||_F \leq C' \sigma \frac{\sqrt{n + (\sum_k r_k )(d_1 + d_2)/K}}{\sqrt{s_\star^2}}\sqrt{\log(n\vee s_\star)}.
    \end{split}
    \end{align}
    If it is also true that 
    \begin{align}
        \frac{d_1d_2 \wedge s_\star^2}{n} \gg \log(n \vee s_\star) \qquad K(d_1d_2 \wedge s_\star^2) \gg (\sum_k r_k)(d_1 + d_2)\log(n \vee s_\star),
    \end{align}
    then $\frac{1}{\sqrt{d_1d_2K}} ||\hat{\bTheta} - \bTheta^*\P ||_F \overset{P}{\to} 0$ and $\frac{1}{\sqrt{nK}}||\hat{\bPi} - \bPi^*\P||_F \overset{P}{\to} 0$.
\end{proposition}

\begin{proof} We perform our analysis under event that both \Cref{prop:est-Shat} and \Cref{thm:main} hold, which is true with probability at least $1 - O((n \vee p)^{-10})$. Under this event, we have that $(\hat{\bPi}, \hat{\bTheta})$ satisfies
    \begin{align}
        ||\hat{\bTheta} - \bTheta^* \P||_F \leq ||\hat{\bPi}\hat{\bTheta}^T - \bPi^* \P \P^T \bTheta^{*T} ||_F \leq C \sigma \sqrt{(nK + \sum_k r_k (d_1 + d_2))\log(n \vee s_\star)}
    \end{align}
    since $\hat{\bPi}_{\hat{S},:} = \bPi^*_{\hat{S},:}\P = \I_K$. Thus, 
    \begin{align}
        \frac{1}{\sqrt{d_1d_2K}} ||\hat{\bTheta} - \bTheta^*\P ||_F \leq C \sigma \frac{\sqrt{nK + \sum_k r_k (d_1 + d_2)}}{\sqrt{Kd_1d_2}}\sqrt{\log(n \vee s_\star)}.
    \end{align}
    As for consistency of $\hat{\bPi}$, we observe that \begin{align}
    \begin{split}
        ||\hat{\bPi} - \bPi^*\P||_F \sigma_{\text{min}}(\bTheta^*) &\leq ||(\hat{\bPi} - \bPi^* \P) \P^T \bTheta^{*T} ||_F \\
        &\leq C\sigma \sqrt{(nK + \sum_k r_k(d_1 + d_2))\log(n\vee s_\star)} + ||\hat{\bPi}(\hat{\bTheta}- \bTheta^*\P)^T||_F \\
        &\leq C\sigma \sqrt{(nK + \sum_k r_k(d_1 + d_2))\log(n\vee s_\star)} + \sigma_{\text{max}}(\hat{\bPi}) || \hat{\bTheta} - \bTheta^*\P||_F
    \end{split}
    \end{align}
    where the second inequality comes from the triangle inequality and \Cref{thm:main}. To conclude, we use the assumption that $\sigma_{\text{min}}(\bTheta^*) \asymp s_\star$, and we observe that the simplex constraint on the rows of $\hat{\bPi}$ implies $\sigma_{\text{max}}(\hat{\bPi}) \leq \sqrt{||\hat{\bPi}||_1 ||\hat{\bPi}||_\infty } = \sqrt{1 \cdot n}$. Hence, we arrive at \begin{align}
        \frac{1}{\sqrt{nK}}||\hat{\bPi} - \bPi^*\P||_F \leq C' \sigma \frac{\sqrt{n + (\sum_k r_k )(d_1 + d_2)/K}}{\sqrt{s_\star^2}}\sqrt{\log(n\vee s_\star)}.
    \end{align}
\end{proof}

\section{Additional Simulations and Simulation Details}\label{sec:AdditionalSimulations}

\subsection{Generative Models}\label{sec:GenerativeModels}

As mentioned in Section \ref{sec:Simulations}, we use separate generative models for Bernoulli and Normal data. In specifying the form of these generative models, we take inspiration from \cite{LX25}. For both models, the membership scores are generated in the following way to enforce the pure subjects condition in Assumption \ref{ass:pure-subjects} and the separation condition in Assumption \ref{ass:separation}:
\begin{equation}
\begin{split}\label{generate-pi}
    \bpi_i^* &= \e_i \text{ for } i = 1,...,K;\quad
    \bpi^*_i \overset{i.i.d.}{\sim} \text{Dirichlet}(\mathbf{1}_K) \mathbbm{1}(\min_{k\in[K]} \| \bpi^*_i - \e_k\|_2>0.05 ) \text{ for } K+1 \leq i \leq n
\end{split}
\end{equation}

The data we generate for the Bernoulli model uses stochastic block models as basis matrices. This is a natural choice since the stochastic block model is a common low-rank model for network structures, and networks are often modeled using the Bernoulli distribution. For both models, we set $d_1=d_2=d$ and $r_k = r$, and we remark that the basis matrices in this Bernoulli model will be symmetric. We regard basis matrix $\M_k^*$ as a prototypical network consisting of $d$ nodes belonging to $r \ll d$ communities. We set $\M_k^* = \Z_k \B_k \Z_k^\top$, where $\Z_k$ is a $d \times r$ matrix encoding node community membership with rows satisfying $(\Z_k)_{i,:} = \e_{s_{ik}}$ and $s_{ik} \sim \text{Unif}([r])$. We let $\xi$ be a $K$-dimensional vector of evenly-spaced real numbers from $(\text{maxp} \times \text{gapratio})$ to $\text{maxp}$ representing the strength of the signal of each of the $K$ basis matrices, and we set $\text{gapratio} = 0.6$ for all experiments. We let $\B_k = (\I_{r} + {\boldsymbol 1}{\boldsymbol 1}^\top)\xi_k/2$ be an $r \times r$ matrix representing edge probabilities between nodes of each of the $r$ communities. Finally, we generate $\X_i(j_1, j_2) \overset{iid}{\sim} \text{Bern}(\sum_{k=1}^K \bpi^*_{ik} \M^*_k(j_1, j_2))$.

The normal model is generated as follows. To obtain the basis matrix $\M^*_k$, generate two $d \times r$ matrices $\U_k$ and $\V_k$ independently, each being the matrix of the left singular vectors of a $d \times r$ matrix of i.i.d. $N(0,1)$ entries. $\underline{\bSigma}$ is set to be the diagonal matrix of the decreasing sequence of $r$ evenly-spaced values from $3$ to $1$. We let our unscaled basis matrix be $\tilde{\M}_k^* = \U_k \tilde{\bSigma} \V_k^\top$, and combine basis matrices to form $\tilde{\bTheta}$ as $\tilde{\bTheta}_{:,k} = \text{vec}(\tilde{\M}^*_k)$ for $k=1,...,K$. Next, we rescale all basis matrices by the same scalar to ensure the model satisfies the signal-to-noise conditions in Assumption \ref{ass:snr}. In particular, obtain minimum singular value $\sigma_K(\tilde{\bTheta})$ and rescale $\bSigma = (d/\sigma_K(\tilde{\bTheta}))\tilde{\bSigma}$ to get $\M^*_k = \U_k \bSigma \V_k^\top$. Another way to see the necessity of this scaling condition is that without it, the magnitude of the entries tends to $0$, leading to an extremely low SNR. Finally, we generate $\X_i(j_1, j_2) \overset{iid}{\sim} \text{Norm}(\sum_{k=1}^K \bpi^*_{ik} \M^*_k(j_1,j_2), \sigma^2)$. Section \ref{sec:ModelCheck} contains a Lemma and simulation that demonstrate that these generative models satisfy the SNR conditions of Assumption \ref{ass:snr}.

\subsection{Consistency Simulation Details}

For each of $n \in \{600, 800, 1000, 1200, 1400 \}$, we simulate a multimodal dataset with $M=5$ modalities, with $d=n/5$. We do this by generating membership scores according to (\ref{generate-pi}) and generating basis matrices independently for each of the $5$ modalities with $K=3$ and $r=3$, and then generating each subject's $5$ $d \times d$ matrices independently using the basis matrices from the corresponding modality. We vary the sparsity and noise values, with $\sigma \in \{ 0.5, 1.0, 1.5, 2.0 \}$ and $\text{maxp} \in \{ 0.15, 0.25, 0.35, 0.45 \}$. In order to compare Algorithm \ref{alg:Algorithm} with Algorithm \ref{alg:Algorithm-MM}, we fit Algorithm \ref{alg:Algorithm} on the dataset consisting only of matrices from the first modality of all subjects, while Algorithm \ref{alg:Algorithm-MM} is fit on the dataset consisting of matrices from all five modalities for all subjects. We run each algorithm for 10 iterations, using $\gamma = (1,...,1)$ and $r_{U,m} = r_{V,m}= r_U = r_V = r \times K = 9$. We compute membership MSE $\min_{\P} \frac{1}{nK}\|\hat{\bPi} - \bPi^* \P \|_F^2$ for both algorithms. For Algorithm \ref{alg:Algorithm}, we compute basis MSE on only the first modality $\min_{\P} \frac{1}{Kd^2}\|\hat{\bTheta}_1 - \bTheta_1^* \P \|_F^2$, while for Algorithm \ref{alg:Algorithm-MM}, we average basis MSE over all modalities $\min_{\P} \frac{1}{Kd^2M}\sum_{m=1}^M \|\hat{\bTheta}_m - \bTheta_m^* \P \|_F^2$, with all errors minimized over all permutation matrices $\P$. For each $n$, noise/sparsity value, and error distribution, we average the results over $30$ simulations. The results are displayed in Figure \ref{fig:Consistency}. Both algorithms are consistent. For each (model, error-type) pair, the unimodal and multimodal results have been plotted on the same scale to facilitate comparison. 

\subsection{Generative models satisfy SNR Conditions}\label{sec:ModelCheck}

Figure \ref{fig:ModelCheck} provides evidence that the generative models used in our simulation studies and specified in full detail in section \ref{sec:GenerativeModels} satisfy Assumption \ref{ass:snr} when $K$, $\sigma$, and $r$ are fixed and $d=n/5$. The first two columns of this figure indicate that $\sigma_K(\bPi^*)$ and $\sigma_1(\bPi^*)$ scale like $\sqrt{n}$, demonstrating that (\ref{membership-spread}) holds. The next two columns indicate $\sigma_K(\bTheta^*)$ and $\sigma_1(\bTheta^*)$ scale like $\sqrt{p}$, demonstrating that (\ref{basis-spread}) holds when $s_\star = \sqrt{p}$. The last column indicates that $\|\bTheta^*\|_{2,\infty}$ is $O(1)$ for the Bernoulli model and $O(\log(p))$ for the normal model. Lemma \ref{lem:verify-gen-model} demonstrates how these results altogether show that condition (\ref{snr-incoherence}) holds. Lastly, condition (\ref{snr-sub-Gaussian}) holds trivially. Hence, the generative models chosen for our simulations satisfy Assumptions \ref{ass:noise}, \ref{ass:pure-subjects}, \ref{ass:snr}, and \ref{ass:separation}. 

\begin{figure}[htbp]
    \centering
    \begin{subfigure}[t]{0.8\linewidth}
        \centering
        \includegraphics[width=\linewidth]{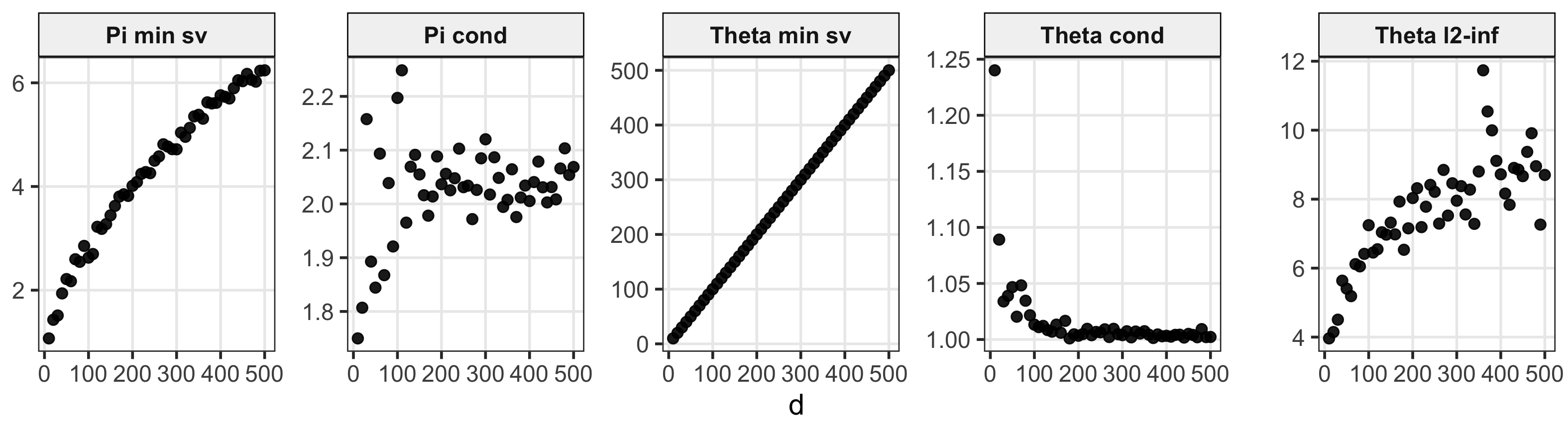}
        \subcaption{Normal}
    \end{subfigure}
    \begin{subfigure}[t]{0.8\linewidth}
        \centering
        \includegraphics[width=\linewidth]{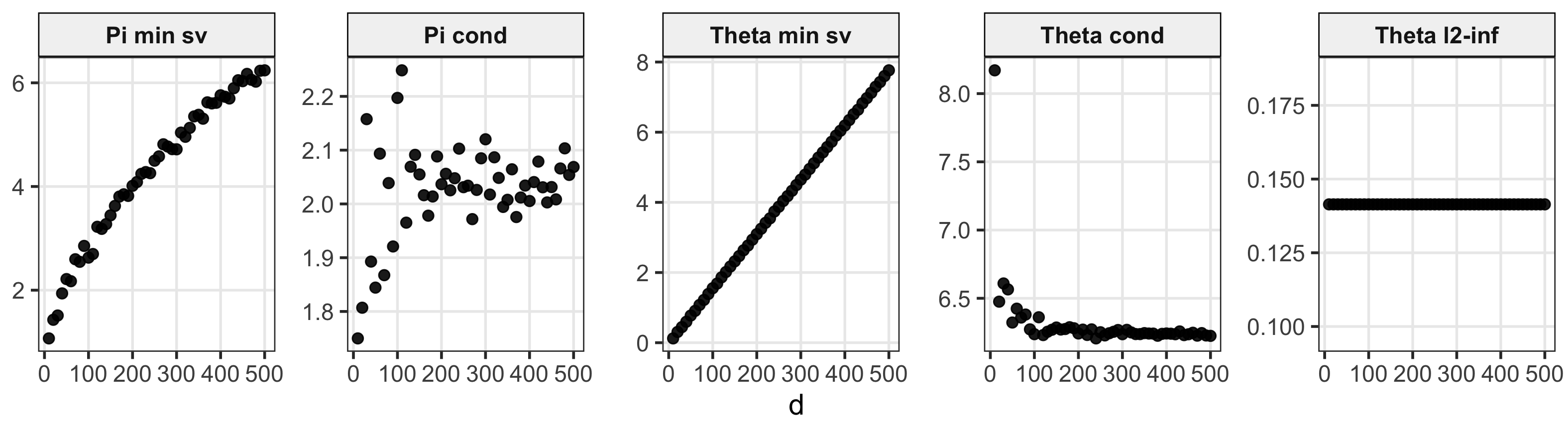}
        \subcaption{Bernoulli}
    \end{subfigure}

    \caption{Verification that our simulated data satisfies Assumption \ref{ass:snr}. Minimum singular value $\sigma_K(\bPi^*)$ and condition number $\kappa(\bPi^*)$ of $\bPi^*$, minimum singular value $\sigma_K(\bTheta^*)$ and condition number $\kappa(\bTheta^*)$ of $\bTheta^*$, and l2-infinity norm $\|\bTheta^*\|_{2,\infty}$ of $\bTheta^*$ for many values of $d$, where $d=n/5$.}
    \label{fig:ModelCheck}
\end{figure}

\begin{lemma}\label{lem:verify-gen-model}
    Fix $K$, $r$, and let $d = n/5$. Provided $\sigma_K(\bPi^*) \gtrsim \sqrt{n}$, $\kappa(\bPi^*) \lesssim 1$, $\sigma_K(\bTheta^*) \gtrsim \sqrt{p}$, and $\|\bTheta^*\|_{2,\infty} \lesssim \sqrt{p}\log^{-11}(p)$, condition (\ref{snr-incoherence}) holds. 
\end{lemma}

\subsection{Rate Verification for Algorithm \ref{alg:Algorithm}}

This simulation empirically verifies that Algorithm \ref{alg:Algorithm} attains the same estimation error rates obtained in Theorem \ref{thm:decouple-rate} for the two-step estimator \eqref{eq:estimator-twostep}. Using the same generative models described in section \ref{sec:GenerativeModels}, we run two sets of simulations. In the first, we vary $d\in\{ 100, 120, ..., 300 \}$ while $n=800$ and $r=3$ are fixed. In the second, we vary $n \in \{ 500, 550, ..., 1000 \}$ while $d=200$ and $r=3$ are fixed. We run each simulation $100$ times and average MSE over all simulations. Results are displayed in Figure \ref{fig:RateVerification}, with membership error plotted vs $1/d$ and basis error plotted vs $1/n$. Based on Theorem \ref{thm:decouple-rate}, the error of the two-step estimator \eqref{eq:estimator-twostep} is expected to scale like $O(1/d)$ when $r$ and $n$ are fixed, and like $O(1+1/n)$ when $d$ and $r$ are fixed, modulo the slowly-growing logarithmic term $\log(n \vee s_\star)$ which is disregarded for the purpose of this simulation. The plots indicate that Algorithm \ref{alg:Algorithm} attains these same rates. We note that in all cases the rates are sharp except for the case of membership estimation as a function of $d$, in which case the rate holds but is not sharp. A potential explanation of this phenomenon may lie in the difference between how the two-step estimator and Algorithm \ref{alg:Algorithm} estimate the pure subjects. In the regime where $d$ grows but $n$ is fixed, condition (\ref{snr-incoherence}) of Assumption \ref{ass:snr} no longer holds, and consequently, Proposition \ref{prop:est-Shat} no longer holds. Accordingly, the estimated pure subjects used by the two-step estimator are no longer accurate with the same high probability. In contrast, Algorithm \ref{alg:Algorithm} iteratively re-estimates the pure-subjects, which may facilitate improved estimation of the pure subjects and consequently a faster rate of estimation of membership parameters.

\begin{figure}[htbp]
    \centering
    \begin{subfigure}[t]{0.9\linewidth}
        \centering
        \includegraphics[width=\linewidth]{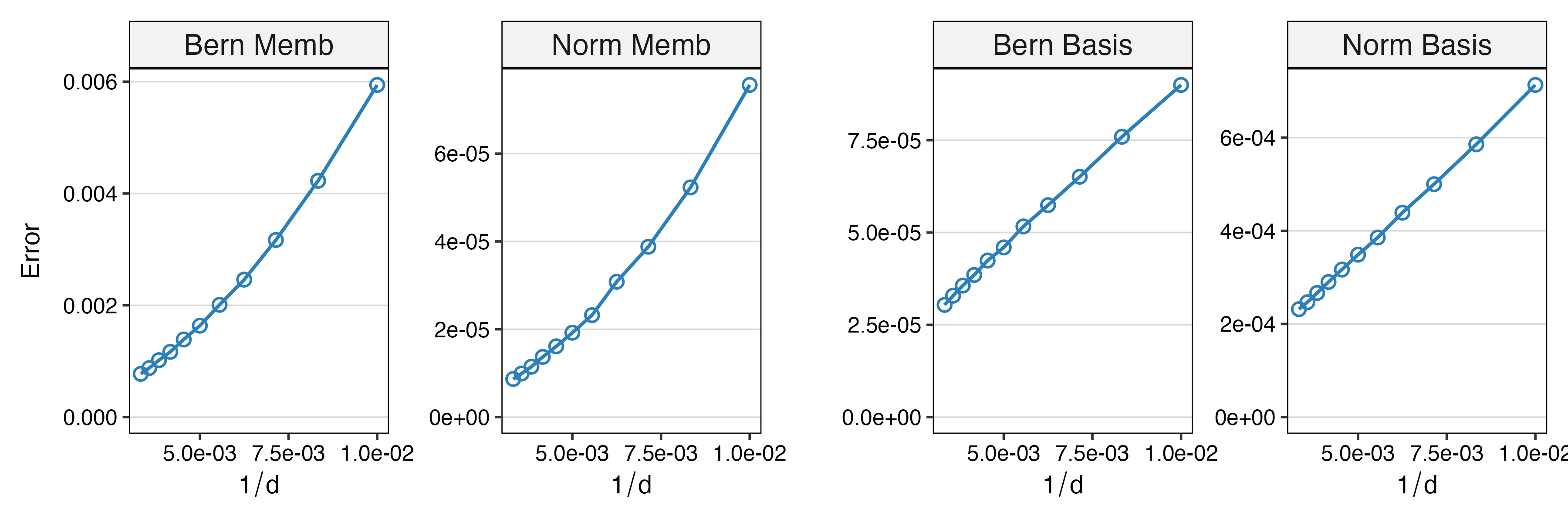}
        \subcaption{Rate for $d$}
    \end{subfigure}
    \begin{subfigure}[t]{0.9\linewidth}
        \centering
        \includegraphics[width=\linewidth]{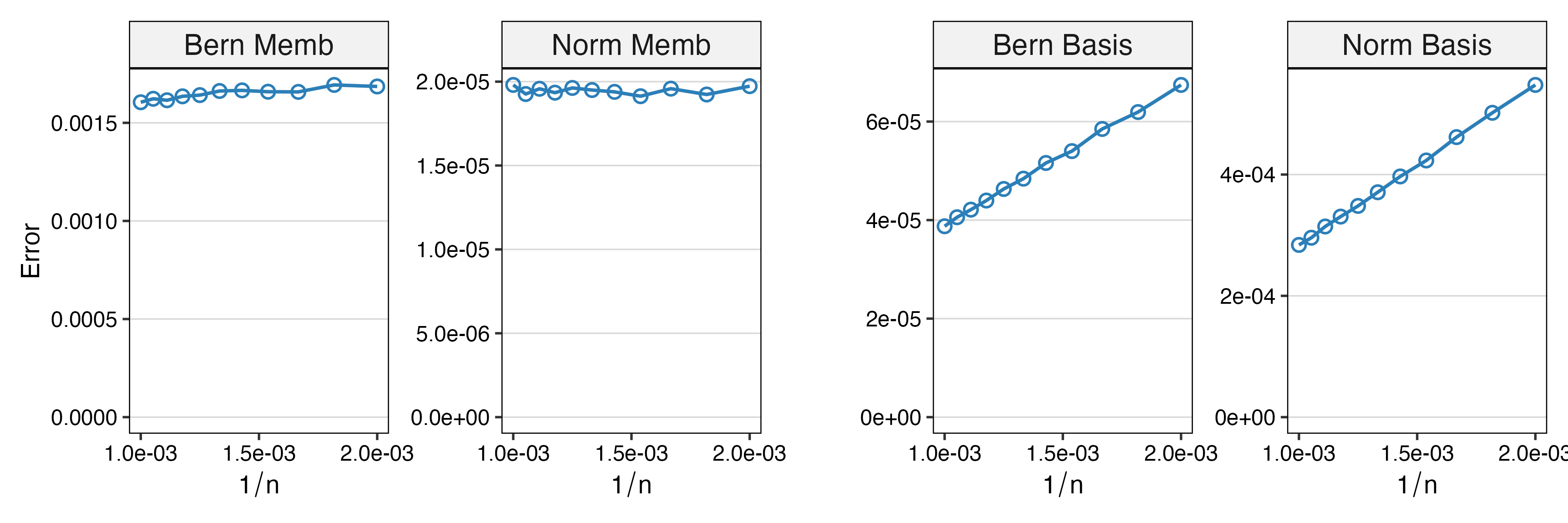}
        \subcaption{Rate for $n$}
    \end{subfigure}

    \caption{Empirical verification that Algorithm \ref{alg:Algorithm} attains the rates obtained in Theorem \ref{thm:decouple-rate}}
    \label{fig:RateVerification}
\end{figure}

\subsection{Simulated Neuroscience Application}

\begin{figure}[htbp]
    \centering
    \begin{subfigure}[t]{0.9\linewidth}
        \centering
        \includegraphics[width=\linewidth]{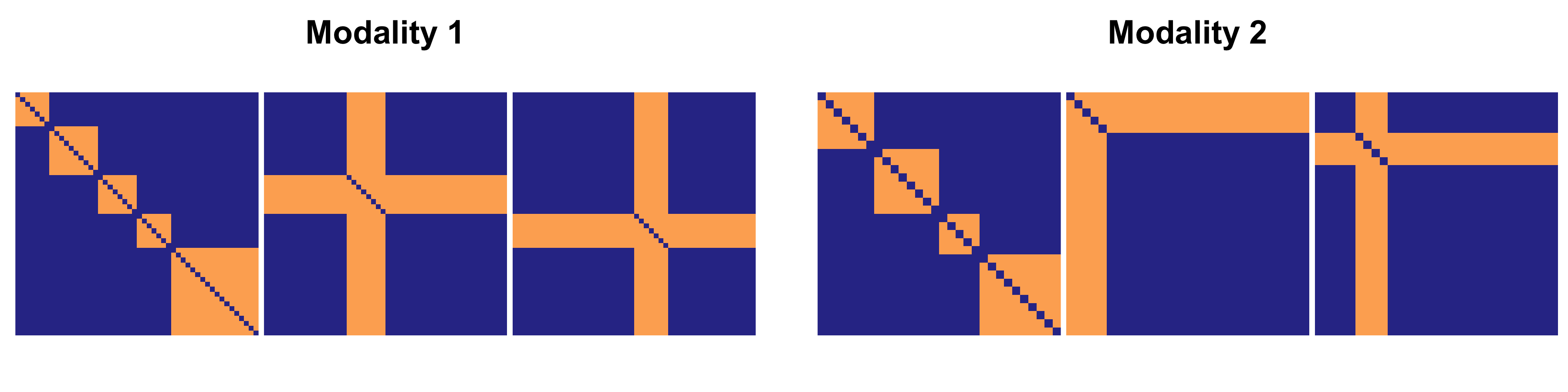}
        \subcaption{Signals}
    \end{subfigure}
    \begin{subfigure}[t]{0.9\linewidth}
        \centering
        \includegraphics[width=\linewidth]{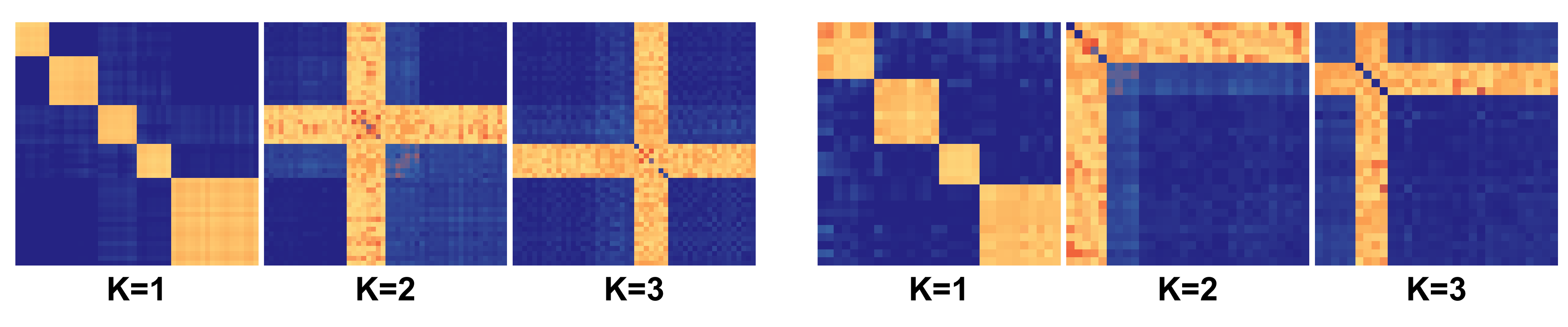}
        \subcaption{Estimates}
    \end{subfigure}

    \caption{Simulated application of multimodal algorithm to a bimodal dataset where each modality has different dimension. Estimated basis matrices accurately estimate the signal basis matrices used to simulate the data. Orange is 1 and dark blue is 0.}
    \label{fig:NeuroSim}
\end{figure}

We provide a visual demonstration of Algorithm \ref{alg:Algorithm-MM} on a simulated bimodal dataset designed to mimic a neuroscience application. Each individual in the simulated dataset supplies two matrices. The first matrix is a $50 \times 50$ matrix corresponding to the first modality, and the second matrix is a $30 \times 30$ matrix corresponding to the second modality. With $K=3$, we manually set the true ``signal" basis matrices to be the binary matrices pictured in the top row of Figure \ref{fig:NeuroSim}(a), with the three on the left belonging to the first modality, and the three on the right belonging to the second modality. Visually, orange corresponds to $1$, and dark blue corresponds to $0$. We generate the membership scores for $n=200$ subjects as Dirichlet(3,2,1.5), representing uneven mixed membership among the extreme latent profiles with no strict enforcement of the separation condition or pure subjects condition. We apply Algorithm \ref{alg:Algorithm-MM} with $r_{U,1}=r_{U,2}=r_{V,1}=r_{V,2}=10$, $r=5$, $K=3$, and $\gamma_1=\gamma_2=1$. The estimated basis matrices are pictured in the bottom row of Figure \ref{fig:NeuroSim} after permutation of extreme latent profiles to visually match with their corresponding signals. Despite the unevenness and the lack of pure subjects and lack of separation of the true membership scores, the basis matrix estimation is accurate.

\section{Additional Data Analysis Details}\label{sec:AdditionalData}

\subsection{Preprocessing Details}\label{sec:PreProcessing}
Preliminary preprocessing steps are provided in \citep{VanEssen13, Glasser2013MinimalPreprocessing}. After this, we applied the same additional preprocessing steps used in \cite{Sripada19} to obtain time series of 264 ``regions of interest" (ROIs) within the brain, these regions being initially proposed in \cite{Power11}. For each (subject, modality) pair, we obtained the absolute value of correlation between all $\binom{264}{2}$ pairs of ROIs across the corresponding time series, providing a measurement of ``functional connectivity" between these regions. Then, we collected these data into a binary matrix by assigning entry $(j_1,j_2)$ of the $264 \times 264$ matrix $\X_{i,m}$ to be $1$ if the absolute value of the correlation between ROI $j_1$ and $j_2$ in the time series of subject $i$ in modality $m$ is greater than threshold 0.2. If $j_1 = j_2$, the entry was assigned to be $0$. Hence, the processed data for each subject and modality comes in the form of an adjacency matrix in which a value of 1 indicates the presence of functional connectivity between two regions. The HCP also includes several auxiliary phenotypic covariates measuring performance on various cognitive tests, which we use for downstream interpretation of our estimated latent structures. After removing subjects with erroneous data according to the preprocessing steps in \cite{Sripada19} and additionally removing subjects with missing modalities and a covariate with few subject responses, we are left with a dataset of $n=913$ subjects and $W=24$ covariates.

\subsection{Selection of Parameters}\label{sec:ModelSelection}

In order to apply our method, we must choose parameters $K, ((r_{k,m})_{k=1}^K)_{m=1}^M$. To select $K$, we perform a tensor-informed scree plot method. To justify our scree plot method, we introduce the following notation. Firstly, we denote by $\underline{\bcalX}$ the $d_1 \times d_2 \times n \times M$ tensor that satisfies $\underline{\bcalX}(:,:,:,m) = \bcalX_m$, where $\bcalX_m$ is the $d_1 \times d_2 \times n$ tensor that satisfies $\bcalX_m(:,:,i) = \X_{i,m}$. Denote by $\underline{\bcalM}^*$ the $d_1 \times d_2 \times K \times M$ tensor satisfying $\underline{\bcalM}^*(:,:,k,m) = \M_{k,m}^*$, and denote by $\underline{\Y}$ the $n \times (d_1d_2M)$ matrix given by the mode-3 unfolding of $\underline{\bcalX}$. We analyze the singular values of $\underline{\Y}$ since the rank of its expectation $\EE[\underline{\Y}] = \bPi^* M_3(\underline{\bcalM}^*) = \underline{\U}^* \underline{\bSigma}^* \underline{\V}^{*\top}$ is $K$. To do this, we need to estimate $\underline{\bSigma}^*$, and we can account for the tensor structure of our dataset in this estimation by using the following Lemma. \begin{lemma}\label{lem:model-selection}
    $\text{rank}(M_j(\EE[\underline{\bcalX}])) = \underline{r}_j$ for $j=1,2,3,4$, where $\underline{r}_1, \underline{r}_2 \leq \sum_{m=1}^M \sum_{k=1}^K r_{k,m}$, $\underline{r}_3 \leq K$, $\underline{r}_4 \leq M$.
\end{lemma}
\noindent $\EE[\underline{\bcalX}]$ can be written as 
% \begin{align*}
    $\bcalS^* \times_1 \underline{\W}_1^* \times_2 \underline{\W}_2^* \times_3 \underline{\W}_3^* \times_4 \underline{\W}_4^*,$ 
% \end{align*}
where $\underline{\W}_j^*$ have orthogonal columns and satisfy $\text{rank}(\underline{\W}_j^*) = \underline{r}_j$ for $j=1,2,3,4$. We know $\EE[\underline{\Y}] = \underline{\U}^* \underline{\bSigma}^* \underline{\V}^{*\top} = \underline{\W}_3^* M_3(\bcalS^*) [\underline{\W}_1^* \otimes \underline{\W} _2^* \otimes \underline{\W}_4^*]^T$. Therefore, we estimate $\underline{\bSigma}^*$ with $\hat{\underline{\bSigma}}$, the diagonal matrix consisting of the singular values of $\hat{\underline{\W}}_3\hat{\underline{\W}}^{\top}_3 \underline{\Y} [\hat{\underline{\W}}_1\hat{\underline{\W}}_1^\top \otimes \hat{\underline{\W}}_2 \hat{\underline{\W}}_2^\top \otimes \hat{\underline{\W}}_4 \hat{\underline{\W}}_4^\top]$, where $(\hat{\underline{\W}}_1, \hat{\underline{\W}}_2, \hat{\underline{\W}}_3, \hat{\underline{\W}}_4)$ are obtained from HOOI on $\underline{\bcalX}$ with ranks $(\underline{r}_1, \underline{r}_2, \underline{r}_3, \underline{r}_4)$. We used $(\underline{r}_1, \underline{r}_2, \underline{r}_3, \underline{r}_4) = (14, 14, 14, 8)$ since these were the elbows of the singular values of the mode-wise unfoldings of $\underline{\bcalX}$. The resulting estimated singular values are presented in Figure \ref{fig:ModelSelection}, with the large first singular value omitted to facilitate visibility of the remaining singular values. There is some ambiguity as to whether $K=3$ or $K=6$ is the elbow. We choose $K=3$ based on the observation that this value leads to a more interpretable fit. Similar choices have become common in the topic modeling literature (\cite{KeWang24}). 

To select $((r_{k,m})_{k=1}^K)_{m=1}^M$, we first make the simplifying assumption that $r_{k,m} = r$ for all $m, k$. We perform steps \ref{alg-mm:first-step} to \ref{alg-mm:renormalize-init} of Algorithm \ref{alg:Algorithm-MM} using $K=3$ and conservative preliminary estimates $r_{U,m}=r_{V,m}=14$ to obtain $\hat{\bPi}^{(0)}$. For each modality $m$, we obtained $\tilde{\bTheta}_m^{(1)}$ as in step \ref{alg-mm:theta-tilde}, and then obtained the singular values of the matricization of each column of $\tilde{\bTheta}_m^{(1)}$, denoted by $\sigma_{r,k,m}$. Denote by $\sigma_r$ the sum $\sum_{m=1}^M \sum_{k=1}^K \sigma_{r,k,m}$. We looked for an elbow in the plot $(\sigma_r)_{r}$, since if all matrices $\M_{k,m}^*$ are rank $r^*$, then we anticipate the plot of $(\sigma_r)_r$ to have an elbow at $r^*$ provided appropriate signal-to-noise conditions hold and the initial estimation $\tilde{\bTheta}_m^{(1)}$ is sufficiently accurate. Using this method, we identify an elbow at $r=3$ as shown in Figure \ref{fig:ModelSelection}. Hence, we use $K=3, r=3$ when fitting our algorithm. We also use $r_{U,m} = r_{V,m} = r \times K = 9$, since by \Cref{lem:rank-size}, this is a conservative upper bound. 

\begin{figure}[htbp]
    \centering
    \begin{subfigure}[t]{0.3\linewidth}
        \centering
        \includegraphics[width=\linewidth]{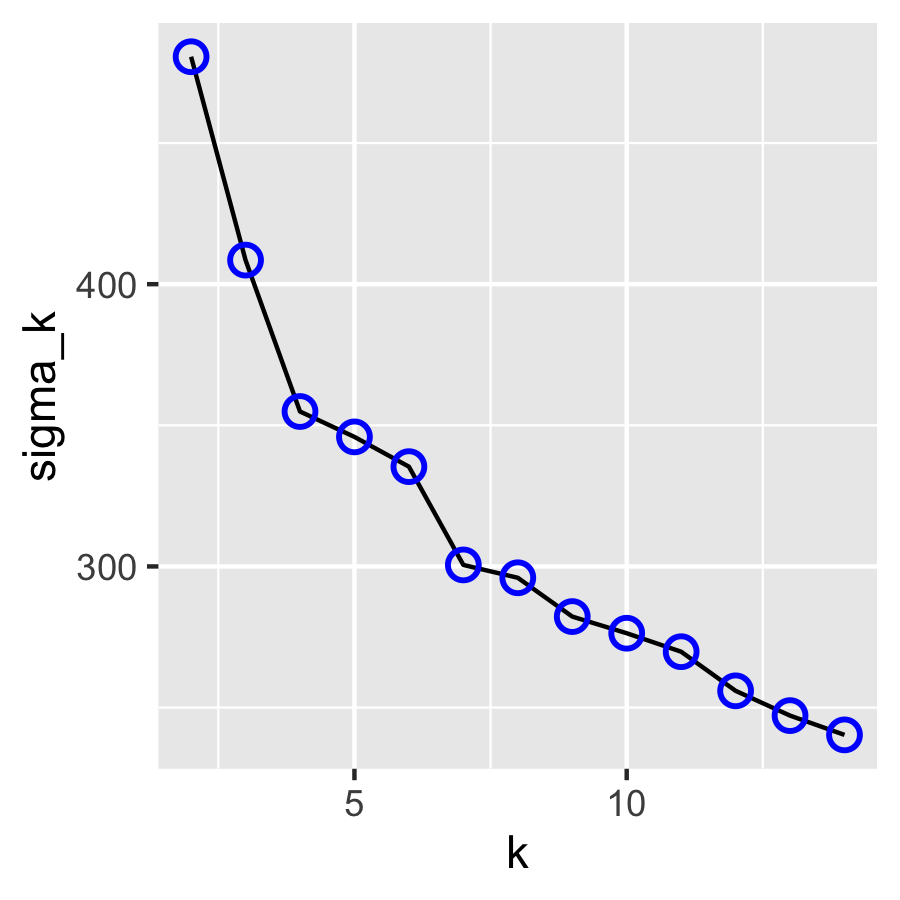}
        \subcaption{Selection of $K$}
    \end{subfigure}
    \qquad
    \begin{subfigure}[t]{0.3\linewidth}
        \centering
        \includegraphics[width=\linewidth]{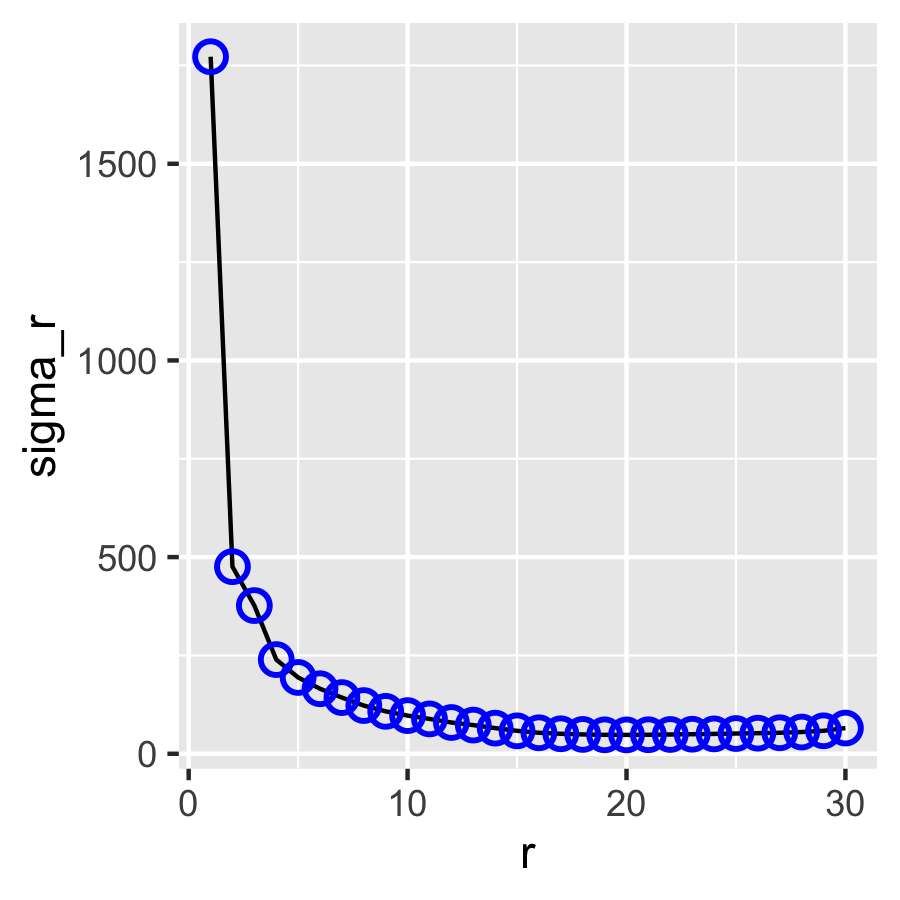}
        \subcaption{Selection of $r$}
    \end{subfigure}

    \caption{Scree plot methods for selection of $K$ and $r$. The first singular value is omitted from plot (a) in order to facilitate observation of the remaining smaller singular values.}
    \label{fig:ModelSelection}
\end{figure}

\subsection{Obtaining basis matrices from Tensor Mixed Membership Blockmodel}

The Tensor Mixed Membership Blockmodel models the $n \times d_1 \times d_2$ data tensor $\bcalX$ whose $i$-th slice is $\bcalX(:,:,i) = \X_i$ via Tucker decomposition as \begin{align}\label{eq:TMMBM}
\begin{split}
    \EE \bcalX &= \bcalC \times_1 \U_1 \times_2 \U_2 \times_3 \U_3 \\
    &= \bcalS \times_1 \bPi_1 \times_2 \bPi_2 \times_3 \bPi_3,
\end{split}
\end{align}
where $\bPi_l = \U_l (\U_l)_{S_l, :}^{-1}$ for $l=1,2,3$ \citep{AZ25}. Given that in our setting $\bcalX$ is a dataset of matrix-valued observations, a natural interpretation of tensor $\bcalM := \bcalS \times_1 \bPi_1 \times_2 \bPi_3$ is the tensor whose $k$-th slice $\bcalM(:,:,k)$ is the k-th basis matrix; indeed, the above model implicitly expresses the $i$-th subject's data as $\sum_{k=1}^K \bPi_3(i,k) \bcalM(:,:,k)$. Hence, our estimate of the $k$-th basis matrix under TMMBM is given by $\hat{\bcalM}(:,:,k)$, where $\hat{\bcalM} := \hat{\bcalS} \times_1 \hat{\bPi}_1 \times_2 \hat{\bPi}_2$ and $\hat{\bcalS} := \hat{\bcalC} \times_1 (\hat{\U}_1)_{\hat{S}_1, :} \times_2 (\hat{\U}_2)_{\hat{S}_2, :} \times_3 (\hat{\U}_3)_{\hat{S}_3, :}$, and $\hat{\bPi}_l = \hat{\U}_l (\hat{\U}_l)_{\hat{S}_l,:}^{-1}$. Here, $(\hat{\bcalC}, \hat{\U}_1, \hat{\U}_2, \hat{\U}_3)$ are the parameters estimated by HOOI on $\bcalX$, and $\hat{S}_l$ are the pure subject indices estimated by the Successive Projection Algorithm on the rows of $\hat{\U}_l$.

\subsection{Measuring prediction error}

Here, we present in full detail our procedure for obtaining out-of-sample prediction error on all $n=913$ observations for prediction of phenotypic covariate using each method's estimated embedding as regressors. We randomly divide the $913$ observations into $B=10$ folds. For each fold, we regard that fold as the ``test set" and the remaining folds as the ``training set." We fit each method on the training set, fit an OLS model regressing covariate on embedding using the training set, project the observations in the test set onto the membership/score space of the training set, and then calculate squared error of the OLS model for each observation in the test set. In order to be precise, we must specify how we project the observations in the test set onto the membership/score space of the training set. The manner in which we do this projection is necessarily different for each of our considered methods. For LrLloyd, we assign each observation in the test set to the class corresponding to the basis matrix that is closest to the observation in Frobenius norm distance. Similarly, for uLrMMM, we assign observation $i$ in the test set to the membership vector given by $\text{argmin}_{\bpi \in \Delta_k} || \X_i - \sum_{k=1}^K \bpi_{ik} \hat{\M}_k||_F^2$, where $(\hat{\M}_k)$ are the basis matrices estimated on the training set. And for mLrMMM, we assign observation $i$ in the test set to the membership vector given by $\text{argmin}_{\bpi \in \Delta_k} \sum_m \lambda_m || \X_{i,m} - \sum_{k=1}^K \bpi_{ik} \hat{\M}_{k,m}||_F^2$, where $(\hat{\M}_{k,m})$ are the basis matrices estimated on the training set.

Our projection method for TMMBM differs substantially from the previous methods due to differences in TMMBM's modeling structure. In contrast to the previously mentioned methods, TMMBM estimates its membership embedding not by finding membership vectors that minimize $|| \X_i - \sum_{k=1}^K \bpi_{ik} \hat{\M}_k||_F^2$ for basis matrices $(\hat{\M}_k)$, but rather using the Successive Projection Algorithm (SPA) on the singular subspaces estimated by Higher Order Orthogonal Iteration (HOOI) on the data tensor $\bcalX$. We incorporate this subspace-based estimation into our projection procedure as follows. Supposing we have $n_{tr}$ observations in the training set and $n_{te}$ in the test set, denote by $\bcalX^{(tr)}$ the $n_{tr} \times d_1 \times d_2$ tensor collecting the training data, and denote by $\bcalX^{(te)}$ the $n_{te} \times d_1 \times d_2$ tensor collecting the test data. According to the TMMBM model structure expressed in \eqref{eq:TMMBM}, the expectation of the training dataset is estimated as $M_3(\bcalX^{(tr)}) \approx \U^{(tr)}_3 M_3(\bcalC^{(tr)}) [\U^{(tr)}_1 \otimes \U^{(tr)}_2]^\top$, where $\bcalC^{(tr)}, \U^{(tr)}_1, \U^{(tr)}_2, \U^{(tr)}_3$ are the parameters estimated by HOOI on $\bcalX^{(tr)}$. Noting that the mode-3 subspace $\U^{(tr)}_3$ provides an embedding of the observations in the training set, we extend this embedding to the observations in the test set via $\U^{(te)}_3 = \text{argmin}_{\U \in \R^{n_{te} \times K}} || M_3(\bcalX^{(te)}) - \U M_3( \bcalC^{(tr)} ) [\U^{(tr)}_1 \otimes \U^{(tr)}_3]^\top ||_F^2$. We then assign to the $i$-th subject in the test set the membership vector given by the projection of the $i$-th row of the following matrix onto $\bDelta_K$: $\U^{(te)}_3 [(\U^{(tr)}_3)(S^{(tr)}_3, :)]^{-1}$.

As for SS-TPCA, starting with $k=1$, for each observation $i$ in the test set, we obtain $\hat{s}_{i,1} = \text{argmin}_{s \in \RR} || \X_i - s \hat{\M}_1||_F^2$, where $\hat{\M}_1$ is the first factor loading matrix estimated on the training set. We then obtain the residual $\mathring{\X}_i := \X_i - \hat{s}_{i,1} \hat{\M}_1$ and estimate $\hat{s}_{i,2} = \text{argmin}_{s \in \RR} || \mathring{\X}_i - s \hat{\M}_2||_F^2$, where $\hat{\M}_2$ is the second factor loading matrix estimated on the training set. Since we are using $K=2$ for SS-TPCA and gJisstPCA as remarked in Section \ref{sec:ModelComparison}, we are finished; $\hat{s}_i = (\hat{s}_{i,1}, \hat{s}_{i,2})$ are the estimated scores of subject $i$. 

As for gJisstPCA, we follow a procedure similar to SS-TPCA. Starting with $k=1$, for each observation $i$ in the test set, we obtain $\hat{s}_{i,1} = \text{argmin}_{s \in \RR} \sum_m \lambda_m || \X_{i,m} - s \hat{\M}_{1,m}||_F^2$, where $(\hat{\M}_{1,m})$ are the first factor loading matrices for the different modalities estimated on the training set. We then obtain the residual $\mathring{\X}_{i,m} := \X_{i,m} - \hat{s}_{i,1} \hat{\M}_{1,m}$ and estimate $\hat{s}_{i,2} = \text{argmin}_{s \in \RR} || \mathring{\X}_{i,m} - s \hat{\M}_{2,m}||_F^2$, where $(\hat{\M}_{2,m})$ is the second factor loading matrix estimated on the training set. $\hat{s}_i = (\hat{s}_{i,1}, \hat{s}_{i,2})$ are the estimated scores of subject $i$. 

Our estimate of relative mean squared error is defined as follows. Let $c_{i,j}$ denote the true value of the $j$-th covariate for subject $i$, and for each method, let $\hat{c}^{(\text{method})}_{i,j}$ denote the prediction of the covariate using the OLS model that takes the embedding parameters as its regressors. The relative mean squared error is defined as \begin{align*}
    \text{RelMSE}(j, \text{method}) = \frac{\sum_{i=1}^n (\hat{c}^{(\text{method})}_{i,j} - c_{i,j})^2}{\sum_{i=1}^n (\hat{c}^{(\text{TMMBM})}_{i,j} - c_{i,j})^2}.
\end{align*}

To obtain bootstrap confidence intervals for this relative mean squared error, for $2000$ simulations we sample $913$ subjects with replacement and calculate $\text{RelErr}(j, \text{method})$ using these sampled subjects. We then use the $0.025$- and $0.975$-quantiles of the estimated RelMSE($j$, method) values as the lower and upper bounds of our confidence interval, respectively. 

\subsection{Full results}\label{sec:HCPFullResults}

The $K=3$ estimated basis matrices for all $M=8$ modalities estimated by mLrMMM are pictured in Figures \ref{fig:BasisFirst} and \ref{fig:BasisLast}. The kernel-smoothed ternary plots for all 24 covariates are pictured in Figure \ref{fig:ternary_phenotype}. The complete table of regression results is pictured in Table \ref{tab:findings_table}. Brief descriptions of the phenotypic covariates are provided in Table \ref{tab:hcp_phenotypes_brief}.

\begin{figure}[htbp]
    \centering
    \begin{subfigure}[c]{0.75\textwidth}
      \centering
      \includegraphics[width=\linewidth]{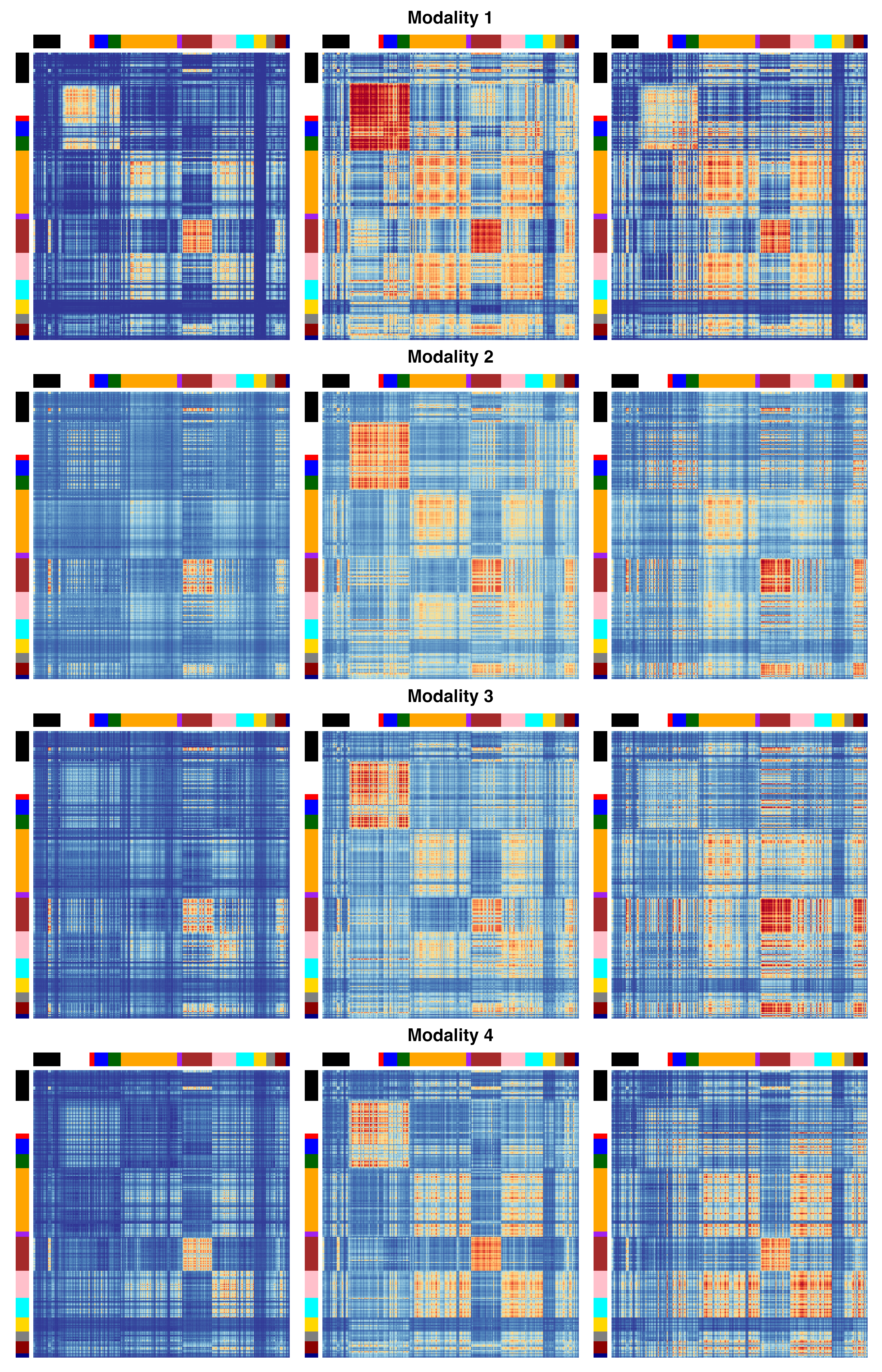}
    \end{subfigure}%
    \begin{subfigure}[c]{0.14\textwidth}
      \centering
      \includegraphics[width=\linewidth]{Figures/Findings/module_key.png}
    \end{subfigure}
    
    \caption{Three basis matrices of HCP dataset for first four modalities: resting, emotion, gambling, and language tasks. Values represent a measurement of functional connectivity, and warmer color indicates higher estimated value. Results are plotted alongside functional modules from domain literature, classified as follows: UN=Uncertain, SMH=Sensory/Somatomotor Hand, SMM=Sensor/Somatomotor Mouth, CO=Cingulo-opercular Task Control, AD=Auditory, DM=Default Mode, MR=Memory Retrieval, VS=Visual, FP=Fronto-parietal Task control, SA=Salience, SC=Subcortical, VAT=Ventral Attention, DAT=Dorsal Attention, CRB=Cerebellar.}
    \label{fig:BasisFirst}
\end{figure}

\begin{figure}[htbp]
    \centering
    \begin{subfigure}[c]{0.75\textwidth}
      \centering
      \includegraphics[width=\linewidth]{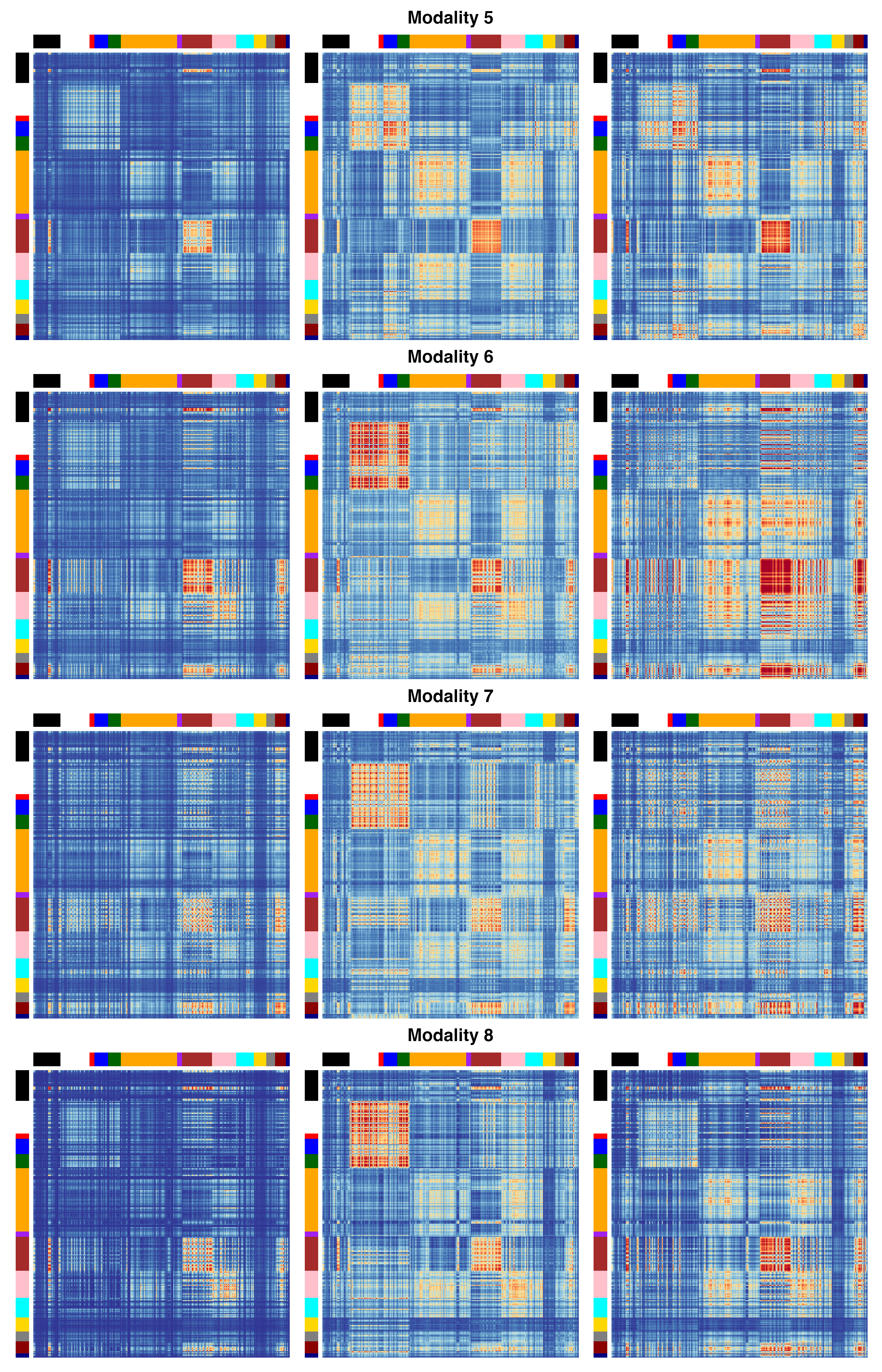}
    \end{subfigure}%
    \begin{subfigure}[c]{0.14\textwidth}
      \centering
      \includegraphics[width=\linewidth]{Figures/Findings/module_key.png}
    \end{subfigure}
    
    \caption{Three basis matrices of HCP dataset for last four modalities: motor, relational, social, and working memory tasks. Values represent a measurement of functional connectivity, and warmer color indicates higher estimated value. Results are plotted alongside functional modules from domain literature, classified as follows: UN=Uncertain, SMH=Sensory/Somatomotor Hand, SMM=Sensor/Somatomotor Mouth, CO=Cingulo-opercular Task Control, AD=Auditory, DM=Default Mode, MR=Memory Retrieval, VS=Visual, FP=Fronto-parietal Task control, SA=Salience, SC=Subcortical, VAT=Ventral Attention, DAT=Dorsal Attention, CRB=Cerebellar.}
    \label{fig:BasisLast}
\end{figure}

\begin{figure}
    \centering
    \includegraphics[width=0.75\linewidth]{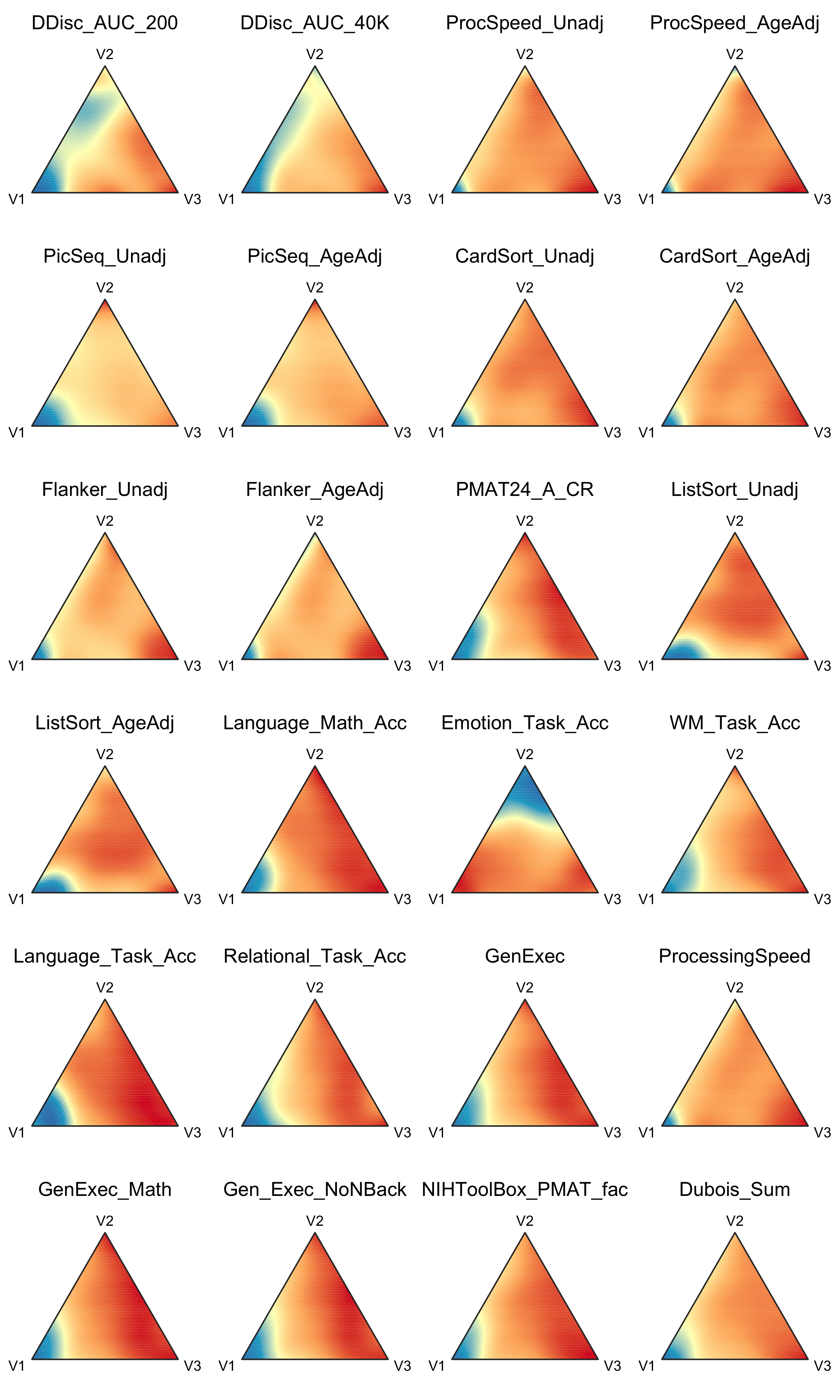}
    \caption{Estimated membership matrix from HCP Dataset colored by kernel-smoothed unseen phenotypic covariates. V1, V2, and V3 represent latent extreme profiles 1, 2, and 3 respectively. Warmer color represents higher value. Higher membership in latent extreme profile 3 is strongly associated with improved performance on a variety of cognitive tests, including tests of working memory, language, and mathematics.}
    \label{fig:ternary_phenotype}
\end{figure}

\begin{table}[h!]
\centering
\begin{tabular}{lrrlrrl}
   & \multicolumn{3}{c}{Membership vector 2} &  \multicolumn{3}{c}{Membership vector 3} \\ 
 \hline
Covariate & coef & fdr\_pval & sig & coef & fdr\_pval & sig \\ 
  \hline
DDisc\_AUC\_200 & 0.01 & 0.835 &  & 0.17 & 0.000 & x \\ 
  DDisc\_AUC\_40K & 0.05 & 0.501 &  & 0.31 & 0.000 & x \\ 
  ProcSpeed\_Unadj & 3.55 & 0.374 &  & 9.06 & 0.005 & x \\ 
  ProcSpeed\_AgeAdj & -0.40 & 0.930 &  & 10.15 & 0.016 & x \\ 
  PicSeq\_Unadj & 5.08 & 0.127 &  & 6.77 & 0.017 & x \\ 
  PicSeq\_AgeAdj & 3.39 & 0.435 &  & 7.96 & 0.023 & x \\ 
  CardSort\_Unadj & 5.48 & 0.030 & x & 7.05 & 0.001 & x \\ 
  CardSort\_AgeAdj & 1.93 & 0.448 &  & 5.90 & 0.005 & x \\ 
  Flanker\_Unadj & 2.89 & 0.264 &  & 5.27 & 0.015 & x \\ 
  Flanker\_AgeAdj & 0.20 & 0.930 &  & 4.51 & 0.033 & x \\ 
  PMAT24\_A\_CR & 4.44 & 0.000 & x & 4.67 & 0.000 & x \\ 
  ListSort\_Unadj & 6.87 & 0.015 & x & 4.55 & 0.061 &  \\ 
  ListSort\_AgeAdj & 4.60 & 0.168 &  & 4.39 & 0.127 &  \\ 
  Language\_Math\_Acc & 7.60 & 0.002 & x & 8.56 & 0.000 & x \\ 
  Emotion\_Task\_Acc & -3.64 & 0.000 & x & -0.42 & 0.606 &  \\ 
  WM\_Task\_Acc & 7.35 & 0.001 & x & 11.23 & 0.000 & x \\ 
  Language\_Task\_Acc & 5.74 & 0.001 & x & 7.66 & 0.000 & x \\ 
  Relational\_Task\_Acc & 9.58 & 0.002 & x & 13.03 & 0.000 & x \\ 
  GenExec & 1.06 & 0.000 & x & 1.38 & 0.000 & x \\ 
  ProcessingSpeed & 0.08 & 0.751 &  & 0.65 & 0.002 & x \\ 
  GenExec\_Math & 1.06 & 0.000 & x & 1.24 & 0.000 & x \\ 
  Gen\_Exec\_NoNBack & 1.01 & 0.000 & x & 1.20 & 0.000 & x \\ 
  NIHToolBox\_PMAT\_fac & 0.73 & 0.001 & x & 1.16 & 0.000 & x \\ 
  Dubois\_Sum & 36.85 & 0.004 & x & 60.08 & 0.000 & x \\ 
   \hline
\end{tabular}
\caption{Regression results for each unseen phenotypic covariate vs membership matrix obtained from multimodal algorithm on HCP Data. Regression is done with intercept.} 
\label{tab:findings_table}
\end{table}

\begin{table}[t]
\centering
\caption{HCP phenotypic covariates (brief descriptions).}
\label{tab:hcp_phenotypes_brief}
\begin{tabular}{ll}
\toprule
Covariate & Brief meaning \\
\midrule
\texttt{DDisc\_AUC\_200}          & Delay discounting AUC (\$200) \\
\texttt{DDisc\_AUC\_40K}          & Delay discounting AUC (\$40K) \\
\texttt{ProcSpeed\_Unadj}         & Processing speed, unadjusted \\
\texttt{ProcSpeed\_AgeAdj}        & Processing speed, age-adjusted \\
\texttt{PicSeq\_Unadj}            & Picture sequence memory, unadjusted \\
\texttt{PicSeq\_AgeAdj}           & Picture sequence memory, age-adjusted \\
\texttt{CardSort\_Unadj}          & Card sort flexibility, unadjusted \\
\texttt{CardSort\_AgeAdj}         & Card sort flexibility, age-adjusted \\
\texttt{Flanker\_Unadj}           & Flanker inhibition/attention, unadjusted \\
\texttt{Flanker\_AgeAdj}          & Flanker inhibition/attention, age-adjusted \\
\texttt{PMAT24\_A\_CR}             & Penn matrices correct responses \\
\texttt{ListSort\_Unadj}          & List sorting working memory, unadjusted \\
\texttt{ListSort\_AgeAdj}         & List sorting working memory, age-adjusted \\
\texttt{Language\_Math\_Acc}       & Language task math accuracy \\
\texttt{Emotion\_Task\_Acc}        & Emotion task accuracy \\
\texttt{WM\_Task\_Acc}             & Working-memory task accuracy \\
\texttt{Language\_Task\_Acc}       & Language task overall accuracy \\
\texttt{Relational\_Task\_Acc}     & Relational reasoning task accuracy \\
\texttt{GenExec}                  & General executive factor score \\
\texttt{ProcessingSpeed}          & Processing speed factor score \\
\texttt{GenExec\_Math}            & Executive factor including math \\
\texttt{Gen\_Exec\_NoNBack}       & Executive factor without n-back \\
\texttt{NIHToolBox\_PMAT\_fac}    & NIH+PMAT cognition factor \\
\texttt{Dubois\_Sum}              & Dubois cognition sum score \\
\bottomrule
\end{tabular}
\end{table}